\documentclass[aps,pra,superscriptaddress,notitlepage, amsfonts,longbibliography, twocolumn]{revtex4-2}

\usepackage{qcircuit}
\usepackage[dvipdfmx]{graphicx}
\usepackage{amsmath,amssymb,amsthm,mathrsfs,amsfonts,dsfont}
\usepackage{subfigure, epsfig}
\usepackage{braket}
\usepackage{bm}
\usepackage{enumerate}
\usepackage{physics}
\newtheorem{thm}{Theorem}
\newtheorem{lem}{Lemma}

\newtheorem{cor}{Corollary}
\usepackage[colorlinks,linkcolor=blue,citecolor=blue]{hyperref}
\usepackage{stmaryrd}

\makeatletter
\renewcommand*{\numberline}[1]{\hb@xt@1em{#1\hfil}} 
\makeatother

\graphicspath{{./fig/}}

\begin{document}

\title{Symmetric channel verification for purifying noisy quantum channels}

\author{Kento Tsubouchi}
\email{tsubouchi@noneq.t.u-tokyo.ac.jp}
\affiliation{Department of Applied Physics, The University of \mbox{Tokyo, 7-3-1} Hongo, Bunkyo-ku, Tokyo 113-8656, Japan}

\author{Yosuke Mitsuhashi}
\affiliation{Department of Basic Science, The University of \mbox{Tokyo, 3-8-1} Komaba, Meguro-ku, Tokyo 153-8902, Japan}

\author{Ryuji Takagi}
\affiliation{Department of Basic Science, The University of \mbox{Tokyo, 3-8-1} Komaba, Meguro-ku, Tokyo 153-8902, Japan}

\author{Nobuyuki Yoshioka}
\email{ny.nobuyoshioka@gmail.com}
\affiliation{\mbox{International Center for Elementary Particle Physics, The University of Tokyo, 7-3-1 Hongo, Bunkyo-ku, Tokyo 113-0033, Japan}}

\begin{abstract}
Symmetry inherent in quantum states has been widely used to reduce the effect of noise in quantum error correction and a quantum error-mitigation technique known as symmetry verification.
However, these symmetry-based techniques exploit symmetry in quantum {\it states} rather than quantum {\it channels}, limiting their application to cases where the entire circuit shares the same symmetry.
In this work, we propose symmetric channel verification (SCV), a channel purification protocol that leverages the symmetry inherent in quantum channels.
By introducing different phases to each symmetric subspace and employing a quantum phase estimation-like circuit, SCV can detect and correct symmetry-breaking noise in quantum channels.
We further propose a hardware-efficient implementation of SCV at the virtual level, which requires only a single-qubit ancilla and is robust against the noise in the ancilla qubit.
Our protocol is applied to various Hamiltonian simulation circuits and phase estimation circuits, resulting in a significant reduction of errors.
Furthermore, in setups where only Clifford unitaries can be used for noise purification, which is relevant in the early fault-tolerant regime, we show that SCV under Pauli symmetry represents the optimal purification method.
\end{abstract}

\maketitle

\section{Introduction}
\label{sec_introduction}
Protecting quantum computers from the effects of noise is one of the central challenges in the field of quantum computation.
To achieve this goal, the symmetry inherent in quantum states has emerged as a powerful tool in reducing the effect of unwanted noise.
For example, quantum error correction~\cite{shor1995scheme, knill1996threshold, aharonov1997fault, lidar2013quantum, ofek2016extending, krinner2022realizing, google2023suppressing, sivak2023real, bluvstein2024logical, ai2024quantum} detects and corrects errors by encoding logical quantum states into a symmetric subspace of the entire Hilbert space of physical qubits.
Moreover, the quantum error mitigation~\cite{temme2017error, endo2021hybrid, cai2023quantum, kim2023evidence} technique known as symmetry verification~\cite{bonet2018low, mcardle2019error} leverages problem-specific symmetry to filter out the effects of symmetry-breaking noise.

Despite their effectiveness, these symmetry-based protocols are primarily designed to exploit symmetry in quantum \textit{states} rather than quantum \textit{channels}.
Therefore, for these methods to be applicable, the entire quantum circuit, including the input state, needs to share the same symmetry structure.
In other words, symmetry-based protocols become ineffective when the input states or individual channels in the quantum circuit exhibit varying symmetries.
To overcome these difficulties, we need to generalize the symmetry-based protocols to the channel level.
One possible generalization is to utilize symmetry in channel purification protocols~\cite{lee2023error, miguel2023superposed, liu2024virtual}, where the noisy channel itself is purified by coherently connecting the input and output of the noisy quantum channel.
However, these protocols mostly require multiple copies of the noisy channel, which are difficult to prepare given the limited hardware capabilities.
While there exist some channel purification protocols that use only a single copy of the noisy channel~\cite{xiong2023circuit, debroy2020extended, gonzales2023quantum, van2023single, das2024purification, yoshioka2024error}, the main focus of such protocols has been limited to utilizing Pauli symmetry for purifying near-Clifford circuits, and a comprehensive analysis of general channels and symmetries remains an open problem.

In this work, we bridge this gap by proposing \textit{symmetric channel verification} (SCV).
SCV is a novel single-copy channel purification protocol applicable to arbitrary symmetry inherent in quantum channels, and thus extends various notions of error countermeasures, such as error detection, virtual error detection, and error correction, to the channel level (see Table~\ref{tab_summary}).
Concretely, SCV detects symmetry-breaking noise in noisy channels by adding different phases to each symmetric subspace and using a quantum phase estimation-like circuit.
We characterize the sufficient conditions for the noise to be completely removed by applying SCV.
We also demonstrate that we can correct errors with no sampling overhead by applying feedback operations according to the detected errors.
Moreover, we propose \textit{virtual symmetric channel verification} (virtual SCV), a hardware-efficient implementation of SCV at the expectation value level.
Virtual SCV can virtually detect symmetry-breaking noise using only a single-qubit ancilla and is robust against noise affecting the ancilla.
This method can be regarded as an extension of the virtual implementation of symmetry verification~\cite{mcclean2020decoding, cai2021quantum, endo2022quantum, tsubouchi2023virtual} to the channel level.
Because of its robustness against the noise in the ancilla qubit, virtual SCV can be used to completely mitigate errors on idling qubits, which appear frequently in early fault-tolerant algorithms based on qubitization~\cite{low2017optimal, babbush2018encoding, low2019hamiltonian, yoshioka2024hunting}.

\begin{table*}[t]
    \centering
    \begin{tabular}{c||c|c|c}
                      & Error detection & Virtual error detection & Error correction \\ \hline
        State-level protocols   & Refs.~\cite{bonet2018low, mcardle2019error} & Refs.~\cite{mcclean2020decoding, cai2021quantum, endo2022quantum, tsubouchi2023virtual} & Refs.~\cite{shor1995scheme, knill1996threshold, aharonov1997fault, lidar2013quantum, ofek2016extending, krinner2022realizing, google2023suppressing, sivak2023real, bluvstein2024logical, ai2024quantum} \\
        Channel-level protocols & \checkmark (SCV, Sec.~\ref{sec_SCV}) & \checkmark (Virtual SCV, Sec.~\ref{sec_VSCV}) & \checkmark (SCV+feedback, Sec.~\ref{sec_SCV_correct})
    \end{tabular}
    \caption{Summary of the contributions of this work. In this work, we expand the error detection, virtual error detection, and error-correction techniques for quantum \textit{states} to quantum \textit{channels}. More specifically, we introduce single-copy channel purification protocols that can detect (symmetric channel verification, SCV), virtually detect (virtual symmetric channel verification, virtual SCV), and correct (SCV + feedback) symmetry-breaking errors happening on quantum channels. Furthermore, in setups where only Clifford unitaries can be used for noise purification, relevant in the early fault-tolerant regime, we show that symmetric channel verification under Pauli symmetry represents the optimal protocol for detecting and correcting errors (Sec.~\ref{sec_Clifford_purification}).}
    \label{tab_summary}
\end{table*}

As an application of our protocols, we apply SCV and virtual SCV under Pauli symmetry to mitigate logical errors on various early fault-tolerant quantum algorithms, such as Hamiltonian simulation and phase estimation.
We demonstrate for certain well-structured Hamiltonian simulation circuits that, by applying SCV at the logical level, arbitrary single-qubit errors on the circuit can be detected and corrected.
Moreover, we demonstrate that the effect of errors on idling qubits can be completely removed by virtual SCV.
This is particularly useful in Hamiltonian simulation and phase estimation based on qubitization, where we observe a quadratic reduction of logical errors.
We also discuss applying SCV to quantum channels under particle number conservation symmetry.

We further discuss the ability of our protocols to purify noisy logical channels in the early fault-tolerant regime.
When mitigating logical errors in this regime, it is desirable to implement Clifford unitaries~\cite{piveteau2021error, lostaglio2021error, suzuki2022quantum, tsubouchi2024symmetric}, as they are relatively easy to implement at the logical level.
Therefore, understanding the fundamental limitations in purifying noisy logical channels using Clifford unitaries is important from a practical perspective.
Here, we show that for Hamiltonian simulation circuits, the noise detectable and correctable using channel purification protocols restricted to Clifford unitaries exactly matches the noise detectable and correctable using SCV under Pauli symmetry.
Moreover, based on a resource-theoretic analysis~\cite{regula2021fundamental}, we derive an optimal bound for approximately purifying noisy Pauli rotation gates, where equality can be achieved using SCV under Pauli symmetry.
These are remarkable examples of SCV fully exploiting the available operations to their ultimate limit.

By proposing SCV and virtual SCV, we provide a unified framework of single-copy channel purification protocols utilizing symmetries inherent in quantum channels.
Thus, our results pave new ways for reducing errors in quantum computation. 
Moreover, SCV and virtual SCV offer various advantages over existing noise-reduction protocols.
Unlike traditional symmetry verification methods~\cite{bonet2018low, mcardle2019error} and their virtual variants~\cite{mcclean2020decoding, cai2021quantum, endo2022quantum, tsubouchi2023virtual}, which are state-level protocols, our channel-level protocols can be applied even when input states lack the requisite symmetry or when different quantum channels in a circuit exhibit varying symmetries.
Compared to other channel purification protocols that require multiple copies of noisy channels as input~\cite{lee2023error, miguel2023superposed, liu2024virtual}, symmetric channel verification requires only a single input, making it easier to implement in real devices.
Furthermore, existing single-copy channel purification protocols~\cite{xiong2023circuit, debroy2020extended, gonzales2023quantum, van2023single, das2024purification} can only purify noisy near-Clifford gates, whereas our method can be applied to general non-Clifford unitaries, regardless of whether the symmetry is discrete or continuous, Abelian or non-Abelian.
Such fundamental distinctions highlight the advancement achieved in our work.

This paper is organized as follows.
In Sec.~\ref{sec_channel_purification}, we introduce the setup of channel purification.
Sec.~\ref{sec_SCV} focuses on symmetric channel verification (SCV), providing theoretical details and implementation strategies.
Sec.~\ref{sec_VSCV} presents virtual symmetric channel verification (virtual SCV), a hardware-efficient variant of SCV.
In Sec.~\ref{sec_SCV_correct}, we consider correcting errors rather than detecting them using SCV.
Finally, Sec.~\ref{sec_Clifford_purification} discusses the limitations of channel purification protocols restricted to Clifford unitaries.
We make our conclusion in Sec.~\ref{sec_conclusion} by summarizing our results and discussing possible future directions.

\section{Setup of channel purification}
\label{sec_channel_purification}
We first provide the problem setup considered in the task of channel purification~\cite{lee2023error, miguel2023superposed, liu2024virtual, xiong2023circuit, debroy2020extended, gonzales2023quantum, van2023single, das2024purification}.
Specifically, we focus on single-copy channel purification protocols, where the noisy channel is accessed only once.
Let us assume that we want to implement an $n$-qubit unitary channel $\mathcal{U}(\cdot) = U\cdot U^\dag$, but the channel is affected by a noise channel $\mathcal{N}$.
As depicted in Fig.~\ref{fig_channel_purification} (a), channel purification aims to eliminate the effect of noise channel $\mathcal{N}$ under the constraint that we only have access to the noisy channel $\mathcal{U}_{\mathcal{N}}$, by coherently connecting the input and output of the noisy channel with ancilla qubits.
Concretely, we prepare $m$-qubit ancilla initialized in $\ket{0^m}$ and apply $(n+m)$-qubit unitary operations $\mathcal{U}_{\mathrm{E}}(\cdot) = U_{\mathrm{E}}\cdot U_{\mathrm{E}}^\dag$ and $\mathcal{U}_{\mathrm{D}}(\cdot) = U_{\mathrm{D}}\cdot U_{\mathrm{D}}^\dag$ before and after the noisy channel $\mathcal{U}_{\mathcal{N}}$.
Then, we measure the ancilla qubits in the computational basis and either (i) post-select, (ii) average, or (iii) discard the measurement results.
In the following, let us describe how these processes can be understood as transforming the noisy channel $\mathcal{U}_{\mathcal{N}}$ using a quantum supermap $\Theta$.

First, post-selecting the measurement results corresponds to transforming the noisy channel $\mathcal{U}_{\mathcal{N}}$ into a trace non-increasing map $\Theta^{\mathrm{det}}(\mathcal{U}_{\mathcal{N}})$ defined as
\begin{equation}
    \label{eq_purification_detection}
    \rho \mapsto \bra{0^m}\mathcal{U}_{\mathrm{D}}\circ(\mathcal{U}_{\mathcal{N}}\otimes\mathcal{I}_{\mathrm{a}})\circ\mathcal{U}_{\mathrm{E}}(\rho\otimes\ketbra{0^m})\ket{0^m}.
\end{equation}
This map outputs the normalized quantum state with probability $\mathrm{tr}[(\Theta^{\mathrm{det}}(\mathcal{U}_{\mathcal{N}}))(\rho)]$.
Here, $\ket{0^m}$ represents the initial state in the ancilla qubits and $\mathcal{I}_{\mathrm{a}}$ is the identity operation on the ancilla qubits.
Channel purification with post-selection can be regarded as detecting errors in the noise channel $\mathcal{N}$.
In Sec.~\ref{sec_SCV}, we propose a channel purification protocol that can detect symmetry-breaking errors by appropriately constructing the unitaries $U_{\mathrm{E}}$ and $U_{\mathrm{D}}$ in this setup.

\begin{figure}[t]
    \begin{center}
        \includegraphics[width=0.99\linewidth]{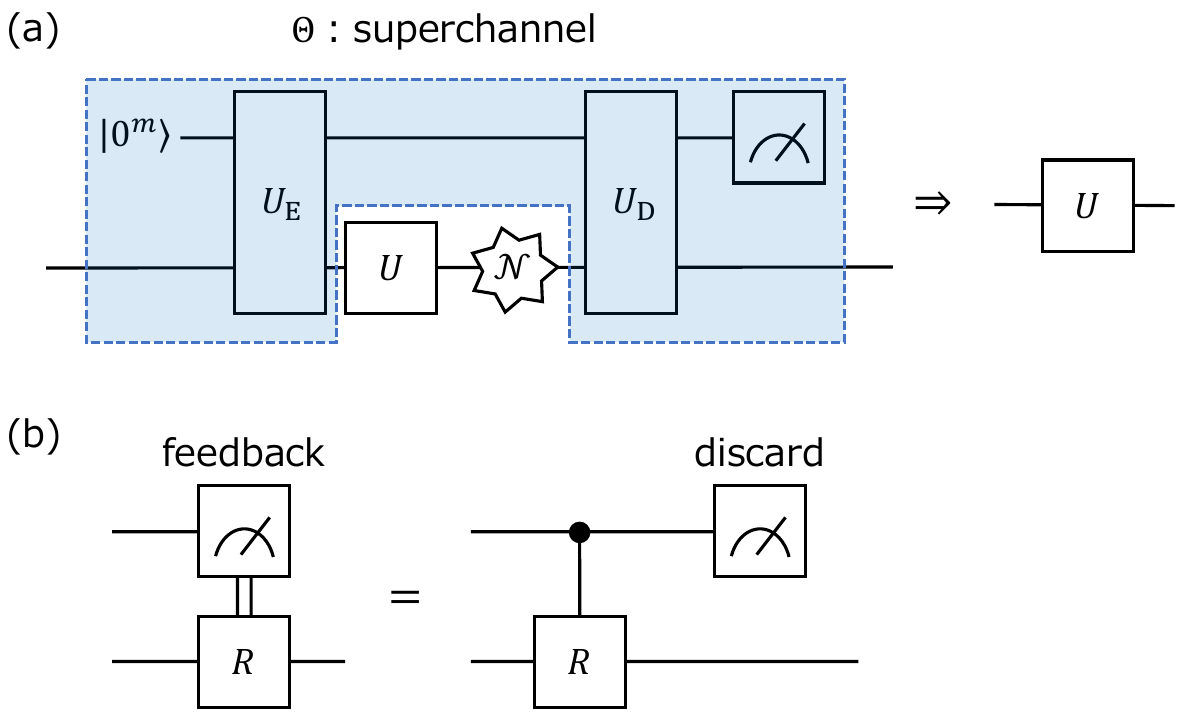}
        \caption{
        Schematic of channel purification. (a) We aim to reduce the effect of the noise $\mathcal{N}$ by preparing $m$-qubit ancilla initialized in $\ket{0^m}$, applying $(n+m)$-qubit unitary operations $U_{\mathrm{E}}$ and $U_{\mathrm{D}}$ before and after the noisy channel, and measuring the ancilla qubits in the computational basis either to (i) post-select, (ii) average, or (iii) discard the measurement results. These operations can be regarded as applying a quantum supermap $\Theta$ to the noisy channel. (b) Applying feedback control to the system can be interpreted as applying control operation and discarding the measurement results.
        }
        \label{fig_channel_purification}
    \end{center}
\end{figure}

Second, when we average the measurement results, the effect of the transformed channel can be regarded as a virtual map $\Theta^{\mathrm{vir}}(\mathcal{U}_{\mathcal{N}})$ acting as
\begin{equation}
    \label{eq_purification_virtual}
    \rho \mapsto \mathrm{tr}_{\mathrm{a}}[\mathcal{U}_{\mathrm{D}}\circ(\mathcal{U}_{\mathcal{N}}\otimes\mathcal{I}_{\mathrm{a}})\circ\mathcal{U}_{\mathrm{E}}(\rho\otimes\ketbra{0^m})(I_n\otimes Z^{\otimes m})],
\end{equation}
where $I_n$ is the identity operator of the $n$-qubit system.
Here, a virtual map means that we cannot obtain the output quantum state $(\Theta^{\mathrm{vir}}(\mathcal{U}_{\mathcal{N}}))(\rho)$ itself, but it can only be used to calculate an expectation value of an observable $O$ as $\mathrm{tr}[(\Theta^{\mathrm{vir}}(\mathcal{U}_{\mathcal{N}}))(\rho)O]$.
In Sec.~\ref{sec_VSCV}, we propose a virtual error-detection protocol that can simulate the error-detected map $\Theta^{\mathrm{det}}(\mathcal{U}_{\mathcal{N}})$ using a virtual map $\Theta^{\mathrm{vir}}(\mathcal{U}_{\mathcal{N}})$ with smaller hardware overhead.

Finally, discarding the measurement results leads to applying a superchannel $\Theta^{\mathrm{cor}}$ which maps the noisy channel $\mathcal{U}_{\mathcal{N}}$ into a valid quantum channel $\Theta^{\mathrm{cor}}(\mathcal{U}_{\mathcal{N}})$ defined as 
\begin{equation}
    \label{eq_purification_correction}
    \rho \mapsto \mathrm{tr}_{\mathrm{a}}[\mathcal{U}_{\mathrm{D}}\circ(\mathcal{U}_{\mathcal{N}}\otimes\mathcal{I}_{\mathrm{a}})\circ\mathcal{U}_{\mathrm{E}}(\rho\otimes\ketbra{0^m})].
\end{equation}
We note that measuring the ancilla qubits and applying feedback control to the system can be characterized using $\Theta^{\mathrm{cor}}$ by replacing the measurement+feedback control with controlled-unitary before the measurement (see Fig.~\ref{fig_channel_purification} (b)).
Therefore, channel purification discarding the measurement results corresponds to correcting errors in the noise channel $\mathcal{N}$.
In Sec.~\ref{sec_SCV_correct},  we construct a channel purification protocol that can correct symmetry-breaking errors in this setup.

Throughout this work, unless otherwise specified, we primarily focus on single-copy channel purification in the early fault-tolerant regime. 
In this regime, purification protocols are applied at the logical level: we assume that ideal recovery operations have already been applied to the noisy channel, such that $\mathcal{U}_{\mathcal{N}}$ can be regarded as a noisy \textit{logical} channel.
In this case, $\mathcal{N}$ represents a logical error that was not corrected by the quantum error-correction process.
A notable feature of this regime is that Clifford operations can be implemented much more efficiently than non-Clifford operations. 
Therefore, when reducing the effect of noise, it is desirable to rely on Clifford operations~\cite{piveteau2021error, lostaglio2021error, suzuki2022quantum, tsubouchi2024symmetric}. 
This implies that $U_{\mathrm{E}}$ and $U_{\mathrm{D}}$ should preferably be Clifford operations in order to minimize the hardware overhead of purification.
We discuss theoretical limitations of such Clifford-restricted channel purification protocols in Sec.~\ref{sec_Clifford_purification}.
Note that these analyses do not imply that our methods are not applicable to the near-term regime. We further discuss the applicability of our results to near-term quantum devices in Appendix~\ref{sec_NISQ}.

\section{Symmetric channel verification}
\label{sec_SCV}

\subsection{General framework}
In many quantum circuits, the unitary $U$ we aim to implement often exhibits inherent symmetry. We assume that we know an operator $S$ such that $[U, S] = 0$; for instance, when considering the time evolution of a Hamiltonian $H$ with particle number conservation symmetry, $U = e^{i\theta H}$ commutes with $S = \sum_i Z_i$.
Additionally, when $U$ is part of a quantum circuit that can be diagonalized in the computational basis, we have $[U, S] = 0$ with $S = Z_i$.
In these cases, $U$ can be block diagonalized with respect to the eigenspaces of $S$.
Let $\{\Pi_i\}_{i=1}^M$ represent the set of projectors onto the eigenspaces of $S$.
The ideal unitary $U$ can then be expressed as $U = \sum_{i} \Pi_i U \Pi_i$, representing its block diagonalization.
Thus, the ideal channel $\mathcal{U}$ can be decomposed as
\begin{equation}
    \label{eq_unitary_decomposition}
    \mathcal{U}(\cdot) = \sum_{i,j} \Pi_i\mathcal{U}(\Pi_i\cdot\Pi_j)\Pi_j,
\end{equation}
indicating that $\mathcal{U}$ maps the $(i, j)$-th block element to the $(i, j)$-th block element.

When the input state $\rho$ to the unitary $\mathcal{U}$ lies within an eigenspace of $S$ as $\rho = \Pi_i \rho \Pi_i$, the ideal output state also remains in the eigenspace as $\mathcal{U}(\rho) = \Pi_i \mathcal{U}(\rho) \Pi_i$.
This property can be utilized to detect noise $\mathcal{N}$ that breaks the symmetry by using a quantum error-mitigation technique called symmetry verification~\cite{bonet2018low, mcardle2019error}.
By measuring the noisy state $\mathcal{U}_{\mathcal{N}}$ with the POVM $\{\Pi_i\}_{i=1}^M$ and post-selecting $\Pi_i$, symmetry verification detects symmetry-breaking errors and produces a state
\begin{equation}
    \frac{\Pi_i\mathcal{U}_{\mathcal{N}}(\rho)\Pi_i}{\mathrm{tr}[\Pi_i\mathcal{U}_{\mathcal{N}}(\rho)\Pi_i]}
\end{equation}
with probability ${\mathrm{tr}[\Pi_i\mathcal{U}_{\mathcal{N}}(\rho)\Pi_i]}$.
This state has a larger overlap with the ideal state $\mathcal{U}(\rho)$ compared to the noisy state $\mathcal{U}_\mathcal{N}(\rho)$, so the impact of noise is effectively reduced.

Conversely, if the input state $\rho$ does not lie within an eigenspace of $S$, the ideal output state $\mathcal{U}(\rho)$ will also not belong to an eigenspace of $S$.
In such cases, symmetry verification cannot be applied to detect noise $\mathcal{N}$, as the measurement with the POVM $\{\Pi_i\}_{i=1}^M$ may disrupt the ideal state $\mathcal{U}(\rho)$.
These scenarios frequently arise when one wishes to simulate the time evolution of either a time-independent or dependent Hamiltonian for arbitrary states, or when $U$ is a component of a quantum circuit where other parts lack the same symmetry.
Nevertheless, the noise $\mathcal{N}$ still breaks the symmetry of the ideal channel $\mathcal{U}$, preventing the $(i, j)$-th block element from being mapped to the $(i, j)$-th block element as described in Eq.~\eqref{eq_unitary_decomposition}.
This symmetry-breaking property can be exploited to reduce noise effects by purifying the noisy channel $\mathcal{U}_{\mathcal{N}}$ itself, rather than the output state $\mathcal{U}_{\mathcal{N}}(\rho)$.

Here, we propose \textit{symmetric channel verification} (SCV), a single-copy channel purification protocol that purifies the noisy channel $\mathcal{U}_{\mathcal{N}}(\cdot)$ by detecting symmetry-breaking noise.
To implement SCV, we utilize the circuit structure depicted in Fig.~\ref{fig_SCV}.
The SCV gadget introduces different phases to each symmetric subspace and employs a quantum phase estimation-like circuit to detect changes in the symmetric subspace before and after the noisy channel.
We first prepare $m = \lceil \log_2M \rceil$-qubit ancilla and apply a Hadamard gate to each ancilla qubit.
Next, for $k=1,\ldots,m$, we sandwich the noisy channel $\mathcal{U}_{\mathcal{N}}$ between controlled-$V_S^{2^{k-1}}$ and controlled-$V_S^{\dag2^{k-1}}$,  where $k$-th ancilla qubit is taken as the control qubit and $V_S$ applies a phase $\exp[\frac{2\pi i}{2^m}j]$ to the $j$-th eigenspace as
\begin{equation}
    \label{eq_V_S}
   V_S = \sum_j\exp\qty[\frac{2\pi i}{2^m}j] \Pi_j.    
\end{equation}
Then, the inverse quantum Fourier transform is applied to the ancilla.
Finally, we measure the ancilla in the computational basis and post-select the measurement result $\ket{0^m}$.

\begin{figure}[t]
    \begin{center}
        \includegraphics[width=0.99\linewidth]{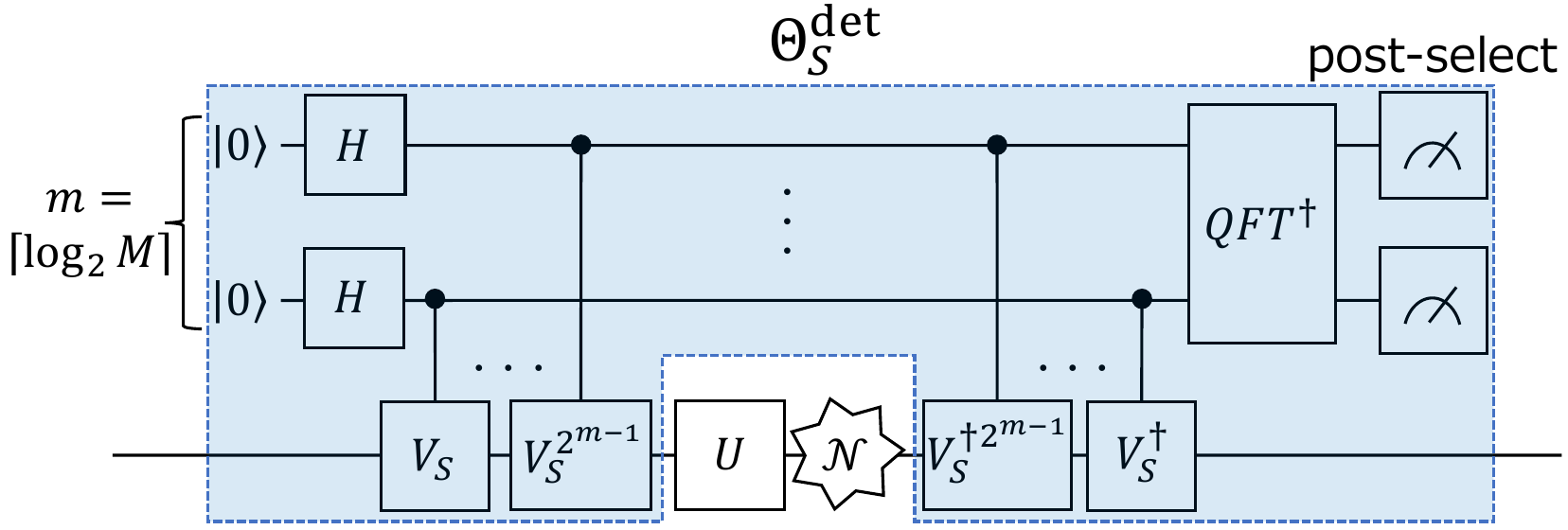}
        \caption{Circuit structure of symmetric channel verification (SCV). We prepare $m = \lceil \log_2M \rceil$-qubit ancilla, apply Hadamard gate, apply controlled-$V_S^{2^{k-1}}$ for $k=1,\ldots,m$ where the $k$-th ancilla is the control qubit and $V_S = \sum_j \exp[\frac{2\pi i}{2^m}j]\Pi_j$, apply the noisy channel $\mathcal{U}_{\mathcal{N}} = \mathcal{N}\circ\mathcal{U}$, apply controlled-$V_S^{\dag2^{k-1}}$ for $k=1,\ldots,m$, perform the inverse quantum Fourier transform, measure the ancilla qubits in the computational basis, and post-select the measurement result $\ket{0^m}$. This operation transforms the noisy channel $\mathcal{U}_\mathcal{N}(\cdot)$ into $(\Theta^{\mathrm{det}}_S(\mathcal{U}_\mathcal{N}))(\cdot)=\sum_{ij} \Pi_i \mathcal{U}_{\mathcal{N}}(\Pi_i\cdot \Pi_j)\Pi_j$.
        }
        \label{fig_SCV}
    \end{center}
\end{figure}

When a state $\rho$ is an input to this circuit, the state before measurement is
\begin{equation}
    \frac{1}{2^m} \sum_{ii'jj'kk'} \delta_{i, i'+k}\delta_{j, j'+k'} \Pi_{i'}\mathcal{U}_{\mathcal{N}}(\Pi_{i}\rho\Pi_{j})\Pi_{j'} \otimes \ketbra{k}{k'},
\end{equation}
where $\delta$ represents the Kronecker delta with $\mathrm{mod}~2^m$.
Upon post-selection with $\ket{0^m}$, the resulting state becomes
\begin{equation}
    \frac{\sum_{ij}\Pi_i\mathcal{U}_\mathcal{N}(\Pi_i \rho \Pi_j)\Pi_j}{\mathrm{tr}[\sum_{ij}\Pi_i\mathcal{U}_\mathcal{N}(\Pi_i \rho \Pi_j)\Pi_j]}
\end{equation}
with probability $\mathrm{tr}[\sum_{ij}\Pi_i\mathcal{U}_\mathcal{N}(\Pi_i \rho \Pi_j)\Pi_j]$.
Thus, SCV maps the noisy channel $\mathcal{U}_{\mathcal{N}}$ to a trace non-increasing map $(\Theta^{\mathrm{det}}_S (\mathcal{U}_{\mathcal{N}}))(\cdot) = \sum_{ij}\Pi_i\mathcal{U}_\mathcal{N}(\Pi_i \cdot \Pi_j)\Pi_j$.

Through this channel transformation, the effects of symmetry-breaking noise are removed.
To illustrate, consider a noise channel expressed as $\mathcal{N}(\cdot) = (1-p)\cdot + p \sum_j N_j \cdot N_j^\dag$ with $\sum_i \Pi_i N_j \Pi_i = 0$, indicating that with probability $p$, a quantum state in the $i$-th eigenspace transitions outside the eigenspace.
In this scenario, $\Theta^{\mathrm{det}}_S (\mathcal{U}_{\mathcal{N}}) = (1-p) \mathcal{U}$, meaning the noise effect is removed with probability $1-p$.
More generally, we establish the following theorem, whose proof is provided in Appendix~\ref{sec_select}.
\begin{thm}
    \label{thm_1}
    Let $\Theta^{\mathrm{det}}_S$ denote a quantum supermap as defined in Fig.~\ref{fig_SCV}.
    If the noise channel $\mathcal{N}(\cdot) = \sum_j N_j\cdot N_j^\dag$ satisfies $\sum_i \Pi_i N_j \Pi_i \propto I_n$, then
    \begin{equation}
        (\Theta^{\mathrm{det}}_S(\mathcal{U}_\mathcal{N}))(\cdot) = \sum_{ij}\Pi_i\mathcal{U}_\mathcal{N}(\Pi_i \cdot \Pi_j)\Pi_j \propto \mathcal{U}(\cdot).
    \end{equation}
    Consequently, when a state $\rho$ is input to the circuit in Fig.~\ref{fig_SCV}, the ideal output state $\mathcal{U}(\rho)$ is obtained with probability $\mathrm{tr}[\sum_{ij}\Pi_i\mathcal{U}_\mathcal{N}(\Pi_i \rho \Pi_j)\Pi_j]$.
\end{thm}

SCV can be seen as a generalization of symmetry verification to the channel level.
Indeed, we can detect the effects of symmetry-breaking noise, regardless of the input state. Furthermore, SCV can detect symmetry-breaking noise occurring during the execution of $\mathcal{U}$.
For simplicity, let us assume that $\mathcal{U} = \mathcal{U}_2\circ\mathcal{U}_1$ is implemented by performing two consecutive unitary channels $\mathcal{U}_1(\cdot)=U_1\cdot U_1^\dag$ and $\mathcal{U}_2(\cdot)=U_2\cdot U_2^\dag$, with the noise channel $\mathcal{N}$ occurring between $\mathcal{U}_1$ and $\mathcal{U}_2$ as $\mathcal{U}_2 \circ\mathcal{N}\circ\mathcal{U}_1$.
Even in this case, if $U_1$ and $U_2$ both commute with the operator $S$ and the noise $\mathcal{N}$ satisfies the assumption in Theorem~\ref{thm_1}, the SCV gadget can completely eliminate the noise as $\Theta^{\mathrm{det}}_S(\mathcal{U}_2 \circ\mathcal{N}\circ\mathcal{U}_1) \propto \mathcal{U}_2\circ\mathcal{U}_1 = \mathcal{U}$.
This demonstrates that SCV can purify arbitrary symmetry-breaking noise occurring during circuit execution.

If the ideal unitary $U$ has multiple commuting operators, we can concatenate the SCV gadget for different symmetries to further purify the noisy channel.
For example, in addition to $S$, if the unitary operator $U$ also commutes with another operator $S'$, which has the set of projectors $\{\Pi_i'\}_{i=1}^{M'}$ onto its eigenspaces, the ideal channel satisfies $\mathcal{U}(\cdot) = \sum_{ijkl}\Pi_i'\Pi_j\mathcal{U}(\Pi_j\Pi_i' \cdot \Pi_k'\Pi_l)\Pi_l\Pi_k'$. Therefore, by concatenating the SCV gadget for the symmetric operators $S$ and $S'$, we obtain a purified channel
\begin{equation}
    (\Theta_{S'}^{\mathrm{det}}*\Theta_S^{\mathrm{det}}(\mathcal{U}_{\mathcal{N}}))(\cdot) = \sum_{ijkl}\Pi_i'\Pi_j\mathcal{U}_{\mathcal{N}}(\Pi_j\Pi_i' \cdot \Pi_k'\Pi_l)\Pi_l\Pi_k'.
\end{equation}
Here, the composition of supermaps is defined as $\Theta_{S'}^{\mathrm{det}}*\Theta_S^{\mathrm{det}}(\cdot) = \Theta_{S'}^{\mathrm{det}}(\Theta_S^{\mathrm{det}}(\cdot))$.
This channel is proportional to the ideal channel $\mathcal{U}$ when the noise channel $\mathcal{N}(\cdot) = \sum_k N_k \cdot N_k^\dag$ satisfies $\sum_{ij} \Pi_j'\Pi_i N_k \Pi_i\Pi_j' \propto I_n$.
This condition is less restrictive than the assumption in Theorem~\ref{thm_1}: it is satisfied when $N_k$ breaks the symmetry with respect to either $S$ or $S'$.
Therefore, we can purify a broader range of noise by concatenating SCV gadgets for different symmetries.

While SCV can effectively remove the effect of symmetry-breaking errors, the SCV gadget itself may introduce additional errors.
These errors can be categorized into two types: gate errors and idling errors on the ancilla.
Gate errors are errors in the unitaries $U_{\mathrm{E}}$ and $U_{\mathrm{D}}$, where $U_{\mathrm{E}}$ and $U_{\mathrm{D}}$ represent the unitaries before and after the noisy channel $\mathcal{U}_{\mathcal{N}}$, respectively. 
These unitaries are composed of controlled-$V_S$ gates and a quantum Fourier transform, which may necessitate many quantum gates and introduce additional noise into the circuit.
However, the hardware overhead for implementing unitaries $U_{\mathrm{E}}$ and $U_{\mathrm{D}}$ does not depend on the depth of the target unitary $U$.
Therefore, when the unitary $U$ is sufficiently deep and contains significantly more quantum gates than $U_{\mathrm{E}}$ and $U_{\mathrm{D}}$, the dominant gate errors in the circuit will be the noise $\mathcal{N}$ following the unitary $U$.
In this case, the gate errors in the SCV gadget become negligible, and the overall effect of noise can be reliably reduced.

Another source of error in the SCV gadget is idling errors in the ancilla, which occur in the idle ancilla qubits between unitaries $U_{\mathrm{E}}$ and $U_{\mathrm{D}}$. 
While Pauli-$Z$ errors in the idling qubits can be detected by the measurements, Pauli-$X$ errors become problematic: such errors are undetectable and propagate to the system qubits. 
Still, such idling errors can be detected by introducing flag qubits~\cite{chao2018quantum, chamberland2018flag, chao2020flag}. 
We discuss the use of flag qubits to reduce the effect of the idling errors in Appendix~\ref{sec_flag}. 
Furthermore, even without the flag qubits, we can use noise-robust qubits as an ancilla to reduce the effect of the idling errors. 
For example, in the early fault-tolerant regime, ancilla qubits can be prepared as logical qubits with a larger code distance, so that the errors in the system qubits become dominant and idling errors become negligible.
While this increases the overall qubit count, logical idling errors can also be mitigated without any additional space overhead by post-selecting low-confidence syndrome measurement outcomes~\cite{smith2024mitigating}.
In the case where these methods are not applicable, we can use the hardware-efficient variant of SCV we introduce in Sec.~\ref{sec_VSCV}.

In addition to noise introduced by the SCV gadget, SCV also suffers from sampling overhead.
Namely, the number of circuit runs required to obtain the desired results with target accuracy is increased, because we obtain the purified channel probabilistically.
For example, when the noise can be expressed as $\mathcal{N}(\cdot)  = (1-p)\cdot + p \sum_j N_j \cdot N_j^\dag$ with $\sum_i \Pi_i N_j \Pi_i = 0$, we obtain the noiseless purified channel with probability $1-p$. Therefore, we need approximately $(1-p)^{-1}\sim(1+p)$ times as many circuit runs as in the noiseless case.
Furthermore, when there are $L$ noisy channels $\mathcal{U}_{\mathcal{N}}$ and SCV is applied to each channel, the sampling overhead scales as $O((1+p)^L)$.
Notably, this is quadratically smaller than the lower bound on sampling overhead for quantum error mitigation, which scales as $O((1+p)^{2L})$~\cite{tsubouchi2023universal,takagi2023universal}.
This improvement mainly arises from the use of post-selection based on the symmetric structure~\cite{cai2023quantum}.
We note that a similar advantage also holds for other post-selection-based techniques, such as symmetry verification~\cite{bonet2018low, mcardle2019error}.

\subsection{Application to Pauli symmetry}
\label{sec_SCV_Pauli}
As an application of SCV, let us consider a quantum channel under Pauli symmetry, which is present in various quantum algorithms including quantum many-body simulation, quantum machine learning, and quantum information processing.
Here, we first derive a general condition to completely eliminate the effect of noise, and then later provide numerical demonstrations for practical quantum circuits.

Let us define the set of $n$-qubit Pauli operators as $\mathcal{P}_n = \{I, X, Y, Z\}^{\otimes n}$.
Here, $I$ represents the identity operator on a single-qubit system, which we distinguish from the $n$-qubit identity operator $I_n$.
For simplicity, let us analyze the case where the noise channel $\mathcal{N}$ is Pauli noise represented as
\begin{equation}
    \label{eq_Pauli_noise}
    \mathcal{N}(\cdot) = \sum_{i=0}^K p_i P_i\cdot P_i,
\end{equation}
where $p_i \neq 0$ denotes the probability of an $n$-qubit Pauli operator $P_i\in\mathcal{P}_n$ affecting the target unitary, $P_0 = I_n$ represents the $n$-qubit identity operator, and $P_i\neq P_j$ if $i\neq j$.
Although our discussion assumes Pauli noise, it can be generalized to arbitrary noise by decomposing the Kraus operators into the sum of Pauli operators.

First, we consider the simplest case where the ideal unitary operator $U$ commutes with a single nontrivial Pauli operator $Q\in\mathcal{P}_n \backslash \{I_n\}$, i.e., $[U, Q] = 0$.
The Pauli operator $Q$ has two eigenspaces with projectors $\Pi_0 = \frac{1}{2}(I+Q)$ and $\Pi_1 = \frac{1}{2}(I-Q)$.
In this scenario, the operator $V_Q$ defined in Eq.~\eqref{eq_V_S} is the Pauli operator $Q$ itself.
Therefore, implementing SCV requires only a single-qubit ancilla and two controlled-$Q$ gates.
By performing SCV, we can detect Pauli error $P_j\in\mathcal{P}_n \backslash \{I_n\}$ satisfying $\sum_i \Pi_i P_j \Pi_i \propto I_n$, i.e., Pauli error $P_j$ that anti-commute with $Q$.

Next, we consider the case where the ideal unitary operator $U$ commutes with multiple nontrivial Pauli operators.
As we have discussed, SCV gadgets can be concatenated for different commuting Pauli operators.
To formalize this analysis, we define the set of Pauli operators with non-zero coefficients in the Pauli expansion of $U$ as
\begin{equation}
    \label{eq_Q_U'}
    \mathcal{Q}_{U}' = \{Q\in\mathcal{P}_n ~|~ \mathrm{tr}[UQ] \neq 0\},
\end{equation}
the set of Pauli operators generated by $\mathcal{Q}_{U}'$ up to phase as
\begin{equation}
    \label{eq_Q_U}
    \mathcal{Q}_{U} = \{Q\in\mathcal{P}_n ~|~ \exists Q_{i_1},\ldots Q_{i_j}\in\mathcal{Q}_{U}', Q\propto Q_{i_1} \cdots Q_{i_j}\},
\end{equation}
and the set of Pauli operators that commutes with $U$ as
\begin{equation}
    \mathcal{Q}_{U}^{\mathrm{com}} = \{Q\in\mathcal{P}_n ~|~ [U, Q] = 0\}.
\end{equation}

\begin{figure}[t]
    \begin{center}
        \includegraphics[width=0.99\linewidth]{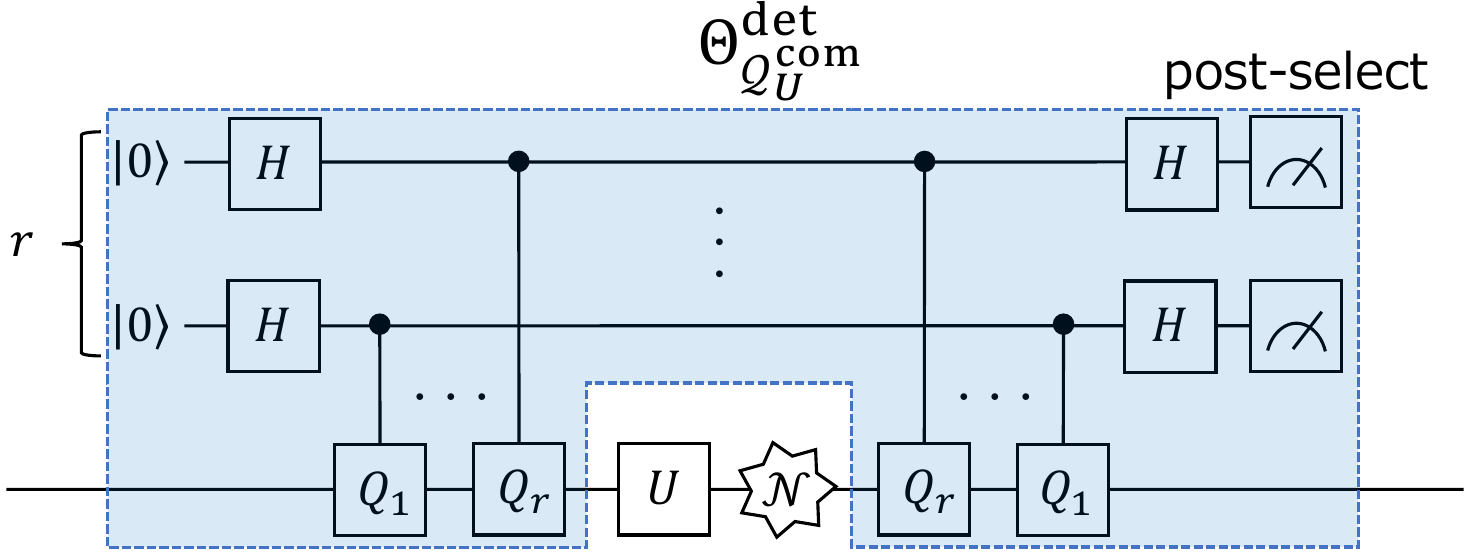}
        \caption{Circuit structure of SCV for the set of Pauli operators $\mathcal{Q}_U^{\mathrm{com}}$.
        We concatenate the SCV gadget for the generator $Q_1,\ldots,Q_r\in\mathcal{Q}_U^{\mathrm{com}}$ as $\Theta^{\mathrm{det}}_{\mathcal{Q}^{\mathrm{com}}_{U}} = \Theta^{\mathrm{det}}_{Q_1} * \cdots * \Theta^{\mathrm{det}}_{Q_r}$.
        }
        \label{fig_SCV_Pauli}
    \end{center}
\end{figure}

Let $Q_1, \ldots, Q_r \in \mathcal{Q}_{U}^{\mathrm{com}}$ be generators of the set of Pauli operators $\mathcal{Q}_{U}^{\mathrm{com}}$ up to phase.
As shown in Fig.~\ref{fig_SCV_Pauli}, we concatenate the SCV gadget for $Q_1, \ldots, Q_r$ and define the SCV gadget for the set of commuting Pauli operators $\mathcal{Q}_{U}^{\mathrm{com}}$ as
\begin{equation}
    \Theta^{\mathrm{det}}_{\mathcal{Q}^{\mathrm{com}}_{U}} = \Theta^{\mathrm{det}}_{Q_1} * \cdots * \Theta^{\mathrm{det}}_{Q_r}.
\end{equation}
Since $\Theta^{\mathrm{det}}_{Q_i}$ can detect Pauli errors that anti-commute with $Q_i$, $\Theta^{\mathrm{det}}_{\mathcal{Q}^{\mathrm{com}}_{U}}$ can detect Pauli errors that anti-commute with at least one of the generators.
In other words, the SCV gadget for $\mathcal{Q}_{U}^{\mathrm{com}}$ can detect a Pauli error $P_j$ if there exists a Pauli operator $Q \in \mathcal{Q}_{U}^{\mathrm{com}}$ that anti-commutes with $P_j$.
Thus, the performance of SCV for Pauli symmetry is characterized by the following corollary:
\begin{cor}
    \label{cor_1}
    For Pauli noise channel $\mathcal{N}(\cdot) = \sum_{i=0}^K p_i P_i\cdot P_i,$ defined in Eq.~\eqref{eq_Pauli_noise}, let $\mathcal{P}^{\mathrm{{det}}}_{\mathcal{N}} = \{P_i\}_{i=1}^K$ denote the set of Pauli operators that yield non-zero error probability.
    Then, the SCV gadget $\Theta^{\mathrm{det}}_{\mathcal{Q}^{\mathrm{com}}_{U}}$ for the set of commuting Pauli operators $\mathcal{Q}_U^{\mathrm{com}}$, as defined in Fig.~\ref{fig_SCV_Pauli}, can completely detect the Pauli noise $\mathcal{N}$, i.e., $\Theta^{\mathrm{det}}_{\mathcal{Q}^{\mathrm{com}}_{U}}(\mathcal{U}_\mathcal{N}) = p_0\mathcal{U}$, if and only if all the Pauli error $P_i\in\mathcal{P}^{\mathrm{{det}}}_{\mathcal{N}}$ anti-commutes with at least one of the $Q_j \in \mathcal{Q}^{\mathrm{com}}_{U}$:
    \begin{equation}
        \label{eq_cor_1_1}
        \forall P_i \in \mathcal{P}^{\mathrm{{det}}}_{\mathcal{N}}, \exists Q_j \in \mathcal{Q}^{\mathrm{com}}_{U}, [P_i, Q_j] \neq 0,
    \end{equation}
    or equivalently,
    \begin{equation}
        \label{eq_cor_1_2}
        \mathcal{P}^{\mathrm{{det}}}_{\mathcal{N}} \cap \mathcal{Q}_{U} = \emptyset.
    \end{equation}
\end{cor}

The proof of Corollary~\ref{cor_1}, which follows from Theorem~\ref{thm_1}, is provided in Appendix~\ref{sec_select}.
Corollary~\ref{cor_1} states that the SCV gadget for the set of Pauli operators $\mathcal{Q}_U$ can detect arbitrary Pauli errors that are not part of the ideal unitary $U$.
In this case, a noiseless channel $\mathcal{U}$ is obtained with probability $p_0$.
However, the SCV gadget cannot detect Pauli errors contained within $U$ because such errors share the same Pauli symmetry as $U$ and commute with all elements in $\mathcal{Q}_{U}^{\mathrm{com}}$.
We note that, while Corollary~\ref{cor_1} assumes Pauli noise, it can be generalized to arbitrary noise channels $\mathcal{N}(\cdot) = \sum_j N_j \cdot N_j^\dag$ by decomposing each Kraus operator $N_j$ into a sum of Pauli operators.
Specifically, the noisy channel can be purified as $\Theta^{\mathrm{det}}_{\mathcal{Q}^{\mathrm{com}}_{U}}(\mathcal{U}_\mathcal{N}) \propto \mathcal{U}$, provided that the set $\mathcal{P}^{\mathrm{{det}}}_{\mathcal{N}} = \{P\in\mathcal{P}_n-\{I_n\}~|~\exists i, \mathrm{tr}[N_i P]\neq 0\}$ satisfies the conditions in Corollary~\ref{cor_1}.
This is because if a set of Pauli operators $\mathcal{P}^{\mathrm{det}}_{\mathcal{N}}$ satisfies the conditions in Corollary~\ref{cor_1}, and hence also satisfies the condition in Theorem~\ref{thm_1}, then any linear combination of these operators also satisfies the condition in Theorem~\ref{thm_1}.

To demonstrate the power of SCV for Pauli symmetry, we consider the time-independent and time-dependent Hamiltonian dynamics simulation.
While state-level error-detection protocols~\cite{bonet2018low, mcardle2019error} can only be used when the entire circuit possesses the same symmetric structure, SCV can be applied even when the input state does not commute with the same Pauli operators or when the symmetric structure of the Hamiltonian changes over time.
Let us first discuss the time-independent Hamiltonian simulation circuit $U = e^{i\theta H}$, where $H$ represents the Hamiltonian of interest.
Note that, while Trotter or post-Trotter circuits are implemented in practice, here we neglect the effect of imperfect exponentiation for the sake of demonstrating SCV.
Under perfect exponentiation, we generally have $\mathcal{Q}_{U} = \mathcal{Q}_{H}$ and $\mathcal{Q}_{U}^{\mathrm{com}} = \mathcal{Q}_{H}^{\mathrm{com}}$, and SCV can be applied regardless of the input state.

One interesting time-independent Hamiltonian to examine is the one-dimensional (1D) Heisenberg model $H = \sum_{i}(X_{i}X_{i+1} + Y_{i}Y_{i+1} + Z_{i}Z_{i+1})$, whose dynamics simulation using quantum computer is an active topic of research~\cite{peng2022quantum, rosenberg2024dynamics}.
For this Hamiltonian, we have $\mathcal{Q}_{H}^{\mathrm{com}} = \{I^{\otimes n}, X^{\otimes n}, Y^{\otimes n}, Z^{\otimes n}\}$.
Since any weight-1 Pauli operator anti-commutes with $X^{\otimes n}$ or $Z^{\otimes n}$, we can detect single-qubit Pauli errors occurring during the time evolution.
Meanwhile, two-qubit Pauli errors such as $XX$ commute with all the elements in $\mathcal{Q}_{H}^{\mathrm{com}}$, so such errors cannot be detected.
To numerically verify the effectiveness of SCV, we consider an $n=8$ qubit 1D Heisenberg model with an open boundary condition and apply the time evolution $U = e^{i\theta H}$ to the input state $\ket{01}^{\otimes n/2}$.
We assume that the noisy circuit is represented as $\bigcirc_{l=1}^L \mathcal{N}_l \circ \mathcal{U}_l$, where $\mathcal{N}_l$ is a local depolarizing noise with an error rate $p_{\mathrm{err}}$ and $\mathcal{U}_l(\cdot) = e^{i(\theta/L) H} \cdot e^{-i(\theta/L) H}$ represents a single time step.
We also assume that the SCV gadget $\Theta^{\mathrm{det}}_{\mathcal{Q}^{\mathrm{com}}_{U}}$ is subject to noise.
Following the assumptions of the early fault-tolerant regime, we assume that $U_{\mathrm{E}}$ and $U_{\mathrm{D}}$ in the SCV gadget are affected by local depolarizing noise with an error rate $p_{\mathrm{err}}/100$.
The difference in error rates arises because approximately 100 Clifford+T gates are required to synthesize Pauli rotation gates with an accuracy of $10^{-6}$~\cite{ross2016optimal}.
We consider the case where logical qubits with a higher code distance are used for the ancilla, such that idling errors on the ancilla become negligible.
We set $\theta = 2\pi$ and $L = 100$.

\begin{figure}[t]
    \begin{center}
        \includegraphics[width=0.85\linewidth]{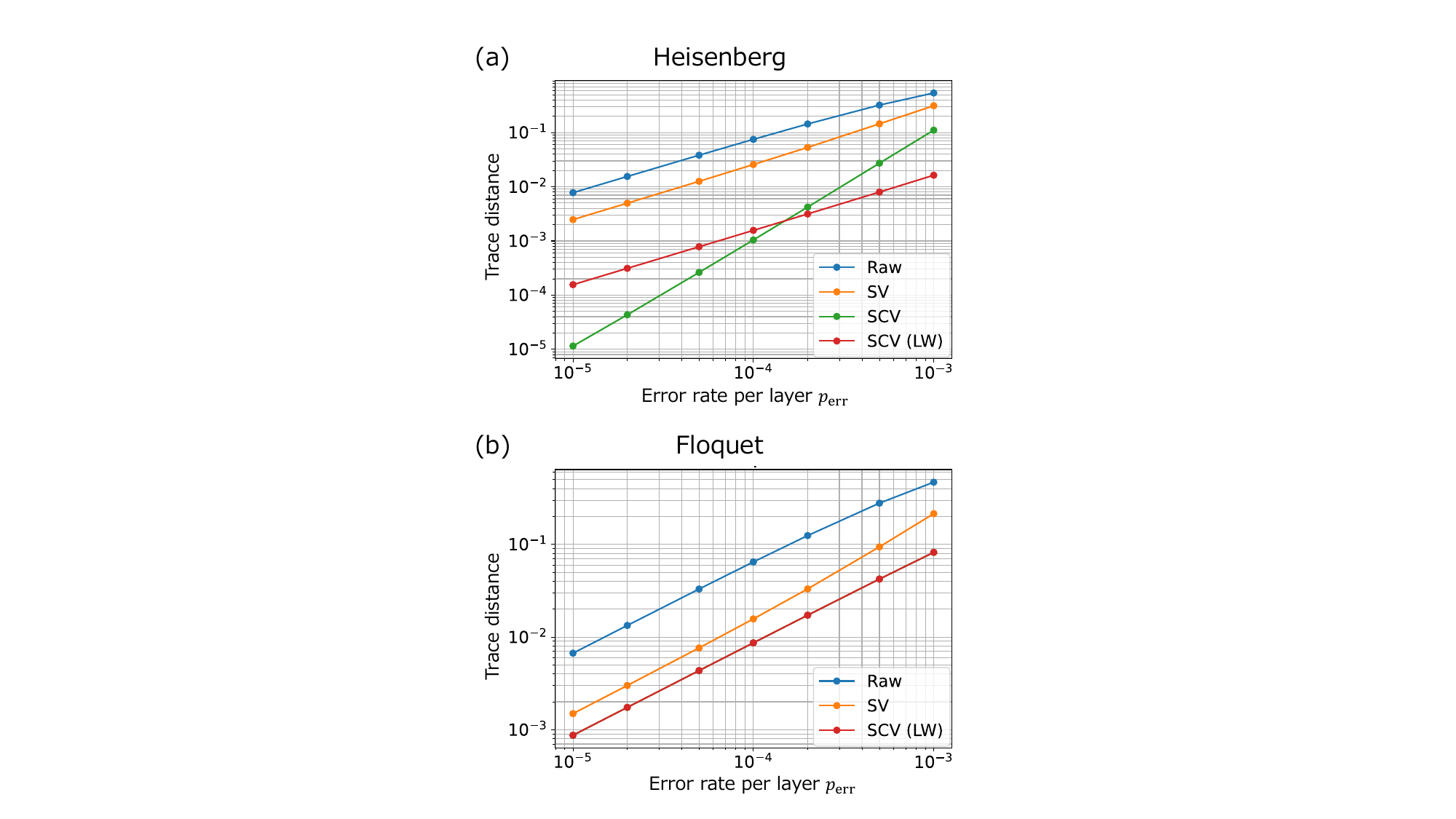}
        \caption{Performance of SCV in detecting errors on Hamiltonian simulation circuits for (a) 1D Heisenberg model and (b) Floquet dynamics. The ``Raw" line represents the trace distance between the noisy quantum state and the ideal quantum state, while the ``SV" line represents the trace distance for the quantum state purified by symmetry verification. The ``SCV" line represents the trace distance where the SCV is applied to the entire noisy circuit $\bigcirc_{l=1}^L \mathcal{N}_l \circ \mathcal{U}_l$, and the ``SCV (LW)" line represents the layer-wise application of SCV, meaning that SCV is applied to every noisy layer [$\mathcal{N}_l \circ \mathcal{U}_l$ for (a) Heisenberg model and $\mathcal{N}_l \circ \mathcal{U}_{X,l}$ and $\mathcal{N}_l \circ \mathcal{U}_{Z,l}$ for (b) Floquet dynamics].
        }
        \label{fig_numerics_Pauli}
    \end{center}
\end{figure}

Our results are illustrated in Fig.~\ref{fig_numerics_Pauli} (a).
Since the input state $\ket{01}^{\otimes n/2}$ is an eigenstate of the operator $Z^{\otimes n}$, we can apply symmetry verification for $Z^{\otimes n}$ and detect arbitrary single-qubit bit-flip errors.
However, it is not an eigenstate of the operator $X^{\otimes n}$, so symmetry verification cannot be applied for $X^{\otimes n}$, even though the Hamiltonian commutes with $X^{\otimes n}$.
Meanwhile, SCV does not depend on the input state, allowing us to detect arbitrary errors that anti-commute with $X^{\otimes n}$ or $Z^{\otimes n}$.
Thus, while symmetry verification reduces the error by a constant factor, SCV applied to the entire noisy circuit $\bigcirc_{l=1}^L \mathcal{N}_l \circ \mathcal{U}_l$ can detect arbitrary single-qubit errors and quadratically reduce the error rate.
Further error reduction can be achieved by applying SCV at every layer $\mathcal{N}_l \circ \mathcal{U}_l$.
In this case, more errors can be detected during the time evolution, but the impact of errors on the SCV gadget itself becomes dominant.
As a result, the scaling of the error worsens compared to applying SCV to the entire circuit, but in high-error-rate regimes, layer-wise SCV application performs better.

Next, as an example of time-dependent Hamiltonian dynamics simulation, we consider the Floquet system~\cite{oka2019floquet}, whose Hamiltonian changes time-periodically.
Floquet systems exhibit a rich set of physical phenomena and thus have been widely simulated on quantum computers~\cite{mi2022time, frey2022realization, ippoliti2021many}.
We consider a noiseless circuit represented as $\bigcirc_{l=1}^L \mathcal{U}_{X,l} \circ \mathcal{U}_{Z,l}$, where $\mathcal{U}_{X,l}(\cdot) = e^{i\theta H_X}\cdot e^{-i\theta H_X}$ with $H_X = \sum_i X_i$ and $\mathcal{U}_{Z,l}(\cdot) = e^{i\theta H_Z}\cdot e^{-i\theta H_Z}$ with $H_Z = \sum_i Z_{i}Z_{i+1}$.
This represents that Hamiltonian changes time-periodically between $H_X$ and $H_Z$.
When we take $\ket{+}^{\otimes n}$ as an input to this circuit, the output state is an eigenstate of the operator $X^{\otimes n}$, allowing us to perform symmetry verification for $X^{\otimes n}$.
Meanwhile, each layer $e^{i\theta H_X}$ and $e^{i\theta H_Z}$ commutes with additional Pauli operators: $\mathcal{Q}_{H_X}^{\mathrm{com}} = \{I,X\}^{\otimes n}$ and $\mathcal{Q}_{H_Z}^{\mathrm{com}} = \{I,Z\}^{\otimes n} \cdot \{I_n,X^{\otimes n}\}$.
Therefore, by applying SCV to each layer, we can detect more errors compared to simply applying symmetry verification to the output state.

To evaluate the effectiveness of SCV in this scenario, we consider a noisy circuit $\bigcirc_{l=1}^L \mathcal{N}_l\circ\mathcal{U}_{X,l} \circ \mathcal{N}_l \circ \mathcal{U}_{Z,l}$ with $n=4$, where $\mathcal{N}_l$ is a local depolarizing noise with an error rate $p_{\mathrm{err}}$.
We use the same noise model for the SCV gadget $\Theta^{\mathrm{det}}_{\mathcal{Q}^{\mathrm{com}}_{U}}$ as in the case for the Heisenberg model.
Specifically, we assume that $U_{\mathrm{E}}$ and $U_{\mathrm{D}}$ composing the SCV gadget are also affected by local depolarizing noise with an error rate $p_{\mathrm{err}}/100$, and neglect idling errors on the ancilla.
We set $\theta = 2\pi/100$ and $L = 100$.
As shown in Fig.~\ref{fig_numerics_Pauli} (b), even with noise in the SCV gadget, SCV can significantly reduce the effect of noise in the circuit.

In our numerical simulations, we have neglected idling errors in the SCV gadget by assuming that the ancilla qubits have a higher code distance than the system qubits.  
Even if this assumption does not hold—i.e., idling errors on the ancilla qubits cannot be neglected—we can remove such errors using flag qubits~\cite{chao2018quantum, chamberland2018flag, chao2020flag}.  
We further discuss this point in Appendix~\ref{sec_flag}.
Moreover, we have assumed that the error rates of the unitaries $U_{\mathrm{E}}$ and $U_{\mathrm{D}}$ in the SCV gadget are one hundred times as small as that of the noisy unitary channel $\mathcal{U}_{\mathcal{N}}$ to be purified, which is a relevant assumption in the early fault-tolerant regime.  
In Appendix~\ref{sec_NISQ}, we demonstrate that our method remains effective even when the error rate in the SCV gadget is comparable to that of $\mathcal{U}_{\mathcal{N}}$, as is often the case in setups of near-term quantum hardware.

Let us comment on the differences between our analysis of SCV under Pauli symmetry and previous works based on similar ideas.
In Refs.~\cite{xiong2023circuit, debroy2020extended, gonzales2023quantum, van2023single, das2024purification}, the authors introduce a purification gadget similar to the one shown in Fig.~\ref{fig_SCV_Pauli} to purify noisy near-Clifford circuits.
The primary focus of these studies is on the case where $U$ consists of Clifford gates.
They proposed that, in this setting, arbitrary errors can be detected by sandwiching $U$ with a controlled Pauli gate and its conjugate under $U$.
Furthermore, they analyzed the case where $U$ consists of Clifford gates supplemented with Pauli-$Z$ rotation gates.
In this scenario, they claimed that noisy Pauli-$Z$ rotation gates can be purified by sandwiching them with controlled-$Z$ gates, which is precisely what SCV for the Pauli-$Z$ operator is doing.
One of the main contributions of our work is to extend these analyses beyond simple Pauli rotation gates to general non-Clifford channels $U$.
As stated in Corollary~\ref{cor_1} and discussed in Sec.~\ref{sec_Clifford_purification}, we establish the ultimate performance limits of such purification protocols.

\subsection{Application to particle number conservation symmetry}
\label{sec_SCV_U1}
Symmetry arising from particle number conservation plays a crucial role in many quantum systems, especially in the context of fermionic systems.
The particle number conservation symmetry is defined by the commutation of the Hamiltonian or unitary operator $U$ with the total particle number operator $S = \sum_i a_i^\dag a_i$, where $a_i^\dag$ and $a_i$ represent the creation and annihilation operators, respectively.
In this case, the eigenspaces of $S$ correspond to states with fixed particle numbers, and the symmetry can be utilized to detect errors that break particle number conservation.

When the fermion system is transformed into a qubit system using the Jordan-Wigner transformation, the particle number operator can be expressed as $S = \sum_i Z_i$.
The operator $S$ has $n+1$ eigenspaces with projectors defined as $\Pi_i = \sum_{\abs{\vb*{x}}=i}\ketbra{\vb*{x}}$.
Here, $\vb*{x}\in\{0,1\}^n$ represents a bit string of length $n$ and $\abs{\vb*{x}}$ represents the number of 1 in the bit string $\vb*{x}$.
In this case, the operator $V_S$ defined in Eq.~\eqref{eq_V_S} can be represented as
\begin{equation}
    V_S = \sum_{\vb*{x}\in\{0,1\}^n} \exp\left[\frac{2\pi i}{2^m}\abs{\vb*{x}}\right] \ketbra{\vb*{x}} = R\left(\frac{2\pi}{2^m}\right)^{\otimes n},
\end{equation}
where $R(\theta) = \ketbra{0} + e^{i\theta}\ketbra{1}$.
Therefore, implementing SCV for particle number conservation symmetry requires $m = \lceil \log_2(n+1) \rceil$ qubit ancilla, $2n\lceil\log_2(n+1) \rceil$ controlled-$R(\theta)$ gates, and quantum Fourier transform.
Although noise in these operations may degrade the performance of SCV, it is effective when the target unitary $U$ is composed of more quantum gates and contains more noise.
Moreover, for a large system size $n$, the number of ancilla scales as $O(\log n)$, so the effect of errors on the ancilla is expected to be much smaller than the system qubits.

\begin{figure}[t]
    \begin{center}
        \includegraphics[width=0.85\linewidth]{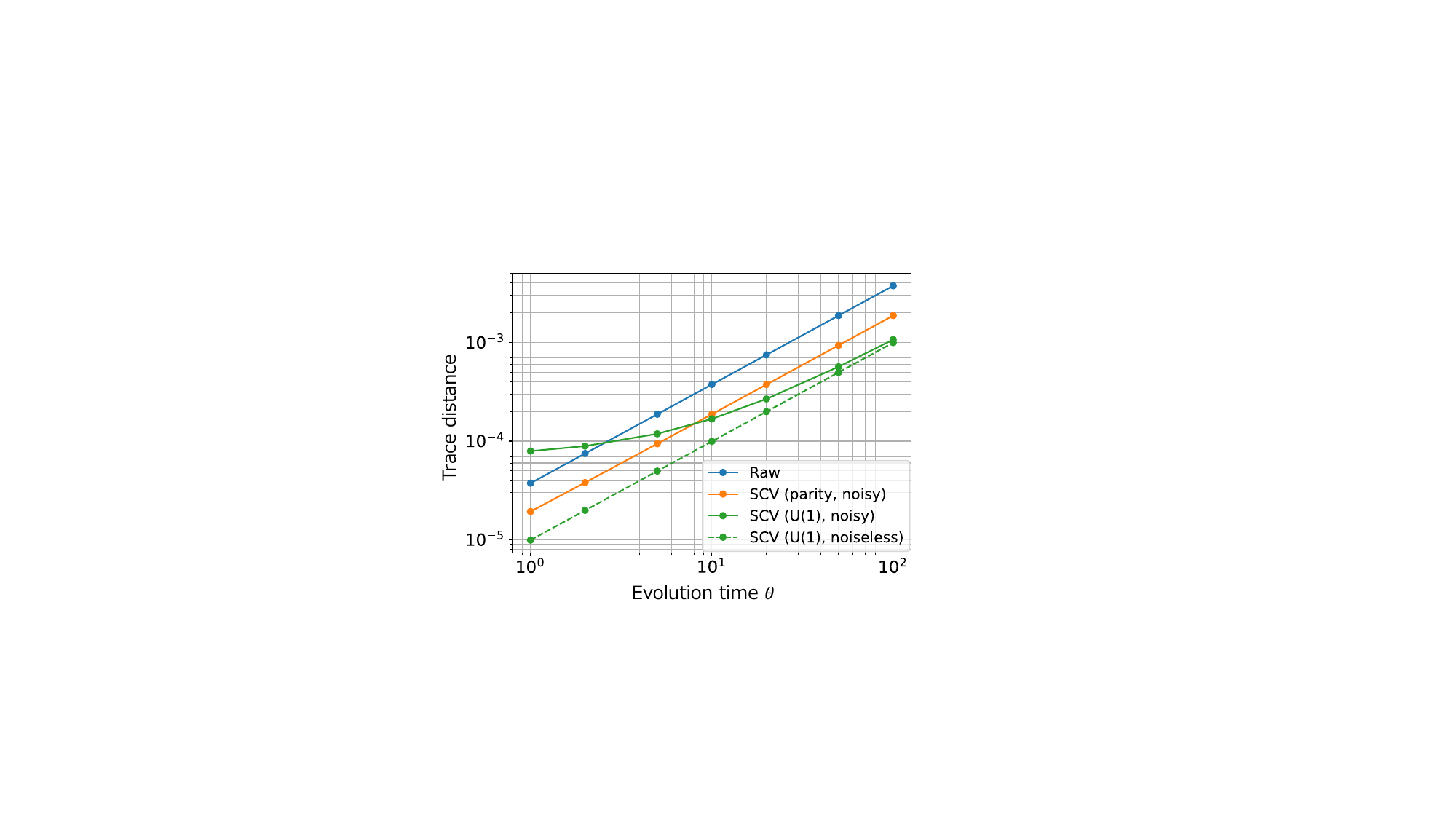}
        \caption{Performance of SCV in detecting errors on a Hamiltonian simulation circuit for the H$_2$ molecule. The ``Raw" line represents the trace distance between the noisy quantum state and the ideal quantum state. The ``SCV (parity, noisy)" line represents the result using SCV for the parity operator $\prod_i Z_i$, where we assume that the SCV gadget is affected by noise. The ``SCV (U(1), noiseless)" and ``SCV (U(1), noisy)" lines represent the trace distance for the quantum state with the noise detected by SCV for the particle number operator $\sum_i Z_i$, where ``noiseless" and ``noisy" represent the results using noiseless and noisy SCV gadgets.
        }
        \label{fig_numerics_U1}
    \end{center}
\end{figure}

To evaluate the effectiveness of SCV for particle number conservation symmetry, we consider a time evolution circuit for the H$_2$ molecule with a Haar-random input state.
Since the input state is not an eigenstate of $S = \sum_i Z_i$, symmetry verification cannot be applied, but SCV can be used to purify the circuit.
In Fig.~\ref{fig_numerics_U1}, we present the trace distance between the ideal noiseless state and the error-detected state for varying evolution time $\theta$.
Here, we assume that the noisy circuit we implement is $\mathcal{N} \circ \mathcal{U}$, where $\mathcal{U}(\cdot) = e^{i\theta H} \cdot e^{-i\theta H}$ represents the time evolution for the H$_2$ molecule (STO-3G basis set) with $n=4$ and an atomic distance of 1.1~\AA, and $\mathcal{N}$ is a global depolarizing noise with error rate $n\theta p_{\mathrm{err}}$ with $p_{\mathrm{err}} = 10^{-5}$.
We further assume that the SCV gadget $\Theta^{\mathrm{det}}_{S}$ is subject to noise.
We assume that $U_{\mathrm{E}}$ and $U_{\mathrm{D}}$ in $\Theta^{\mathrm{det}}_{S}$ is affected by a local depolarizing noise with error rate $p_{\mathrm{err}}$.
Note that $U_{\mathrm{E}}$ and $U_{\mathrm{D}}$ are not Clifford unitaries in this case, and thus we set their error rates to be comparable to that of the noisy unitary channel $\mathcal{N} \circ \mathcal{U}$.
We consider the case where logical qubits with a higher code distance are used for the ancilla, such that idling errors on the ancilla become negligible.
The impact of idling errors is discussed in Appendix~\ref{sec_flag}.

From Fig.~\ref{fig_numerics_U1}, we observe that SCV can reduce the effect of errors when there is no noise in the SCV gadget.
Meanwhile, when the SCV gadget is affected by noise, applying SCV can degrade the performance.
Nevertheless, as the evolution time $\theta$ increases and the circuit accumulates more noise, errors in the target unitary $\mathcal{U}$ become the dominant error source, and the impact of errors on the SCV gadget diminishes.
This demonstrates that when the target unitary contains a sufficient number of quantum gates relative to the SCV gadget, SCV can effectively detect noise.

Furthermore, when $U$ commutes with the particle number operator $\sum_i Z_i$, it also commutes with the parity operator $\prod_i Z_i$.
Therefore, we can apply SCV to the Pauli operator $\prod_i Z_i$ to detect errors, whose corresponding gadget is significantly easier to implement than that for $\sum_i Z_i$.
In Fig.~\ref{fig_numerics_U1}, we also show the performance of SCV for the parity operator $\prod_i Z_i$, where $U_{\mathrm{E}}$ and $U_{\mathrm{D}}$ comprising the SCV gadget are affected by local depolarizing noise with an error rate $p_{\mathrm{err}}/100$.
Note that the difference in error rates arises because approximately 100 Clifford+T gates are required to synthesize Pauli rotation gates with high accuracy~\cite{ross2016optimal}.
Since the SCV gadget consists of simple Clifford gates that do not require gate synthesis, SCV for the parity operator $\prod_i Z_i$ can outperform SCV for the particle number operator $S = \sum_i Z_i$ for small $\theta$.
Nevertheless, when $\theta$ increases and the noise on the target unitary becomes dominant, SCV for the particle number operator $\sum_i Z_i$ can detect twice as much noise as SCV for the parity operator $\prod_i Z_i$.

\section{Virtual symmetric channel verification}
\label{sec_VSCV}

\subsection{General framework}
While SCV is effective for detecting symmetry-breaking noise in noisy channels, it requires sophisticated quantum operations such as controlled-$V_S$ and quantum Fourier transform.
In the near-term and early fault-tolerant regimes, implementing these operations may introduce prohibitive noise into the quantum circuits and degrade the performance of SCV.
Thus, a simpler and more hardware-friendly noise-reduction technique is necessary for practical implementation.

To reduce the hardware overhead for noise purification, we introduce a virtual supermap $\Theta^{\mathrm{vir}}_{ijkl}$ depicted in Fig.~\ref{fig_VSCV}.
This supermap consists of a single-qubit ancilla and two controlled Pauli gates, making it significantly easier to implement compared to the original SCV gadget.
When we take the average of the measurement results from the ancilla qubit, $\Theta^{\mathrm{vir}}_{ijkl}$ maps the noisy channel $\mathcal{U}_{\mathcal{N}}$ to a virtual map represented as
\begin{equation}
    \label{eq_virtual_supermap}
    (\Theta^{\mathrm{vir}}_{ijkl}(\mathcal{U}_{\mathcal{N}}))(\cdot) = \frac{1}{2}(P_k\mathcal{U}_{\mathcal{N}}(P_i\cdot P_j)P_l + P_l\mathcal{U}_{\mathcal{N}}(P_j\cdot P_i)P_k).
\end{equation}
It should be noted that $\Theta^{\mathrm{vir}}_{ijkl}(\mathcal{U}_{\mathcal{N}})$ is a virtual map, which can only be used for the task of expectation value estimation.

\begin{figure}[t]
    \begin{center}
        \includegraphics[width=0.99\linewidth]{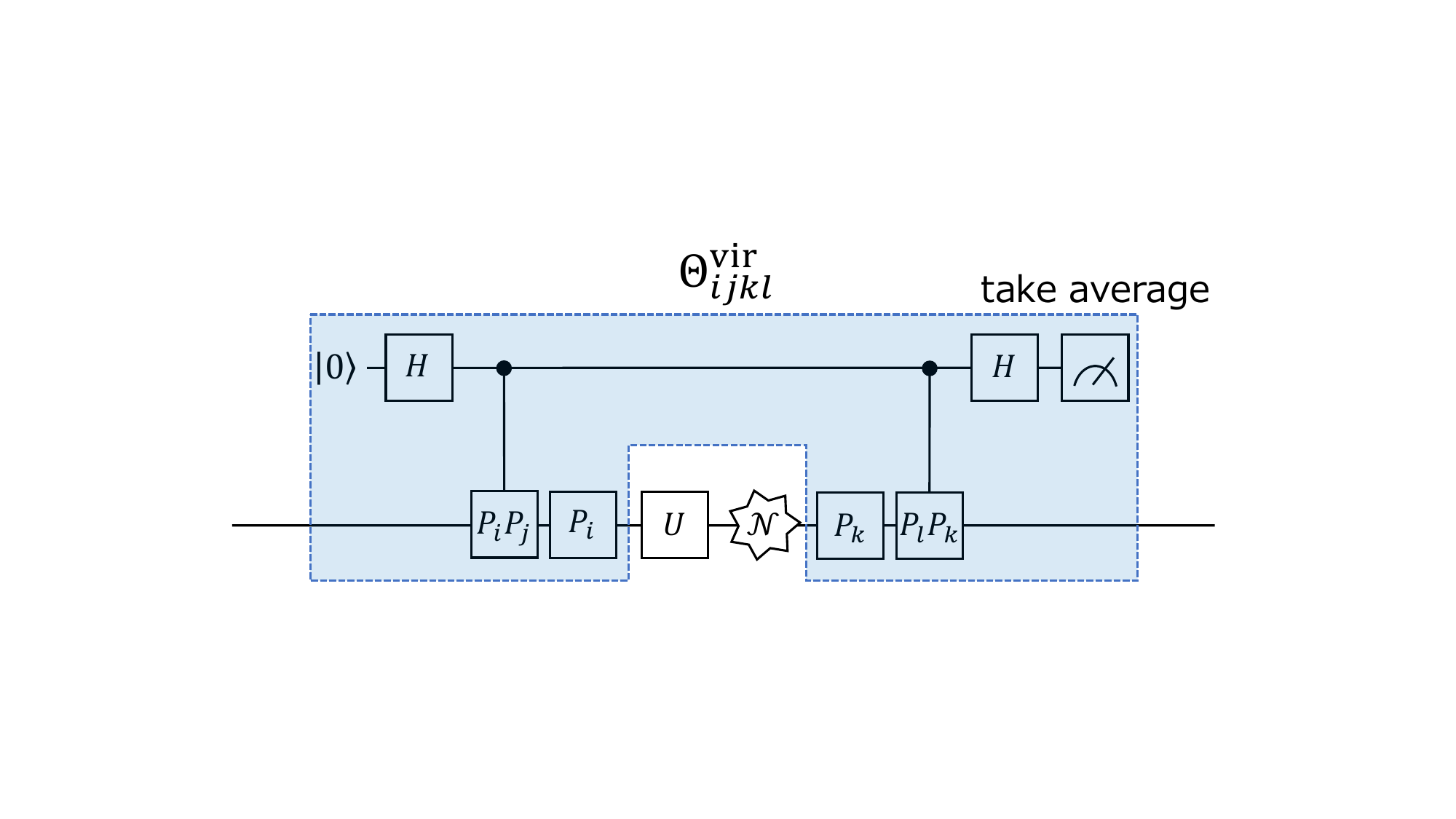}
        \caption{A gadget to implement virtual symmetric channel verification (virtual SCV). We prepare a single-qubit ancilla initialized in $\ket{0}$, apply a Hadamard gate to the ancilla, apply controlled-$P_jP_i$ gate, apply $P_j$, the noisy channel $\mathcal{U}_{\mathcal{N}} = \mathcal{U}\circ\mathcal{N}$, and $P_k$ to system qubits, apply controlled-$P_lP_k$ gate, apply a Hadamard gate to the ancilla, and measure the ancilla with the Pauli-$Z$ operator. The average of the measurement results on the ancilla qubit is taken. This operation transforms the noisy channel $\mathcal{U}_\mathcal{N}(\cdot)$ into $(\Theta^{\mathrm{vir}}_{ijkl}(\mathcal{U}_{\mathcal{N}}))(\cdot) = \frac{1}{2}(P_k\mathcal{U}_{\mathcal{N}}(P_i\cdot P_j)P_l + P_l\mathcal{U}_{\mathcal{N}}(P_j\cdot P_i)P_k)$.
        }
        \label{fig_VSCV}
    \end{center}
\end{figure}

Using the virtual supermap $\Theta^{\mathrm{vir}}_{ijkl}$, SCV can be virtually implemented with much lower hardware overhead.
For simplicity, suppose we aim to estimate the expectation value of an observable $O$ for a quantum state $\mathcal{U}(\rho)$ using SCV.
In this case, the error-detected expectation value is given by
\begin{equation}
    \label{eq_SCV_expval}
    \frac{\mathrm{tr}[(\Theta^{\mathrm{det}}_S(\mathcal{U}_\mathcal{N}))(\rho)O]}{\mathrm{tr}[(\Theta^{\mathrm{det}}_S(\mathcal{U}_\mathcal{N}))(\rho)]}.
\end{equation}
The key idea to reduce the hardware overhead in calculating Eq.~\eqref{eq_SCV_expval} is to decompose the projector $\Pi_i$ using Pauli operators as $\Pi_i = \sum_j \pi_{ij} P_j$.
By similarly decomposing the supermap $\Theta^{\mathrm{det}}_S$, we have
\begin{equation}
    \label{eq_VSCV}
    \Theta^{\mathrm{det}}_S = \sum_{ijkl} \alpha_{ijkl}\Theta^{\mathrm{vir}}_{ijkl} = \gamma \sum_{ijkl} \frac{\abs{\alpha_{ijkl}}}{\gamma}\mathrm{sgn}(\alpha_{ijkl})\Theta^{\mathrm{vir}}_{ijkl},
\end{equation}
where $\alpha_{ijkl} = \sum_{ab}\pi_{ai}\pi_{bj}\pi_{ak}\pi_{bl}$, $\gamma = \sum_{ijkl}\abs{\alpha_{ijkl}}$, and $\mathrm{sgn}(\alpha_{ijkl})$ represents the sign of $\alpha_{ijkl}$.
Thus, Eq.~\eqref{eq_SCV_expval} can also be decomposed as
\begin{equation}
    \label{eq_VSCV_expval}
    \frac{\sum_{ijkl} \frac{\abs{\alpha_{ijkl}}}{\gamma}\mathrm{sgn}(\alpha_{ijkl})\mathrm{tr}[(\Theta^{\mathrm{vir}}_{ijkl}(\mathcal{U}_\mathcal{N}))(\rho)O]}{\sum_{ijkl} \frac{\abs{\alpha_{ijkl}}}{\gamma}\mathrm{sgn}(\alpha_{ijkl})\mathrm{tr}[(\Theta^{\mathrm{vir}}_{ijkl}(\mathcal{U}_\mathcal{N}))(\rho)]}.
\end{equation}
This value can be obtained using the circuit in Fig.~\ref{fig_VSCV} with the following protocol:
\begin{enumerate}
    \item For $s = 1,\cdots, N$, repeat the following steps:
    \begin{enumerate}
        \item Sample $i,j,k,l$ with probability $\abs{\alpha_{ijkl}}/\gamma$.
        \item Run the circuit illustrated in Fig.~\ref{fig_VSCV} and measure the system with the observable $O$.
        \item Record the $Z$ measurement of the ancilla as $a_s$ and the product of $a_s$ and the $O$ measurement as $b_s$.
        \item Replace $a_s$ and $b_s$ with $\mathrm{sgn}(\alpha_{ijkl})a_s$ and $\mathrm{sgn}(\alpha_{ijkl})b_s$.
    \end{enumerate}
    \item Calculate $a = \frac{1}{N}\sum_s a_s$ and $b = \frac{1}{N}\sum_s b_s$.
    \item Output $b/a$.
\end{enumerate}
We refer to this simplified SCV protocol for expectation value estimation as \textit{virtual symmetric channel verification} (virtual SCV).

Virtual SCV generalizes virtual symmetry verification methods~\cite{mcclean2020decoding, cai2021quantum, endo2022quantum, tsubouchi2023virtual} to the channel level.
Indeed, virtual SCV can be implemented even when the input state $\rho$ lacks the same symmetry as the unitary $U$.
Moreover, while our discussion focused on implementing virtual SCV just before measurement, virtual SCV can also be applied to parts of the quantum circuit.
Thus, similar to SCV, virtual SCV can be implemented when different quantum operations in the circuit exhibit different symmetric structures.

Compared to SCV, virtual SCV offers significantly lower hardware overhead.
While SCV requires $\lceil \log_2M \rceil$ qubits for the ancilla, virtual SCV requires only a single-qubit ancilla.
Moreover, SCV involves sophisticated quantum operations such as controlled-$V_S$ and quantum Fourier transform, whereas virtual SCV replaces these with controlled-Pauli operations, which are Clifford operations.
In the early fault-tolerant regime, Clifford gates are expected to be easier to implement than non-Clifford gates, as they do not require gate synthesis using Clifford+T gates~\cite{fowler2011constructing, bocharov2012resource, ross2016optimal}.
Thus, virtual SCV is much easier to implement than SCV, relying only on Clifford gates.

Another advantage of virtual SCV is its robustness to noise in the ancilla qubit, even without employing flag qubits.
For example, let us assume that depolarizing noise with error probability $p_{\mathrm{idle}}$ affects the ancilla qubit at any point during computation.
While such noise may propagate to the system qubits, its effect is removed when averaging the measurement results.
Under the depolarizing noise on the ancilla, the virtual channel $\Theta^{\mathrm{vir}}_{ijkl}(\mathcal{U}_{\mathcal{N}})$ is affected as
\begin{equation}
    \label{eq_noisy_virtualchannel}
    \qty(1-\frac{4}{3}p_{\mathrm{idle}})\frac{1}{2}(P_k\mathcal{U}_{\mathcal{N}}(P_i\cdot P_j)P_l + P_l\mathcal{U}_{\mathcal{N}}(P_j\cdot P_i)P_k),
\end{equation}
indicating that the noise on the ancilla adds only a constant factor of $1-\frac{4}{3}p_{\mathrm{idle}}$ to the original channel.
Since this constant factor is removed during state normalization in Eq.~\eqref{eq_VSCV_expval}, the effect of ancilla noise is completely removed.
Furthermore, idling errors on the ancilla qubit during execution of the noisy channel $\mathcal{U}_{\mathcal{N}}$ can be twirled with single-qubit Clifford unitaries into depolarizing noise~\cite{emerson2005scalable, dankert2009exact}.
Thus, the effect of arbitrary noise beyond depolarizing noise can be removed.

Despite these advantages, virtual SCV has certain drawbacks.
The first is that virtual SCV is a virtual protocol and can only be used for expectation value estimation.
However, since expectation value estimation can be used as a subroutine for various tasks such as phase estimation~\cite{suzuki2022quantum, dutkiewicz2024error}, virtual SCV remains applicable to many early fault-tolerant algorithms.
Another drawback is its sampling overhead.
The leading factor of the sampling overhead for virtual protocols is the square of the inverse of the denominator in Eq.~\eqref{eq_VSCV_expval}~\cite{cai2021quantum}.
This overhead scales as $\gamma^2\mathrm{tr}[(\Theta^{\mathrm{det}}_S(\mathcal{U}_\mathcal{N}))(\rho)]^{-2} \sim \gamma^2(1+p)^2$ when symmetry-breaking errors occur with probability $p$.
For $L$ noisy channels $\mathcal{U}_{\mathcal{N}}$, applying virtual SCV to each channel results in sampling overhead scaling as $O((1+p)^{2L})$, which is quadratically larger than SCV.
If the ancilla qubit is affected by depolarizing noise with probability $p_{\mathrm{idle}}$, the sampling overhead further increases by a factor of $\left(1-\frac{4}{3}p_{\mathrm{idle}}\right)^2$, according to Eq.~\eqref{eq_noisy_virtualchannel}.
Moreover, when applying virtual SCV for the particle number operator $S = \sum_i Z_i$, $\gamma$ grows exponentially with the number of system qubits $n$ and becomes intractable.
Nevertheless, for many practical use cases stated below, $\gamma$ is equal to 1, so it does not contribute to an increase in the sampling overhead.

\subsection{Application to Pauli symmetry}
\label{sec_VSCV_Pauli}
As discussed in Sec.~\ref{sec_SCV_Pauli}, the SCV gadget for a single Pauli operator $Q$ requires only a single-qubit ancilla and two controlled-$Q$ gates, whose hardware overhead is comparable with virtual SCV.
Therefore, when implementing SCV for a single Pauli operator, there is no advantage in employing virtual SCV (except for the fact that virtual SCV is robust against the noise in the ancilla).
However, the SCV gadget designed for a set of commuting Pauli operators $\mathcal{Q}_{U}^{\mathrm{com}}$, which concatenates multiple SCV gadgets for different Pauli operators, requires $r$ ancilla qubits and $2r$ controlled-Pauli gates as depicted in Fig.~\ref{fig_SCV_Pauli}.
In this case, virtual SCV can significantly reduce hardware overhead.

The SCV gadget $\Theta^{\mathrm{det}}_{\mathcal{Q}^{\mathrm{com}}_{U}}$ maps the noisy channel $\mathcal{U}_\mathcal{N}$ to
\begin{equation}
    \label{eq_SCV_Pauli_for_VSCV}
    \begin{aligned}
        &\sum_{i_1j_1\cdots i_rj_r}  \Pi_{i_1\cdots i_r}\mathcal{U}_{\mathcal{N}}(\Pi_{i_1\cdots i_r}^\dag  \cdot  \Pi_{j_1\cdots j_r})\Pi_{j_1\cdots j_r}^\dag\\
        = &\frac{1}{\abs{\mathcal{Q}_U^{\mathrm{com}}}^2}\sum_{Q_i,Q_j\in\mathcal{Q}_U^{\mathrm{com}}}Q_i\mathcal{U}_{\mathcal{N}}(Q_i\cdot Q_j)Q_j,
    \end{aligned}
\end{equation}
where $\Pi_{i_1\cdots i_r} = \Pi_{i_1}^1 \cdots \Pi_{i_r}^r$ with $\Pi_{i_k}^k = \frac{1}{2}(I+(-1)^{i_k}Q_k)$ being the projector onto the eigenspace of $Q_k$ corresponding to the eigenvalue $(-1)^{i_k}$.
This means that the SCV gadget $\Theta^{\mathrm{det}}_{\mathcal{Q}^{\mathrm{com}}_{U}}$ can be represented as
\begin{equation}
    \Theta^{\mathrm{det}}_{\mathcal{Q}^{\mathrm{com}}_{U}} = \frac{1}{\abs{\mathcal{Q}_U^{\mathrm{com}}}^2}\sum_{Q_i,Q_j\in\mathcal{Q}_U^{\mathrm{com}}}\Theta^{\mathrm{vir}}_{ijij},
\end{equation}
where $\Theta^{\mathrm{vir}}_{ijij}$ is the virtual supermap depicted in Fig.~\ref{fig_VSCV} with $P_i = Q_i, P_j =Q_j, P_k=Q_i, P_l=Q_j$.
Therefore, we have $\gamma = 1$, and virtual SCV can be implemented by randomly and uniformly sampling $Q_i,Q_j\in\mathcal{Q}_U^{\mathrm{com}}$.
This approach reduces the number of ancilla qubits from $r$ to $1$ and the number of controlled-Pauli gates from $2r$ to $2$, and the ancilla becomes robust against its noise even without the use of flag qubits.

\subsection{Mitigating errors on idling qubits with virtual symmetric channel purification}
\label{sec_VSCV_idle}
One of the powerful applications of virtual SCV is the mitigation of errors on idling qubits.
Let us consider the case where the ideal channel is the identity channel $\mathcal{I}$ for a single qubit ($n=1$), which is affected by an idling error $\mathcal{N}^{\circ L}$, where $\mathcal{N}$ represents the noise per time step and $L$ represents the idle time.
Since arbitrary error anti-commutes with at least one element of $\mathcal{Q}_I^{\mathrm{com}} = \qty{I, X, Y, Z}$, we can apply the SCV gadget $\Theta^{\mathrm{det}}_{\mathcal{Q}^{\mathrm{com}}_{I}}$ to completely detect the idling error $\mathcal{N}^{\circ L}$.
However, as depicted in Fig.~\ref{fig_idle} (a), the ancilla qubit used in the SCV gadget $\Theta^{\mathrm{det}}_{\mathcal{Q}^{\mathrm{com}}_{I}}$ is also affected by the same idling error $\mathcal{N}^{\circ L}$, which degrades the performance of SCV.

To overcome this limitation, we employ virtual SCV to mitigate the idling error $\mathcal{N}^{\circ L}$.
While the same noise $\mathcal{N}$ affects the ancilla qubit, virtual SCV is robust against depolarizing noise in the ancilla.
Therefore, when we assume that $\mathcal{N}$ is depolarizing noise, almost all the noise in the virtual SCV gadget does not impact the error-mitigation performance (see Fig.~\ref{fig_idle} (b)).
The only remaining errors are those on the system qubit before and after the controlled Pauli gate.
Figure~\ref{fig_idle} (c) illustrates the performance difference in (virtually) detecting errors between SCV and virtual SCV.
We observe that while the error reduction is limited in SCV, virtual SCV significantly reduces the impact of idling errors.

\begin{figure}[t]
    \begin{center}
        \includegraphics[width=0.99\linewidth]{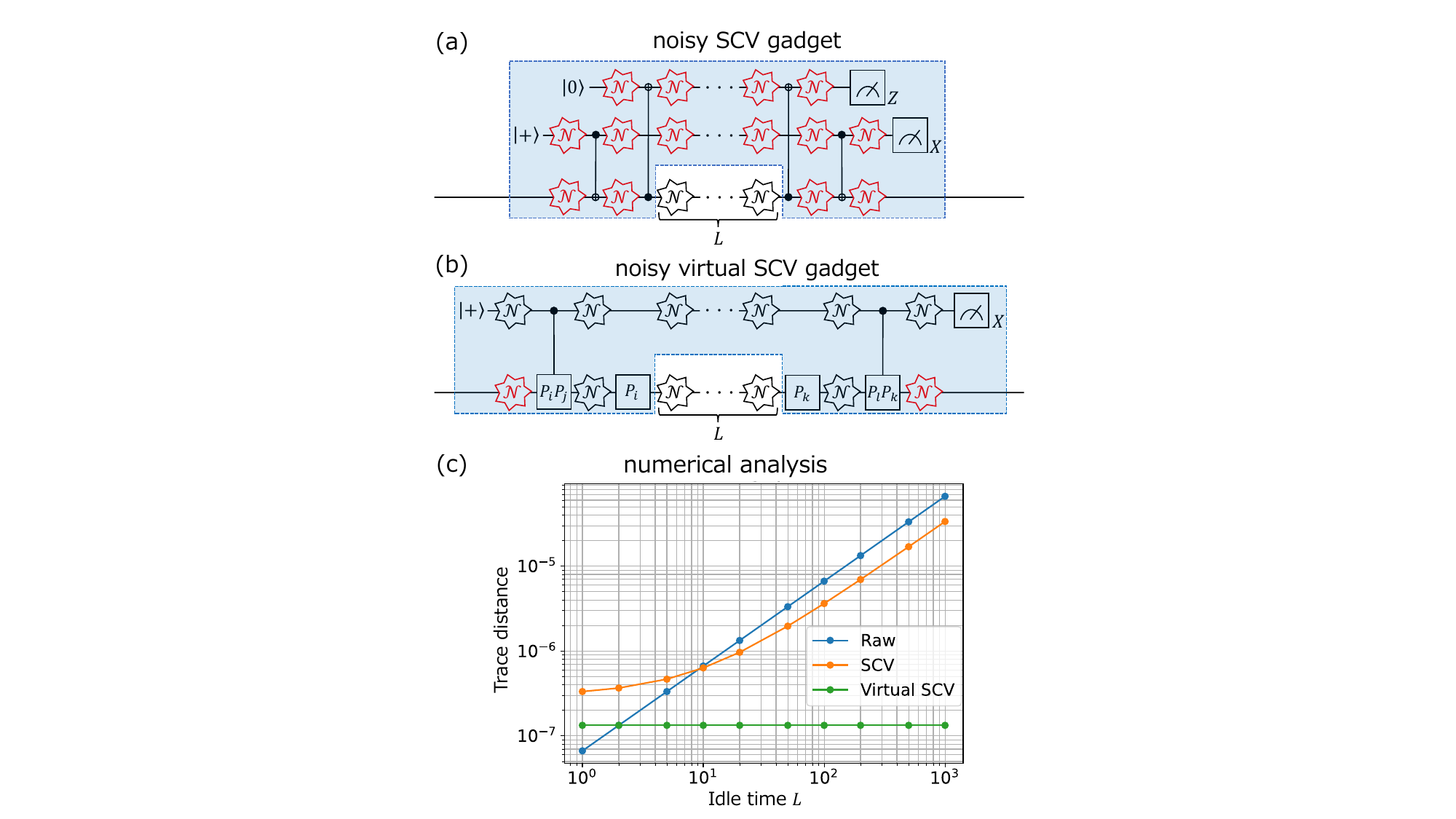}
        \caption{Schematic illustration of (virtually) detecting idling errors $\mathcal{N}^{\circ L}$ using (a) a noisy SCV gadget and (b) a noisy virtual SCV gadget. The noise channel depicted in red degrades the performance of (virtual) error detection, while the noise channel depicted in black does not affect performance. Here, we assume that the noise channel $\mathcal{N}$ is depolarizing noise. Panel (c) shows the trace distance between the noiseless and noisy quantum states affected by the idling error $\mathcal{N}^{\circ L}$, where the noisy SCV and virtual SCV gadgets in panels (a) and (b) are used to suppress the effect of noise. The error rate of the depolarizing noise is set to $p_{\mathrm{err}} = 10^{-7}$.}
        \label{fig_idle}
    \end{center}
\end{figure}

As a practical application of virtual SCV in mitigating idling errors, we consider its use in suppressing errors in idling system qubits during \texttt{SELECT} operations.
\texttt{SELECT} operations are widely used in various quantum algorithms, such as phase estimation and Hamiltonian simulation based on qubitization~\cite{low2017optimal, babbush2018encoding, low2019hamiltonian, yoshioka2024hunting}. 
In \texttt{SELECT} operations, quantum gates are applied sparsely to system qubits, leading to long idle times.
Thus, virtual SCV can be employed to mitigate such idling errors in system qubits.

As a demonstration, we consider the \texttt{SELECT} operation of the two-dimensional (2D) Fermi-Hubbard model on a square lattice with linear system size of $M$ and system qubit count $n = 2M^2$, as discussed in Ref.~\cite{babbush2018encoding}.
We provide the Hamiltonian and the circuit structure of the corresponding \texttt{SELECT} operation in Appendix~\ref{sec_select}.
We assume that all quantum gates and idling qubits are subject to local depolarizing noise with an error rate of $p_{\mathrm{err}}$, and calculate the total error rate of the \texttt{SELECT} operation with and without virtual SCV in Fig.~\ref{fig_select}.
We set $p_{\mathrm{err}} = 10^{-7}$, which is 100 times as small as the value used in the numerical simulations of Sec.~\ref{sec_SCV}.  
This choice is motivated by the fact that the error rate in Sec.~\ref{sec_SCV} corresponds to that of synthesized gates, for which we assumed that roughly 100 basic gates are required to implement each.
Without virtual SCV, \texttt{SELECT} operation has $O(n)$ depth, so each system qubit experiences $O(n)$ idling errors, and the total error rate scales as $O(n^2)$.
In contrast, when virtual SCV is applied, these errors are mitigated, and only $O(1)$ errors per system qubit remain.
Consequently, the dominant error source becomes the $O(n\log n)$ errors on the $O(\log n)$ ancilla qubits required for qubitization.
Therefore, as illustrated in Fig.~\ref{fig_select}, applying virtual SCV to idling system qubits enables a quadratic reduction in the total error rate.

Notably, this reduction is not limited to this specific model but can be generalized to arbitrary Fermi-Hubbard models with short-range 2-body interactions under the Jordan-Wigner transformation and local fermionic encodings, as well as lattice spin models with short-range interactions.
For instance, consider a local Hamiltonian with $L=\mathrm{poly}(n)$ terms, where each qubit participates in $O(L/n)$ terms.  
In this case, the \texttt{SELECT} operation requires $O(\log n)$ ancilla qubits and has depth $O(L)$.
Meanwhile, each qubit participates in only $O(L/n)$ gates, meaning that many system qubits remain idle for extended periods, leading to numerous idling errors.
By mitigating these idling errors using virtual SCV, the total error rate can be reduced from $O(Ln)$ to $O(L\log n)$,  where the dominant contribution comes from errors in the $O(\log n)$ ancilla qubits.

Additionally, the same ancilla qubit can be used to perform virtual SCV on all $n$ system qubits.
Thus, only a single additional ancilla qubit is required to implement virtual SCV, regardless of the number of system qubits $n$.
We discuss this point in detail in Appendix~\ref{sec_select}.

\begin{figure}[t]
    \begin{center}
        \includegraphics[width=0.85\linewidth]{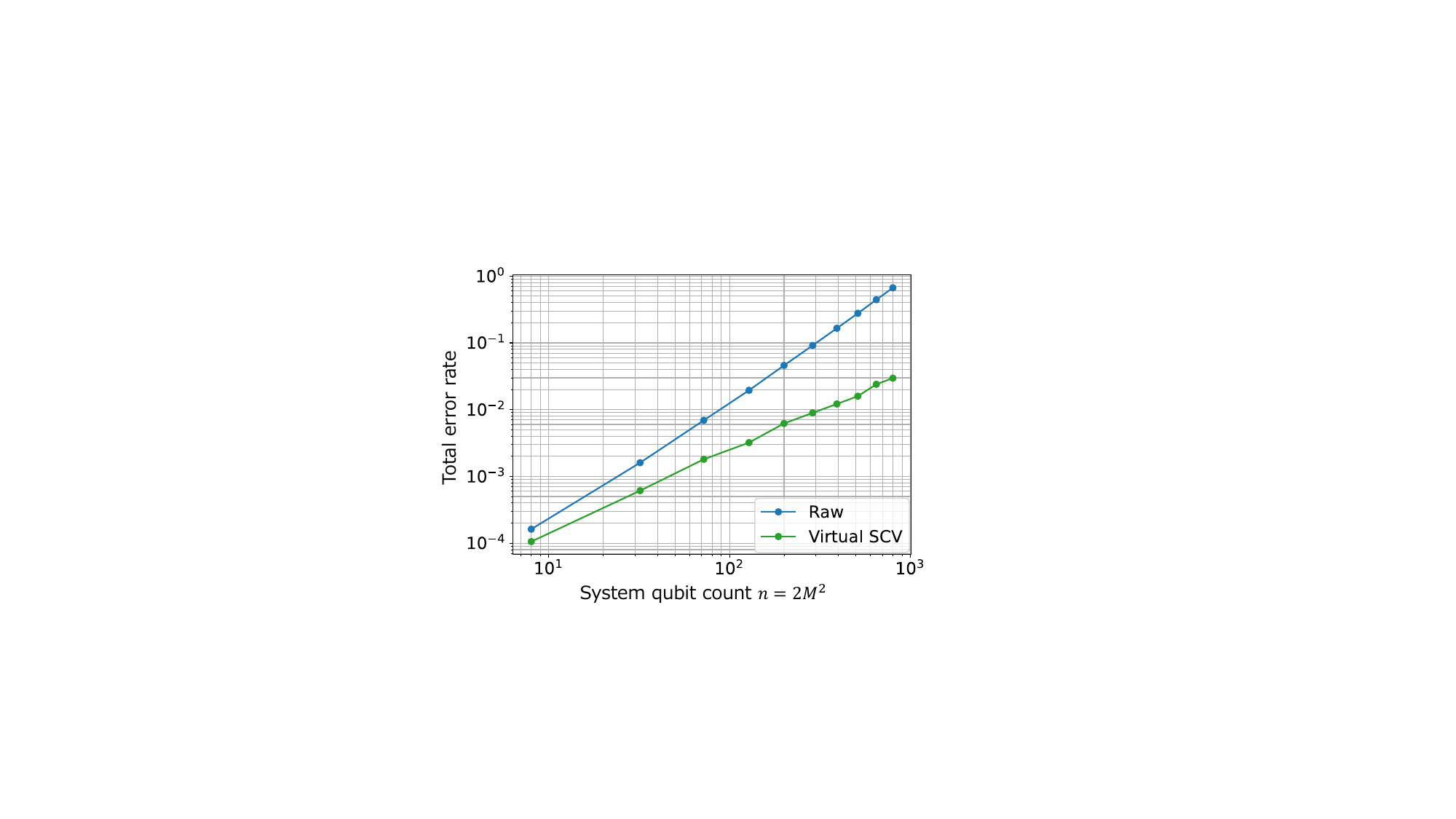}
        \caption{The total error rate of the \texttt{SELECT} operation of the 2D Fermi-Hubbard model on a square lattice with linear system size of $M$ and system qubit count $n = 2M^2$, proposed in Ref.~\cite{babbush2018encoding}. The ``Raw'' and ``Virtual SCV'' lines represent the results without and with virtual SCV, respectively. We assume that all quantum gates and idling qubits are subject to local depolarizing noise with an error rate of $p_{\mathrm{err}} = 10^{-7}$.}
        \label{fig_select}
    \end{center}
\end{figure}

\section{Correcting errors using symmetric channel verification}
\label{sec_SCV_correct}

\subsection{General framework}
Since SCV post-selects the measurement results on the ancilla qubit, it probabilistically provides a purified channel, resulting in an additional sampling overhead of $O((1+p)^{L})$ to detect the effect of noise.
Furthermore, virtual SCV requires a quadratically larger sampling overhead of $O((1+p)^{2L})$ for its implementation.
While these methods are effective when the total error rate $pL$ is small, they become inefficient for larger error rates due to the exponential growth in the sampling overhead.

To reduce the sampling overhead, feedback controls can be applied to the system qubits based on the measurement results from the ancilla qubits, instead of post-selecting them.
This approach allows \textit{correction} of symmetry-breaking errors on noisy channels without introducing additional sampling overhead.
As discussed in Sec.~\ref{sec_channel_purification}, feedback controls can be implemented using a superchannel $\Theta^{\mathrm{cor}}$ defined in Eq.~\eqref{eq_purification_correction} by adding controlled operation corresponding to the feedback control into the unitary channel $\mathcal{U}_{\mathrm{D}}$.
Thus, applying feedback control to the SCV gadget $\Theta^{\mathrm{det}}_S$ in Fig.~\ref{fig_SCV} is equivalent to inserting an additional unitary operator before the measurement and discarding the measurement results.
The performance of such deterministic channel purification protocols is summarized in the following theorem:
\begin{thm}
    \label{thm_3'}
    Let $\Theta^{\mathrm{cor}}_S$ denote a quantum superchannel where an additional unitary operator is inserted before the measurement in the circuit depicted in Fig.~\ref{fig_SCV}, and the measurement results are discarded.
    Then, $\Theta^{\mathrm{cor}}_S$ can correct the noise $\mathcal{N}(\cdot) = \sum_{j=0}^K N_j\cdot N_j^\dag$ on the noisy channel $\mathcal{U}_{\mathcal{N}}$, i.e., $\Theta^{\mathrm{cor}}_S(\mathcal{U}_{\mathcal{N}}) = \mathcal{U}$, if and only if for all $j,k\in\{0,\ldots,K\}$
    \begin{equation}
        \sum_i \Pi_i N_j^\dag N_k \Pi_i \propto I_n.
    \end{equation}
\end{thm}

Theorem~\ref{thm_3'} can be derived by interpreting the initial half of the SCV gadget $\Theta^{\mathrm{det}}_S$ as an encoding process on a quantum error-correcting code.
By applying the quantum error-correction condition~\cite{nielsen2010quantum, gottesman2016surviving}, we arrive at Theorem~\ref{thm_3'}.
For a detailed proof, we refer the reader to Appendix~\ref{sec_proof}.

\subsection{Application to Pauli symmetry}
As an application of SCV for error correction, consider a quantum channel under Pauli symmetry. Following the setup in Sec.~\ref{sec_SCV_Pauli}, we analyze the case where the noise channel $\mathcal{N}(\cdot)=\sum_{i=0}^KP_i \cdot P_i$ is Pauli noise as represented in Eq.~\eqref{eq_Pauli_noise}.
By applying feedback controls based on the measurement results from the SCV gadget for $\mathcal{Q}_{U}^{\mathrm{com}}$ in Fig.~\ref{fig_SCV_Pauli}, Pauli errors can be corrected to some extent.
The performance of such deterministic channel purification protocols under Pauli symmetry is characterized by the following theorem:
\begin{cor}
    \label{cor_2}
    Let $\Theta^{\mathrm{cor}}_{\mathcal{Q}^{\mathrm{com}}_{U}}$ denote a quantum superchannel where feedback controls using Pauli operators are applied according to the measurement results from the SCV gadget for $\mathcal{Q}_{U}^{\mathrm{com}}$ in Fig.~\ref{fig_SCV_Pauli}. Then, $\Theta^{\mathrm{cor}}_{\mathcal{Q}^{\mathrm{com}}_{U}}$ can correct the Pauli noise Eq.~\eqref{eq_Pauli_noise} on the noisy channel $\mathcal{U}_{\mathcal{N}}$, i.e., $\Theta^{\mathrm{cor}}_{\mathcal{Q}^{\mathrm{com}}_{U}}(\mathcal{U}_{\mathcal{N}}) = \mathcal{U}$, if and only if
    \begin{equation}
        \label{eq_cor_2_1}
        \mathcal{P}^{\mathrm{{cor}}}_{\mathcal{N}} \cap \mathcal{Q}_{U} = \emptyset,
    \end{equation}
    where $\mathcal{P}^{\mathrm{{cor}}}_{\mathcal{N}} = \{\pm1,\pm i\}\cdot\{P_iP_j\}_{i,j=0}^K \backslash \{I_n\}$ and $\mathcal{Q}_{U}$ is defined in Eq.~\eqref{eq_Q_U}.
\end{cor}

The proof of Corollary~\ref{cor_2}, which follows from Theorem~\ref{thm_3'}, is provided in Appendix~\ref{sec_select}.
We note that, while Corollary~\ref{cor_2} assumes Pauli noise, it can be generalized to arbitrary noise channels $\mathcal{N}(\cdot) = \sum_j N_j \cdot N_j^\dag$ by decomposing each Kraus operator $N_j$ into a sum of Pauli operators.
Specifically, the noisy channel can be purified as $\Theta^{\mathrm{cor}}_{\mathcal{Q}^{\mathrm{com}}_{U}}(\mathcal{U}_\mathcal{N}) = \mathcal{U}$, provided that the set $\mathcal{P}^{\mathrm{{cor}}}_{\mathcal{N}} = \{P\in\mathcal{P}_n-\{I_n\}~|~\exists ij, \mathrm{tr}[N_i^\dag N_j P]\neq 0\}$ satisfies the conditions in Corollary~\ref{cor_1}.
This is because if a set of Pauli operators $\mathcal{P}^{\mathrm{cor}}_{\mathcal{N}}$ satisfies the conditions in Corollary~\ref{cor_2}, and hence also satisfies the condition in Theorem~\ref{thm_3'}, then any linear combination of these operators also satisfies the condition in Theorem~\ref{thm_3'}.

Corollary~\ref{cor_2} implies that if $\mathcal{Q}_{U}$ consists of Pauli operators with weights greater than or equal to $d$, arbitrary $(d-1)/2$-qubit errors can be corrected.
Meanwhile, when focusing on detecting errors, Corollary~\ref{cor_1} indicates that arbitrary $(d-1)$-qubit errors can be detected.
This result is consistent with state-level error correction and detection protocols.
An interesting application is the time evolution circuit of the lattice fermion model using the Majorana loop stabilizer code~\cite{jiang2019majorana}.
In this case, we have $d=3$, meaning that arbitrary single-qubit errors can be corrected.
Another application is the time evolution $U = e^{i\theta H}$ of the Hamiltonian $H = -\sum_{S\in\mathcal{S}} a_SS$, where $\mathcal{S}$ is the set of stabilizers of a stabilizer code with distance $d$ and $a_S\in\mathbb{R}$ is a coefficient for the stabilizer $S$.
There are several known examples whose low-energy effective Hamiltonians take this form.
One notable case is the honeycomb Kitaev model, whose effective Hamiltonian corresponds to that of a topological code; its dynamics have recently been investigated in Refs.~\cite{gohlke2017dynamics, ali2025robust}.  
Furthermore, there is surging interest in Hamiltonians based on quantum LDPC codes, which have been proposed as a platform for observing topological quantum spin glass phases~\cite{placke2024topological}.
In such cases, when the input state to the unitary is within the code space of the stabilizer code, the usual state-level error correction can be applied to correct errors.
However, if the input state spans multiple eigenspaces, syndrome measurement may destroy the state, so the usual state-level error correction cannot be applied.
Even in this case, assuming that the SCV gadget $\Theta^{\mathrm{cor}}_{\mathcal{Q}^{\mathrm{com}}_{U}}$ can be applied without additional errors, arbitrary $(d-1)/2$-qubit errors can be corrected using the SCV gadget $\Theta^{\mathrm{cor}}_{\mathcal{Q}^{\mathrm{com}}_{U}}$.

\begin{figure}[t]
    \begin{center}
        \includegraphics[width=0.85\linewidth]{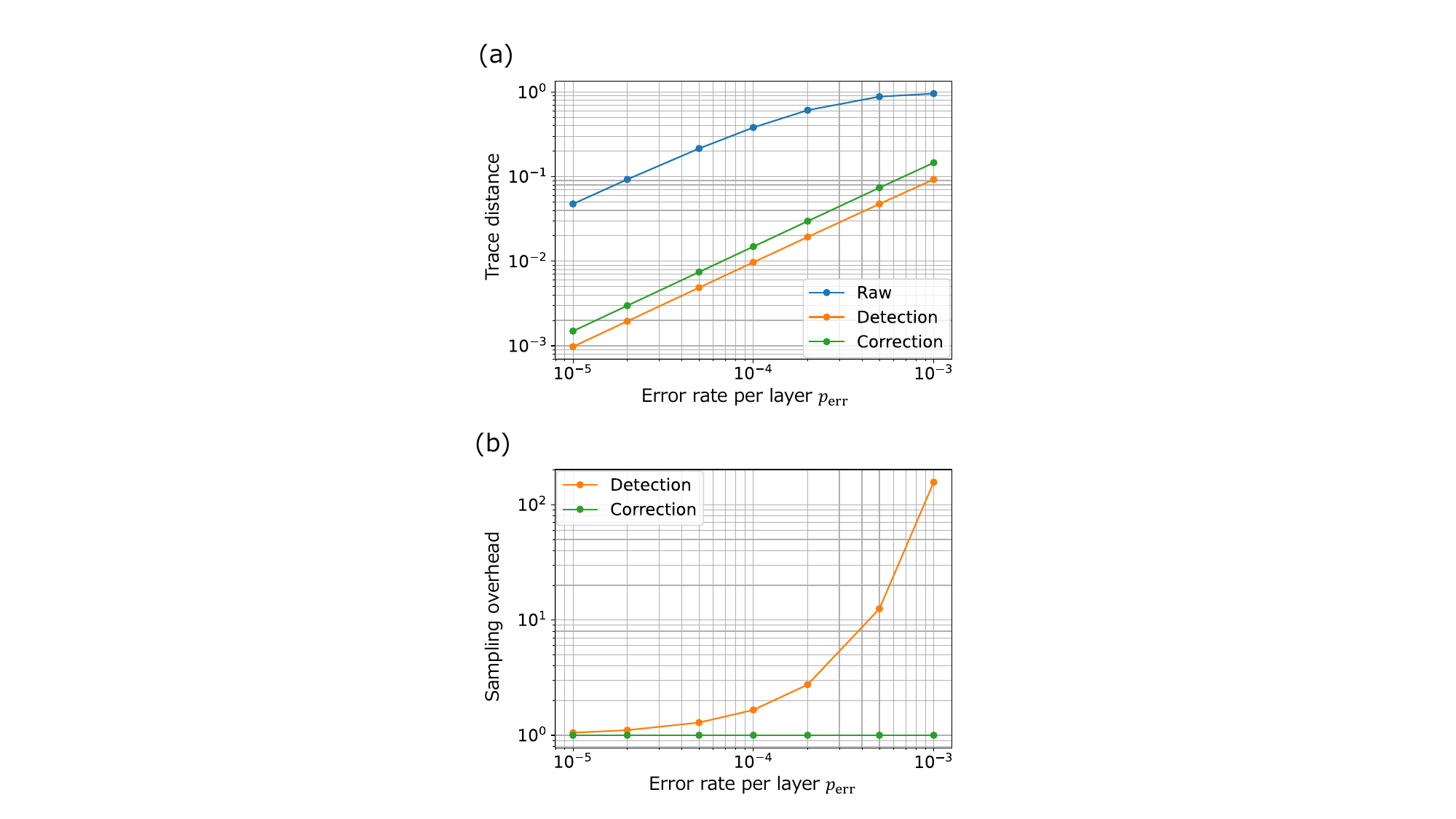}
        \caption{Performance of SCV in detecting and correcting errors in Hamiltonian simulation circuits $U = e^{-i\theta H}$, where the Hamiltonian is given by $H = -\sum_{S\in\mathcal{S}} S$ and $\mathcal{S}$ is the set of stabilizers of the $[[5,1,3]]$ code. Panel (a) shows the trace distance between the ideal quantum state and the noisy quantum state, purified by the error-detection gadget $\Theta^{\mathrm{det}}_{\mathcal{Q}^{\mathrm{com}}_{U}}$ or the error-correction gadget $\Theta^{\mathrm{cor}}_{\mathcal{Q}^{\mathrm{com}}_{U}}$. Panel (b) shows the sampling overhead required for purification.}
        \label{fig_numerics_513code}
    \end{center}
\end{figure}

To evaluate the performance of SCV in correcting errors under noise in the SCV gadget $\Theta^{\mathrm{cor}}_{\mathcal{Q}^{\mathrm{com}}_{U}}$, we consider the time evolution $U = e^{-i\theta H}$ of the Hamiltonian $H = -\sum_{S\in\mathcal{S}} S$, where $\mathcal{S}$ is the set of stabilizers of the $[[5,1,3]]$ code.
We model the noisy circuit as $\bigcirc_{l=1}^L \mathcal{N}_l \circ \mathcal{U}_l$, where $\mathcal{N}_l$ represents local depolarizing noise with an error rate $p_{\mathrm{err}}$, and $\mathcal{U}_l(\cdot) = e^{i(\theta/L) H} \cdot e^{-i(\theta/L) H}$ corresponds to a single time step.
Noisy error-detection gadget $\Theta^{\mathrm{det}}_{\mathcal{Q}^{\mathrm{com}}_{U}}$ or error-correction gadget $\Theta^{\mathrm{cor}}_{\mathcal{Q}^{\mathrm{com}}_{U}}$, where $U_{\mathrm{E}}$ and $U_{\mathrm{D}}$ composing the gadgets are affected by local depolarizing noise with an error rate of $p_{\mathrm{err}}/100$, are applied to each noisy layer $\mathcal{N}_l \circ \mathcal{U}_l$.
The difference in error rates arises because approximately 100 Clifford+T gates are required to synthesize Pauli rotation gates with an accuracy of $10^{-6}$~\cite{ross2016optimal}.
We neglect idling errors on the ancilla by considering the case where logical qubits with a higher code distance are used for the ancilla.
Even when ancilla noise is not negligible, we expect that the use of flag qubits can remove the effect of such errors, as in state-level error correction~\cite{chao2018quantum, chamberland2018flag, chao2020flag}.  
Constructing a channel-level fault-tolerant error-correction scheme using flag qubits that is robust against ancilla noise is left as future work.
We set $\theta = 2\pi$ and $L = 1000$, and use a Haar-random state as the input state.

Our results are shown in Fig.~\ref{fig_numerics_513code}.
When noise is present in the SCV gadgets, performing error detection or correction leads to a constant reduction in the error rate.
This is because, although the SCV gadgets can detect or correct arbitrary single-qubit errors on system qubits, the noise in the SCV gadgets dominates overall performance.
While we expect that the effect of such noise can be removed by extending fault-tolerant error-correction protocols to the channel level, we leave this analysis as future work.
Comparing the performance of the error-detection gadget $\Theta^{\mathrm{det}}_{\mathcal{Q}^{\mathrm{com}}_{U}}$ and the error-correction gadget $\Theta^{\mathrm{cor}}_{\mathcal{Q}^{\mathrm{com}}_{U}}$, we observe that the error-detection gadget achieves higher accuracy.
However, it incurs significant sampling overhead, whereas the error-correction gadget $\Theta^{\mathrm{cor}}_{\mathcal{Q}^{\mathrm{com}}_{U}}$ can correct errors without additional sampling overhead.

\section{Limitations of channel purification using Clifford unitaries}
\label{sec_Clifford_purification}
In previous sections, we did not explicitly impose restrictions on the operations used for the task of channel purification. From a practical viewpoint in the early fault-tolerant regime, one of the most important restrictions is to use only Clifford gates so that the impact of errors within the purification gadget is minimized.
This is because Clifford gates are expected to be easier to implement, especially compared to arbitrary angle Pauli rotations which require gate synthesis~\cite{fowler2011constructing, bocharov2012resource, ross2016optimal}.

In this section, we investigate the fundamental limitations of purifying noisy quantum channels using only Clifford unitaries for the purification gadget.
In particular, we focus on the purification of a Hamiltonian simulation unitary, given by $\mathcal{U}(\cdot) = e^{i\theta H}\cdot e^{-i\theta H}$, when it is subject to Pauli noise $\mathcal{N} = \sum_{i=0}^{K} p_i P_i\cdot P_i$ as defined in Eq.~\eqref{eq_Pauli_noise}.
We note that logical errors arising from imperfect decoding~\cite{beale2018quantum, huang2019performance}, imperfect magic state~\cite{cai2023quantum}, and imperfect gate synthesis~\cite{yoshioka2024error} can all be tailored to be Pauli noise in the early fault-tolerant regime.
We derive the necessary and sufficient conditions for the noise to be completely {\it detectable} and {\it correctable} using Clifford unitaries, formally presented as Theorem~\ref{thm_3} and~\ref{thm_4}, respectively.
We further demonstrate that these conditions coincide with those required for the noise to be both completely detectable and correctable using SCV under Pauli symmetry.
When such conditions are not satisfied, the SCV cannot eliminate the entire effect of noise.
Nonetheless, as presented in Theorem~\ref{thm_5}, we can derive an upper bound on the fidelity of the purified channel, and also show that this bound is saturated by SCV under Pauli symmetry.
These results demonstrate that SCV under Pauli symmetry provides the best channel purification protocol under the constraint of Clifford unitaries.

We first consider the fundamental limitations in detecting errors in the noisy channel $\mathcal{U}_{\mathcal{N}} = \mathcal{N}\circ \mathcal{U}$, when only Clifford gates are used for the entangling operations between the system and ancilla.
As discussed in Sec.~\ref{sec_channel_purification}, detecting errors corresponds to applying a supermap $\Theta^{\mathrm{det}}$, defined in Eq.~\eqref{eq_purification_detection}.
It is convenient to recall that $\Theta^{\rm det}$ transforms the noisy channel $\mathcal{U}_{\mathcal{N}}$ into a trace non-increasing map $\Theta^{\mathrm{det}}(\mathcal{U}_{\mathcal{N}})$ whose action is given as
\begin{equation}
    \label{eq_purification_detection_2}
    \rho \mapsto \bra{0^m}\mathcal{U}_{\mathrm{D}}\circ(\mathcal{U}_{\mathcal{N}}\otimes\mathcal{I}_{\mathrm{a}})\circ\mathcal{U}_{\mathrm{E}}(\rho\otimes\ketbra{0^m})\ket{0^m}.
\end{equation}
When we are restricted to using only Clifford gates for error detection, the unitaries $\mathcal{U}_{\mathrm{E}}(\cdot) = U_{\mathrm{E}} \cdot U_{\mathrm{E}}^\dag$ and $\mathcal{U}_{\mathrm{D}}(\cdot) = U_{\mathrm{D}} \cdot U_{\mathrm{D}}^\dag$, which construct the supermap $\Theta^{\mathrm{det}}$, are constrained to $(n+m)$-qubit Clifford unitaries.
Let us denote such a supermap consisting of Clifford unitaries as $\Theta_{\mathrm{Clif}}^{\mathrm{det}}$.
For the error-detection supermap $\Theta_{\mathrm{Clif}}^{\mathrm{det}}$, we can identify the detectable noise via the following theorem.

\begin{thm}
    \label{thm_3}
    Let $\mathcal{U}(\cdot) = e^{i\theta H}\cdot e^{-i\theta H}$ be an $n$-qubit unitary channel and $\mathcal{N}(\cdot) = \sum_{i=0}^K p_i P_i\cdot P_i$ be a Pauli noise defined in Eq.~\eqref{eq_Pauli_noise}.
    Then, there exists an error-detection supermap $\Theta_{\mathrm{Clif}}^{\mathrm{det}}$ that can detect the set of errors $\mathcal{P}_{\mathcal{N}}^{\mathrm{det}} = \qty{P_i}_{i=1}^K$ as 
    \begin{equation}
        \label{eq_thm_3_1}
        \forall \theta, \; \Theta^{\mathrm{det}}_{\mathrm{Clif}} (\mathcal{N}_{P_i}\circ\mathcal{U}) = \delta_{i0}\mathcal{U},
    \end{equation}
    if and only if
    \begin{equation}
        \label{eq_thm_3_2}
        \mathcal{P}^{\mathrm{{det}}}_{\mathcal{N}} \cap \mathcal{Q}_{H} = \emptyset.
    \end{equation}
    Here, $\mathcal{N}_{P_i}(\cdot) = P_i\cdot P_i$, $\delta_{ij}$ represents the Kronecker delta, and $\mathcal{Q}_{H}$ is a set of Pauli operators generated from Pauli operators with non-zero coefficients in the Pauli expansion of $H$, as defined in Eq.~\eqref{eq_Q_U}.
\end{thm}

The key idea for proving the necessary condition of Theorem~\ref{thm_3} is to interpret the application of the Clifford unitary $U_{\mathrm{E}}$ as encoding into a stabilizer code with a code space spanned by $U_{\mathrm{E}}\ket{\psi}\ket{0^m}$.
Under Eq.~\eqref{eq_thm_3_1}, we can show that all elements except $I^{\otimes n+m}$ in $\mathcal{Q}_{H} \otimes \{I\}^{\otimes m}$ are non-trivial logical Pauli operators.
Since detectable errors cannot be non-trivial logical Pauli operators, we obtain Eq.~\eqref{eq_thm_3_2}.
To show sufficiency, we consider detecting noise with SCV under Pauli symmetry, which we have discussed in Sec.~\ref{sec_SCV_Pauli}.
Interestingly, the necessary condition in Eq.~\eqref{eq_thm_3_2} exactly matches the condition for detectable noise using SCV under Pauli symmetry, which is characterized in Corollary~\ref{cor_1}.
This implies that when error detection for the noisy channel is restricted to Clifford operations, we cannot detect more noise than using SCV under Pauli symmetry.
Therefore, it suffices to use SCV under the restriction, demonstrating the optimality of our protocol.
See Appendix~\ref{sec_proof} for the full proof.

Let us comment on the necessary and sufficient condition in Eq.~\eqref{eq_thm_3_2}.
This condition implies that Pauli errors included in the Hamiltonian cannot be detected; such errors are indistinguishable from the unitary $e^{i\theta H}$.
Moreover, Eq.~\eqref{eq_thm_3_2} implies that, if $\mathcal{Q}_{H}$ consists of Pauli operators with weights greater than or equal to $d$,  SCV can detect arbitrary $(d-1)$-qubit errors.
A notable example is the Heisenberg model discussed in Sec.~\ref{sec_SCV_Pauli}.
For this Hamiltonian, $\mathcal{Q}_{H}$ consists of weight-2 or higher Pauli operators, allowing the detection of arbitrary single-qubit errors.

Next, we consider the limitations in correcting errors in the noisy channel $\mathcal{U}_{\mathcal{N}}$ using Clifford gates.
As discussed in Sec.~\ref{sec_channel_purification}, correcting errors corresponds to applying a superchannel $\Theta^{\mathrm{cor}}$, defined in Eq.~\eqref{eq_purification_correction}, which transforms the noisy channel $\mathcal{U}_{\mathcal{N}}$ into a quantum channel $\Theta^{\mathrm{cor}}(\mathcal{U}_{\mathcal{N}})$, defined as
\begin{equation}
    \label{eq_purification_correction_2}
    \rho \mapsto \mathrm{tr}_{\mathrm{a}}[\mathcal{U}_{\mathrm{D}}\circ(\mathcal{U}_{\mathcal{N}}\otimes\mathcal{I}_{\mathrm{a}})\circ\mathcal{U}_{\mathrm{E}}(\rho\otimes\ketbra{0^m})].
\end{equation}
When we are restricted to using only Clifford gates for error correction, the unitaries $\mathcal{U}_{\mathrm{E}}(\cdot) = U_{\mathrm{E}} \cdot U_{\mathrm{E}}^\dag$ and $\mathcal{U}_{\mathrm{D}}(\cdot) = U_{\mathrm{D}} \cdot U_{\mathrm{D}}^\dag$, which construct the superchannel $\Theta^{\mathrm{cor}}$, are constrained to $(n+m)$-qubit Clifford unitaries.
Let us denote such a superchannel consisting of Clifford unitaries as $\Theta_{\mathrm{Clif}}^{\mathrm{cor}}$.
For the error-correction superchannel $\Theta_{\mathrm{Clif}}^{\mathrm{cor}}$, the correctable noise can be characterized via the following theorem.

\begin{thm}
    \label{thm_4}
    Let $\mathcal{U}(\cdot) = e^{i\theta H}\cdot e^{-i\theta H}$ be an $n$-qubit unitary channel and $\mathcal{N}(\cdot) = \sum_{i=0}^K p_i P_i\cdot P_i$ be a Pauli noise defined in Eq.~\eqref{eq_Pauli_noise}.
    Then, there exists an error-correction superchannel $\Theta_{\mathrm{Clif}}^{\mathrm{cor}}$ that satisfies
    \begin{equation}
        \label{eq_thm_4_1}
        \forall \theta, \; \Theta_{\mathrm{Clif}}^{\mathrm{cor}} (\mathcal{U}_\mathcal{N}) = \mathcal{U},
    \end{equation}
    if and only if
    \begin{equation}
        \label{eq_thm_4_2}
        \mathcal{P}^{\mathrm{{cor}}}_{\mathcal{N}} \cap \mathcal{Q}_{H} = \emptyset.
    \end{equation}
    Here, $\mathcal{P}_{\mathcal{N}}^{\mathrm{cor}} =\{\pm1,\pm i\}\cdot\{P_iP_j\}_{i,j=0}^K  \backslash \qty{I_n}$, $\mathcal{U}_{\mathcal{N}} = \mathcal{N}\circ\mathcal{U}$, and $\mathcal{Q}_{H}$ is a set of Pauli operators generated from Pauli operators with non-zero coefficients in the Pauli expansion of $H$, as defined in Eq.~\eqref{eq_Q_U}.
\end{thm}

The proof of Theorem~\ref{thm_4} is essentially the same as that of Theorem~\ref{thm_3}.
The only difference is that we use the quantum error-correction condition instead of the quantum error-detection condition~\cite{nielsen2010quantum, gottesman2016surviving}, so $\mathcal{P}^{\mathrm{{det}}}_{\mathcal{N}} = \{P_i\}_{i=1}^K$ in Eq.~\eqref{eq_thm_3_2} becomes $\mathcal{P}^{\mathrm{{cor}}}_{\mathcal{N}} = \{\pm1,\pm i\}\cdot\{P_iP_j\}_{i,j=0}^K  \backslash \qty{I_n}$.
See Appendix~\ref{sec_proof} for the full proof.
The interesting aspect of Theorem~\ref{thm_4} is that the necessary and sufficient condition in Eq.~\eqref{eq_thm_4_2} exactly matches the condition for correctable noise using SCV under Pauli symmetry, which is characterized in Corollary~\ref{cor_2}.
Therefore, when we are restricted to using Clifford gates to correct errors in the noisy channel, we cannot correct more noise than using SCV under Pauli symmetry.
This represents the sufficiency and optimality of applying SCV in this scenario.

We remark that our error-correction condition is more strict than those required in covariant quantum error correction~\cite{zhou2021new, kubica2021using}.
In Refs.~\cite{zhou2021new, kubica2021using}, the authors aimed to purify the noisy channel $\mathcal{U}_{\mathcal{N}}$ into a noiseless channel $\mathcal{U}'(\cdot) = e^{i\theta H'} \cdot e^{-i\theta H'}$, where they allow the Hamiltonian to change from $H$ to $H'\not\propto I_n$.
Specifically, they showed that the transformation
\begin{equation}
    \label{eq_covariant_def}
    \forall \theta, \; \Theta^{\mathrm{cor}}_{\mathrm{Clif}} (\mathcal{U}_\mathcal{N}) = \mathcal{U}'
\end{equation}
is possible if and only if
\begin{equation}
    \label{eq_covariant_cond}
    \mathcal{Q}_{H}' \not\subset \mathcal{P}^{\mathrm{{cor}}}_{\mathcal{N}},
\end{equation}
where $\mathcal{Q}_{H}'$ is a set of Pauli operators with non-zero coefficients in the Pauli expansion of $H$, as defined in Eq.~\eqref{eq_Q_U'}.
Even though the Hamiltonian may change when correcting errors, this transformation is useful for estimating $\theta$ at the Heisenberg limit~\cite{zhou2018achieving}.
The condition in Eq.~\eqref{eq_covariant_cond} implies that, if there is at least one term in the Hamiltonian that is not included in the set $\mathcal{P}^{\mathrm{{cor}}}_{\mathcal{N}}$, then we can preserve that term while correcting all the errors.
More precisely, Eq.~\eqref{eq_covariant_cond} states that there exists a Pauli operator $Q \in\mathcal{Q}_{H}'$ such that $Q \notin\mathcal{P}^{\mathrm{{cor}}}_{\mathcal{N}}$.
Under this condition, we can achieve Eq.~\eqref{eq_covariant_def} by setting $H' \propto Q$.
However, our goal in channel purification is not to preserve at least one term of the Hamiltonian, but to preserve the entire Hamiltonian $H$ when correcting errors.
Therefore, all elements in $\mathcal{Q}_{H}'$ (including Pauli operators generated from them) should not be in the set $\mathcal{P}^{\mathrm{{cor}}}_{\mathcal{N}}$, resulting in the condition in Eq.~\eqref{eq_thm_4_2} in Theorem~\ref{thm_4}.

Finally, we consider the case where the effect of noise cannot be entirely eliminated.
In Theorem~\ref{thm_3} and Theorem~\ref{thm_4}, we aim to recover the ideal noiseless channel $\mathcal{U}$ itself by detecting and correcting all noise components.
However, as seen in our numerical results, such channel purification is not possible in many practical setups, even when noise is absent in the SCV gadget.
In such cases, we aim to construct $\Theta^{\mathrm{cor}}$ so that the purified channel $\Theta^{\mathrm{cor}} (\mathcal{U}_\mathcal{N})$ becomes as close as possible to the noiseless channel $\mathcal{U}$.
To quantify the performance of such approximate purification, we define the worst-case fidelity between two quantum channels $\mathcal{E}$ and $\mathcal{F}$ as
\begin{equation}
    \label{eq_worst_case_fid}
    F(\mathcal{E},\mathcal{F}) = \min_{\rho} F(\mathcal{I}\otimes\mathcal{E}(\rho), \mathcal{I}\otimes\mathcal{F}(\rho)),
\end{equation}
where $F$ on the right-hand side represents the fidelity between two quantum states, and $\mathcal{I}$ is the identity channel of the same dimension as $\mathcal{E}$.
Our goal is to derive an upper bound on the worst-case fidelity when the superchannel $\Theta^{\mathrm{cor}}_{\mathrm{Clif}}$ is implemented using Clifford unitaries.

For simplicity, let us consider the case of correcting Pauli noise affecting Pauli rotation gates.
Since arbitrary Pauli rotation gates can be transformed into Pauli-$Z$ rotation gates by conjugating them with Clifford gates, we focus specifically on purifying the Pauli-$Z$ rotation gate $\mathcal{U}(\cdot) = e^{-i\frac{\theta}{2} Z}\cdot e^{i\frac{\theta}{2} Z}$ under single-qubit Pauli noise $\mathcal{N}(\cdot) = p_0 \cdot +p_x X\cdot X + p_y Y\cdot Y + p_z Z\cdot Z$.
By applying an upper bound on the worst-case fidelity obtained from the resource-theoretic analysis in Ref.~\cite{regula2021fundamental}, we establish the following theorem.

\begin{thm}
    \label{thm_5}
    Let $\mathcal{U}(\cdot) = e^{-i\frac{\theta}{2} Z}\cdot e^{i\frac{\theta}{2} Z}$ be a Pauli-$Z$ rotation gate with $0\leq \theta \leq \pi/4$, and let $\mathcal{N}(\cdot) = p_0 \cdot +p_x X\cdot X + p_y Y\cdot Y + p_z Z\cdot Z$ be a single-qubit Pauli noise channel affecting $\mathcal{U}$.
    Then, the superchannel $\Theta^{\mathrm{cor}}_{\mathrm{Clif}}$ satisfies
    \begin{equation}
        \label{eq_thm_5_1}
        \begin{aligned}
            F(\Theta^{\mathrm{cor}}_{\mathrm{Clif}}(\mathcal{U}_{\mathcal{N}}), \mathcal{U})
            \leq \frac{1+\cos\theta}{2}\biggl((&p_0+p_z)R\qty(\theta,\frac{p_z}{p_0+p_z})\\
            + (&p_x+p_y)R\qty(\theta, \frac{p_y}{p_x+p_y})\biggr).
        \end{aligned}
    \end{equation}
    Here, $\mathcal{U}_{\mathcal{N}} = \mathcal{N}\circ\mathcal{U}$, and $R(\theta, p)$ represents the resource robustness of the quantum state $\frac{1}{2}(I + (1-2p)(\cos\theta X + \sin \theta Y))$, given by
    \begin{equation}
        R(\theta, p) = \max\qty{(2 - \sqrt{2})\qty(1 + \abs{1-2p}\cos\qty(\theta-\frac{\pi}{4})), 1}.
    \end{equation}
\end{thm}

While the original bound in Ref.~\cite{regula2021fundamental} was derived using resource robustness~\cite{vidal1999robustness, brandao2015reversible, takagi2019general, seddon2019quantifying} of quantum \textit{channel}, our key technical contribution in proving Theorem~\ref{thm_5} is to show that the resource robustness for a noisy Pauli rotation gate $\mathcal{N}\circ\mathcal{U}$ reduces to that of a single-qubit quantum \textit{state}.
This simplification is achieved by leveraging the symmetry inherent in the Choi state of $\mathcal{N}\circ\mathcal{U}$.
For details and the proof of Theorem~\ref{thm_5}, we refer the reader to Appendix~\ref{sec_resource}.

For $\theta = \pi/4$ and assuming $\frac{p_z}{p_0+p_z}, \frac{p_y}{p_x+p_y}< \frac{2-\sqrt{2}}{4}$, Eq.~\eqref{eq_thm_5_1} simplifies to
\begin{equation}
    \label{eq_thm_5_2}
    F(\Theta^{\mathrm{cor}}_{\mathrm{Clif}}(\mathcal{U}_{\mathcal{N}}), \mathcal{U}) \leq 1-p_y-p_z.
\end{equation}
This bound is achieved by applying the SCV gadget $\Theta^{\mathrm{cor}}_{\mathcal{Q}^{\mathrm{com}}_{Z}}$ to the noisy channel, with a Pauli-$X$ operator used as feedback control.
In this case, Pauli-$X$ errors are corrected, while Pauli-$Y$ errors are transformed into Pauli-$Z$ errors.
As a result, the noise channel is transformed into $\Theta^{\mathrm{cor}}_{\mathcal{Q}^{\mathrm{com}}_{Z}}(\mathcal{N})(\cdot) = (1-p_y-p_z)\cdot + (p_y+p_z)Z\cdot Z$, whose worst-case fidelity saturates the bound in Eq.~\eqref{eq_thm_5_2}.
This result indicates that, even in the context of approximate purification of noisy channels, SCV under Pauli symmetry is an optimal protocol saturating the bound.

Besides showing the optimality of the SCV in noise purification, our result is also insightful from a resource-theoretic perspective.
In fact, the bound in Theorem~\ref{thm_5} holds for a much larger class of superchannels than the operationally motivated one that $\Theta_{\rm Clif}^{\rm cor}$ belongs to. (See Appendix~\ref{sec_resource} for details.)
Therefore, the tightness of the bound in Theorem~\ref{thm_5}---and the original bound in Ref.~\cite{regula2021fundamental} more broadly---with respect to operationally motivated classes of operations had been highly unclear.
Our results demonstrate that it can indeed be achieved using an operationally motivated superchannel implemented by our SCV protocol.
Thus, our findings highlight the practical relevance of resource theories in real-world scenarios.

\section{Discussion}
\label{sec_conclusion}
In this work, we proposed symmetric channel verification (SCV) and its hardware-efficient variant, virtual symmetric channel verification (virtual SCV), as novel methods for purifying quantum channels by leveraging the symmetry inherent in quantum operations.
These protocols generalize symmetry-based error countermeasures to the channel level, enabling noise reduction in scenarios where traditional symmetry verification~\cite{bonet2018low, mcardle2019error} is inapplicable, such as when input states or individual channels in a circuit lack shared symmetry.

SCV introduces a quantum phase estimation-like circuit that detects symmetry-breaking noise through coherent interactions with ancillary qubits.
Virtual SCV provides a hardware-efficient alternative, requiring only a single-qubit ancilla and controlled Pauli gates, which are robust against noise.
We demonstrated the effectiveness of SCV and virtual SCV in purifying quantum channels in various applications such as Hamiltonian simulation circuits and phase estimation circuits based on qubitization.
Furthermore, we discussed their implementation in early fault-tolerant regimes, showing that SCV under Pauli symmetry represents an optimal protocol when operations are restricted to Clifford unitaries.

Our results highlight the versatility and practicality of SCV and virtual SCV as tools for noise reduction in quantum computation.
Compared to traditional symmetry verification methods~\cite{bonet2018low, mcardle2019error} and their virtual variants~\cite{mcclean2020decoding, cai2021quantum, endo2022quantum, tsubouchi2023virtual}, SCV and virtual SCV offer broader applicability by addressing noise in quantum channels with varying symmetries and input states that lack symmetry.
In contrast to other channel purification protocols~\cite{lee2023error, miguel2023superposed, liu2024virtual, xiong2023circuit, debroy2020extended, gonzales2023quantum, van2023single, das2024purification}, SCV and virtual SCV require only a single noisy channel as input and make minimal assumptions about the target unitary.
Consequently, our protocols are applicable to a wider range of scenarios and are easier to implement than existing methods.

Future research could explore several promising directions.
First, while we focused on Hamiltonian simulation and phase estimation circuits as applications, our protocols can be applied to a wider range of tasks.
One example worth exploring is the purification of noisy black-box unitaries under symmetry.
In the setup of quantum amplitude amplification~\cite{brassard2000quantum} or quantum metrology~\cite{zhou2018achieving}, we are given a black-box unitary, but we may know its symmetric structure beforehand.
Our methods can be used to detect and correct errors in such black-box unitaries.
Another interesting application is demonstrating quantum advantage using IQP circuits, which have been shown to be classically hard to simulate~\cite{bremner2011classical, bremner2016average}.
Since IQP circuits possess a highly symmetric structure, our protocols can be straightforwardly applied to combat their noise.

Applications to quantum many-body physics can also be considered an important future direction.
For example, in quench dynamics and dissipative dynamics, dynamical spontaneous symmetry breaking can occur, where the symmetry of the state is reduced compared to the symmetry of the time evolution channel itself~\cite{zhou2021nonequilibrium}.
When performing quantum simulations of such phenomena, it is crucial to precisely detect symmetry breaking arising from emergent phenomena itself but not that from the effect of noise.
Therefore, SCV will play an important role in these simulations.

We can further explore the connection between channel purification and resource theory.
In this work, we restricted our analysis to situations where the unitary operations constituting the purification gadget are Clifford gates.
However, in general, allowing non-Clifford operations in the gadget is a natural extension, making it important to investigate the trade-off with magic-state consumption.
Additionally, it is crucial to analyze resource utilization in fermionic and bosonic systems, as well as the convertibility of resources when utilizing hybrid quantum devices~\cite{hahn2025bridging}.

Finally, it is crucial to develop a fault-tolerant implementation of the SCV gadget.
As discussed above, the use of flag qubits~\cite{chao2018quantum, chamberland2018flag, chao2020flag} can be effective in removing noise in the ancilla qubit.  
However, we have not yet explored how they can be employed to correct channel noise in a fault-tolerant manner.  
Developing a fault-tolerant, channel-level error-correction scheme using SCV could further broaden the applicability of our protocols.

\section*{Data availability}
The code and data that support the findings of this article are openly available~\cite{ktsubo2025github}.

\section*{Acknowledgements}
The authors wish to thank Suguru Endo, Alvin Gonzales, Yuki Koizumi, and Takahiro Sagawa for fruitful discussions.
K.T. is supported by the Program for Leading Graduate Schools (MERIT-WINGS) and JST BOOST Grant No. JPMJBS2418.
Y.M. is supported by JSPS KAKENHI Grant No. JP23KJ0421.
R.T. is supported by JSPS KAKENHI Grant No. JP23K19028, JP24K16975, JST, CREST Grant No. JPMJCR23I3, Japan, and MEXT KAKENHI Grant-in-Aid for Transformative Research Areas A ``Extreme Universe” Grant No. JP24H00943.
N.Y. is supported by JST PRESTO No. JPMJPR2119, JST Grant No. JPMJPF2221, JST CREST Grant No. JPMJCR23I4, IBM Quantum, JST ASPIRE Grant No. JPMJAP2316, JST ERATO Grant No. JPMJER2302, and Institute of AI and Beyond of the University of Tokyo.

\appendix
\section{Using flag qubits to detect idling errors in the SCV gadget}
\label{sec_flag}
As discussed in the main text, the SCV gadget is affected by idling errors, which occur in the idle ancilla qubits between unitaries $U_{\mathrm{E}}$ and $U_{\mathrm{D}}$.
Here, $U_{\mathrm{E}}$ and $U_{\mathrm{D}}$ represent the unitaries applied before and after the noisy channel $\mathcal{U}_{\mathcal{N}}$, respectively.
If the idling errors can be characterized as Pauli-$Z$ errors, they flip the measurement results of the ancilla qubits.
Assuming that the error rates in the noisy channel $\mathcal{U}_{\mathcal{N}}$ and the idling errors are small, such Pauli-$Z$ errors primarily introduce false positives, which do not affect the purification accuracy after post-selection to first order in the error rate.
In contrast, Pauli-$X$ errors on the ancilla qubits can greatly degrade the computational performance.
This is because such errors do not affect the measurement outcomes on the ancilla and are thus undetectable, but they can propagate to the system qubits and introduce additional errors.

To reduce the effect of these detrimental Pauli-$X$ errors, we apply an additional SCV gadget $\Theta^{\mathrm{det}}_Z$ to the idling ancilla qubits.
This can be implemented by adding an extra ancilla qubit for each ancilla qubit in the original SCV gadget, as depicted in Fig.~\ref{fig_SCV_flag}.
These additional qubits are called \textit{flag qubits} in the field of quantum error correction~\cite{chao2018quantum, chamberland2018flag, chao2020flag}.
As discussed in Sec.~\ref{sec_SCV_Pauli}, the flag qubits composing the SCV gadget $\Theta^{\mathrm{det}}_Z$ can detect Pauli-$X$ errors.
Moreover, idling errors in the flag qubits can also be detected by either the original ancilla qubits or the flag qubits themselves.
Therefore, by using flag qubits, the accuracy of the post-selected quantum channel is preserved to first order in the error rate.
Note that the accuracy can still be affected at second order.
For example, if errors occur simultaneously on a system qubit and an idling ancilla qubit, they can cause a false negative, where the channel $\mathcal{U}$ is affected by an error but the error remains undetected due to an error in the ancilla qubit.

\begin{figure}[t]
    \begin{center}
        \includegraphics[width=0.99\linewidth]{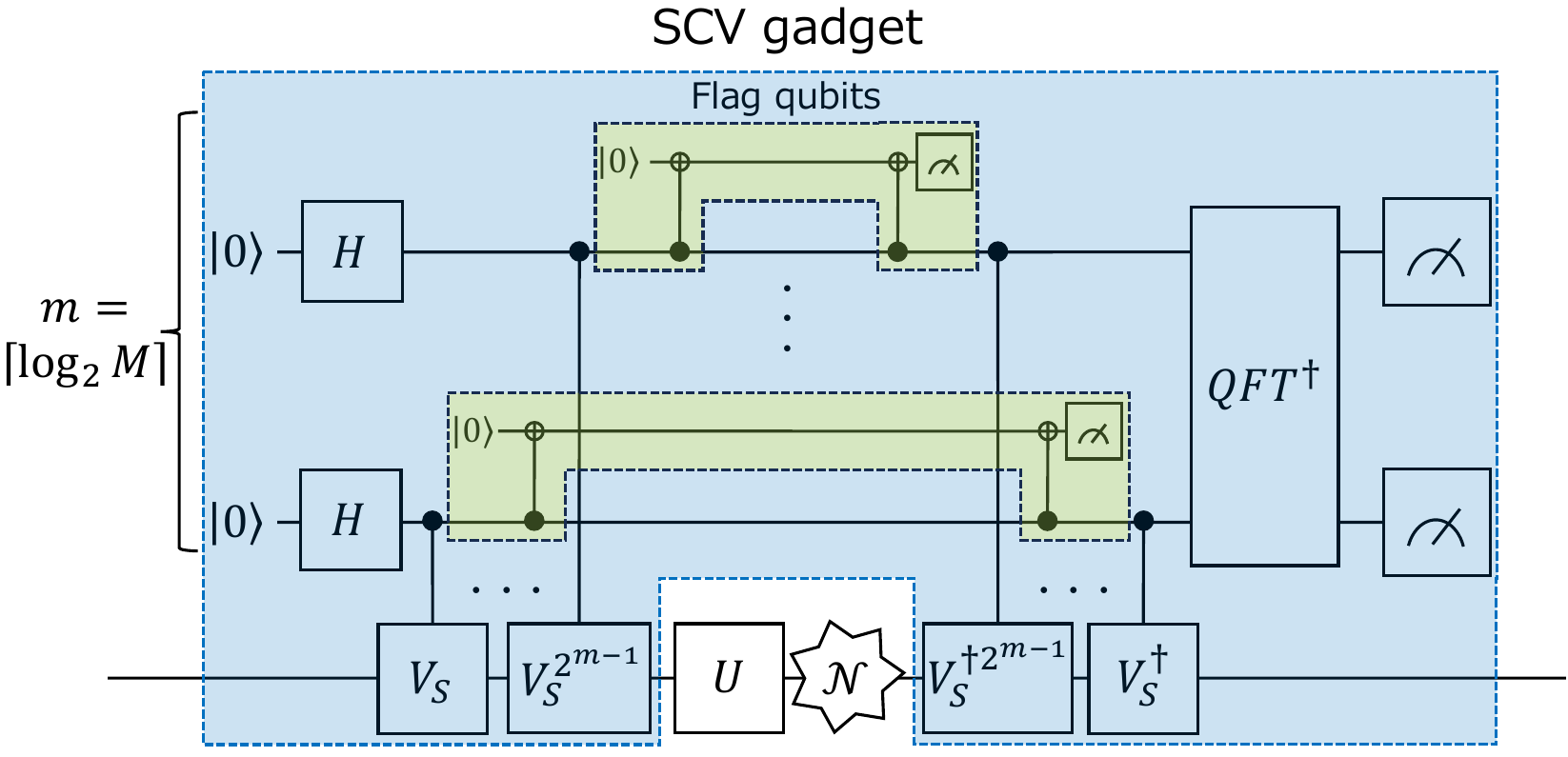}
        \caption{Circuit structure of SCV with flag qubits. A flag qubit is added to each ancilla qubit in the original SCV gadget depicted in Fig.~\ref{fig_SCV} to detect Pauli-$X$ errors on the ancilla qubits.}
        \label{fig_SCV_flag}
    \end{center}
\end{figure}

While idling errors on the ancilla qubits can be detected using flag qubits, this introduces additional sampling overhead in the SCV protocol.
Consider the case where the noise can be expressed as $\mathcal{N}(\cdot) = (1-p)\cdot + p \sum_j N_j \cdot N_j^\dag$ with $\sum_i \Pi_i N_j \Pi_i = 0$, where $\Pi_i$ are projectors onto the eigenspaces of a symmetric operator $S$ under consideration.
This noise breaks the symmetry of a unitary channel $\mathcal{U}$ with probability $p$.
In the absence of idling errors, we obtain a noiseless purified channel with probability $1-p$, resulting in a sampling overhead of approximately $(1-p)^{-1} \sim (1+p)$.
Meanwhile, if idling errors occur with probability $p_{\mathrm{idle}}$ for each ancilla qubit, we obtain a non-trivial measurement result with probability $1-p-2m p_{\mathrm{idle}}$, to first order in $p$ and $p_{\mathrm{idle}}$.
Thus, the sampling overhead increases from $(1+p)$ to $(1+p+2m p_{\mathrm{idle}})$ in the presence of idling errors.

The use of a flag qubit can be regarded as encoding an ancilla qubit into a distance-2 repetition code, which is robust against Pauli-$X$ errors.  
In this way, Pauli-$X$ errors can be detected using flag qubits.  
More generally, one could consider an encoding with a higher code distance, which would enable not only the detection but also the correction of Pauli-$X$ errors.  
We expect that such error correction would reduce the contribution of $2m p_{\mathrm{idle}}$ to the sampling overhead.

\begin{figure}[t]
    \begin{center}
        \includegraphics[width=0.85\linewidth]{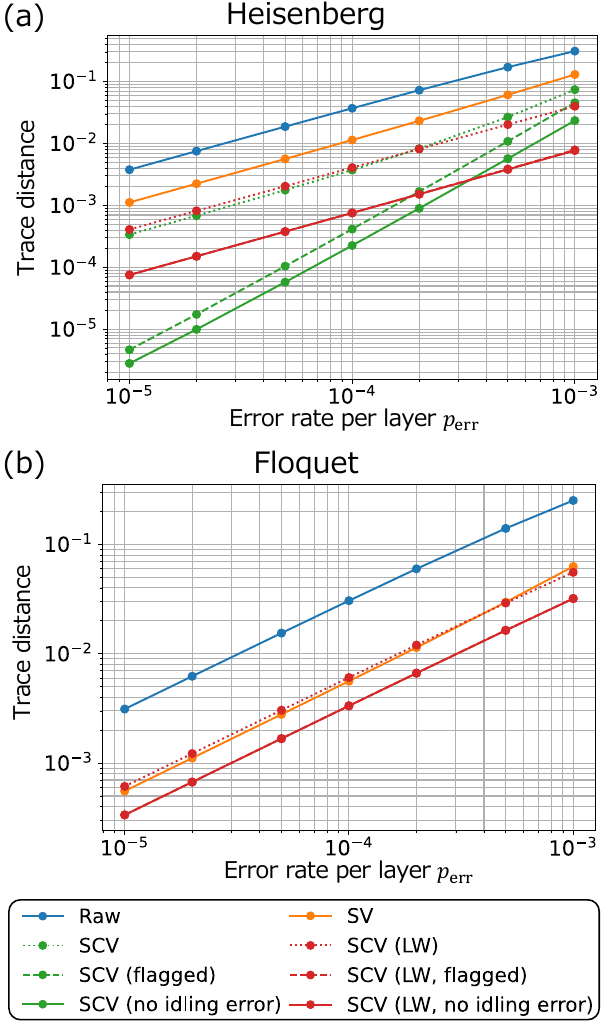}
        \caption{Performance of SCV in detecting errors in Hamiltonian simulation circuits for (a) the 1D Heisenberg model and (b) Floquet dynamics, in the presence of idling errors on ancilla qubits. ``Raw'' denotes the trace distance between the noisy quantum state and the ideal quantum state, while ``SV'' represents the result for the quantum state purified by symmetry verification. The lines labeled ``SCV'' show the results for SCV, with green lines indicating the case where SCV is applied to the entire noisy circuit $\bigcirc_{l=1}^L \mathcal{N}_l \circ \mathcal{U}_l$, and red lines labeled ``LW'' indicating the layer-wise application of SCV. Dotted lines correspond to the case without flag qubits, whereas dashed lines labeled ``flagged'' correspond to the case with flag qubits. The solid line labeled ``no idling error'' shows the results without idling errors. Note that the lines for ``SCV (LW, flagged)'' and ``SCV (LW, no idling error)'' overlap, indicating that the use of flag qubits successfully suppresses the impact of idling errors up to the first order in the error rate.
        }
        \label{fig_numerics_Pauli_flag}
    \end{center}
\end{figure}

\begin{figure}[t]
    \begin{center}
        \includegraphics[width=0.85\linewidth]{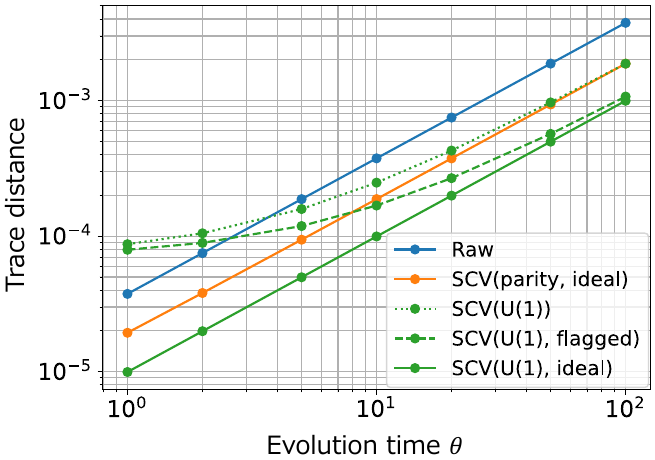}
        \caption{Performance of SCV in detecting errors in a Hamiltonian simulation circuit for the H$_2$ molecule, in the presence of idling errors on the ancilla qubits. ``Raw'' denotes the trace distance between the noisy quantum state and the ideal quantum state. ``SCV (parity, ideal)'' shows the result using SCV for the parity operator $\prod_i Z_i$, assuming a noiseless SCV gadget. ``SCV (U(1))'', ``SCV (U(1), flagged)'', and ``SCV (U(1), ideal)'' show the results using SCV for the particle number operator $\sum_i Z_i$, where the dotted ``SCV (U(1))'' line corresponds to the case without flag qubits, the dashed ``SCV (U(1), flagged)'' line corresponds to the case with flag qubits, and the solid ``SCV (U(1), ideal)'' line corresponds to the case with a noiseless SCV gadget.}
        \label{fig_numerics_U1_flag}
    \end{center}
\end{figure}

In the numerical simulations for SCV under Pauli symmetry (Sec.~\ref{sec_SCV_Pauli}) and particle-number-conservation symmetry (Sec.~\ref{sec_SCV_U1}), we have neglected idling errors on the ancilla by assuming that logical qubits with a higher code distance are used for the ancilla. 
When this assumption does not hold, i.e., when ancilla qubits are affected by idling errors comparable to those on the system qubits, SCV can still remain effective by introducing flag qubits.

We first present the results of numerical simulations for SCV under Pauli symmetry in the presence of idling-ancilla noise, shown in Fig.~\ref{fig_numerics_Pauli_flag}. 
The simulation setup is almost the same as that in Sec.~\ref{sec_SCV_Pauli}, except that we introduce local depolarizing noise for all ancilla (including flag qubits) with an error rate of $p_{\mathrm{err}}$ per layer. 
We set $n=4$ for the Heisenberg model and $n=2$ for the Floquet dynamics. 
Without flag qubits, idling errors on the ancilla significantly degrade the performance of SCV. 
However, introducing flag qubits restores the performance, as they can detect idling errors up to the first order in the error rate. 
The remaining difference between the performance of SCV without idling errors and that with flag qubits appears at second order. 
We expect that this difference can be suppressed to arbitrary order by increasing the number of flag qubits, as proposed in Refs.~\cite{chamberland2018flag, chao2020flag}.

We next present the results of numerical simulations for SCV under particle-number-conservation symmetry in the presence of idling-ancilla noise, shown in Fig.~\ref{fig_numerics_U1_flag}. 
The setup is almost the same as in Sec.~\ref{sec_SCV_U1}, except that we introduce local depolarizing noise for all ancilla (including flag qubits) with an error rate of $p_{\mathrm{err}}$ per layer. 
Without flag qubits, SCV for the particle number operator $\sum_i Z_i$ does not surpass the performance of the parity operator $\prod_i Z_i$. 
However, with flag qubits, it achieves performance comparable to the noiseless case for large evolution times, consistent with the trends observed in Sec.~\ref{sec_SCV_U1}.

We expect that flag qubits will also be effective for correcting errors in quantum channels using SCV, as discussed in Sec.~\ref{sec_SCV_correct}. 
However, as in the case of state-level fault-tolerant protocols using flag qubits~\cite{chao2018quantum, chamberland2018flag, chao2020flag}, this would require the construction of a complex fault-tolerant protocol, which is beyond the scope of this work. 
We leave the development of a fault-tolerant error-correction protocol at the channel level as an interesting direction for future research.

\section{Applicability of SCV in near-term quantum devices}
\label{sec_NISQ}

In the main text, we primarily considered using SCV in the early fault-tolerant regime.
In this regime, Clifford unitaries are easier to implement than non-Clifford unitaries because they do not require gate synthesis.
Accordingly, in the numerical simulations in the main text, we assumed that if $U_{\mathrm{E}}$ and $U_{\mathrm{D}}$ composing the SCV gadget are Clifford unitaries, they are affected by local depolarizing noise with an error rate of $p_{\mathrm{err}}/100$.
However, in near-term quantum devices, this assumption does not hold: Clifford and non-Clifford gates are expected to have comparable error rates.
Even in such cases, our SCV protocol remains effective, provided that the target noisy unitary $\mathcal{U}_{\mathcal{N}}$ contains significantly more quantum gates than $U_{\mathrm{E}}$ and $U_{\mathrm{D}}$ in the SCV gadget.

\begin{figure}[t]
    \begin{center}
        \includegraphics[width=0.85\linewidth]{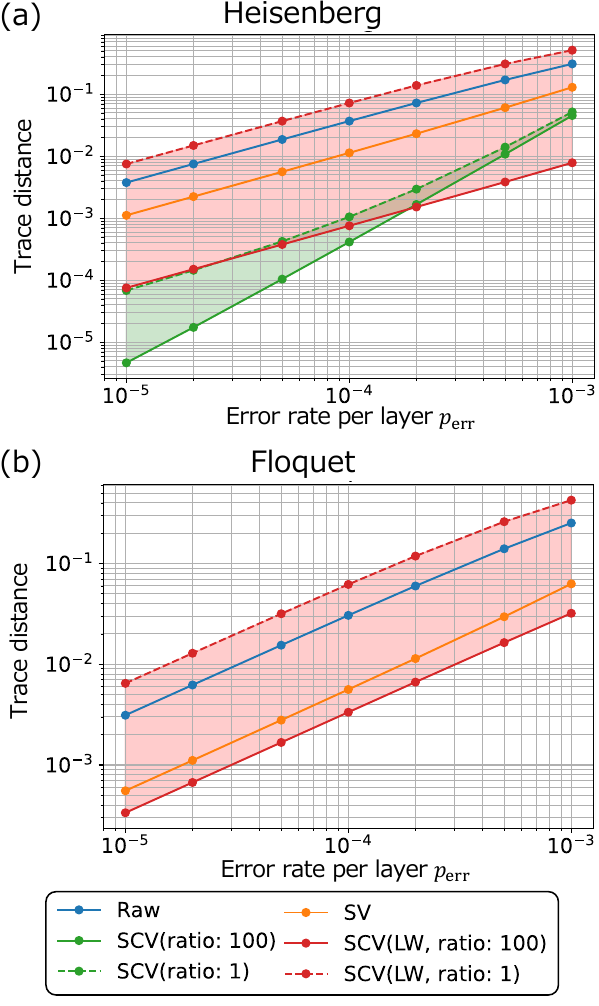}
        \caption{Performance of SCV in detecting errors in Hamiltonian simulation circuits for (a) the 1D Heisenberg model and (b) Floquet dynamics, with varying error rates on the SCV gadget. ``Raw'' denotes the trace distance between the noisy quantum state and the ideal quantum state, while ``SV'' represents the result for the quantum state purified by symmetry verification. The lines labeled ``SCV'' show the results for SCV, with green lines indicating the case where SCV is applied to the entire noisy circuit $\bigcirc_{l=1}^L \mathcal{N}_l \circ \mathcal{U}_l$, and red lines labeled ``LW'' indicating the layer-wise application of SCV. The solid lines labeled ``ratio: 100'' represent the case where $U_{\mathrm{E}}$ and $U_{\mathrm{D}}$ in the SCV gadget are affected by an error rate of $p_{\mathrm{err}}/100$, while the dashed lines labeled ``ratio: 1'' correspond to the case where they are affected by an error rate of $p_{\mathrm{err}}$.
        }
        \label{fig_numerics_Pauli_NISQ}
    \end{center}
\end{figure}

To demonstrate this, we performed numerical simulations of SCV under Pauli symmetry, varying the ratio of the error rates between the noisy unitary $\bigcirc_{l=1}^L \mathcal{N}_l \circ \mathcal{U}_l$ and the unitaries $U_{\mathrm{E}}$ and $U_{\mathrm{D}}$ in the SCV gadget from 100 to 1.
A ratio of 100 corresponds to the early fault-tolerant regime, while a ratio of 1 corresponds to the near-term regime.
The simulation setup is the same as in Appendix~\ref{sec_flag} where we use flag qubits, except that we vary the error rate of $U_{\mathrm{E}}$ and $U_{\mathrm{D}}$.
The results are shown in Fig.~\ref{fig_numerics_Pauli_NISQ}.
As seen in the figure, when the SCV gadget is applied to the entire noisy unitary $\bigcirc_{l=1}^L \mathcal{N}_l \circ \mathcal{U}_l$, which contains many more gates than $U_{\mathrm{E}}$ and $U_{\mathrm{D}}$, SCV remains effective.
Since we have set $L=100$, we observe an error reduction of approximately $2/100 = 1/50$ in the small-error regime.
In contrast, when the noisy SCV gadget is applied to each layer, the performance becomes worse than doing nothing.
This indicates that, in the near-term regime, the target unitaries must contain more gates than the SCV gadgets for the noise introduced by the SCV gadget to be negligible.

\begin{figure*}[t]
    \begin{center}
        \includegraphics[width=0.7\linewidth]{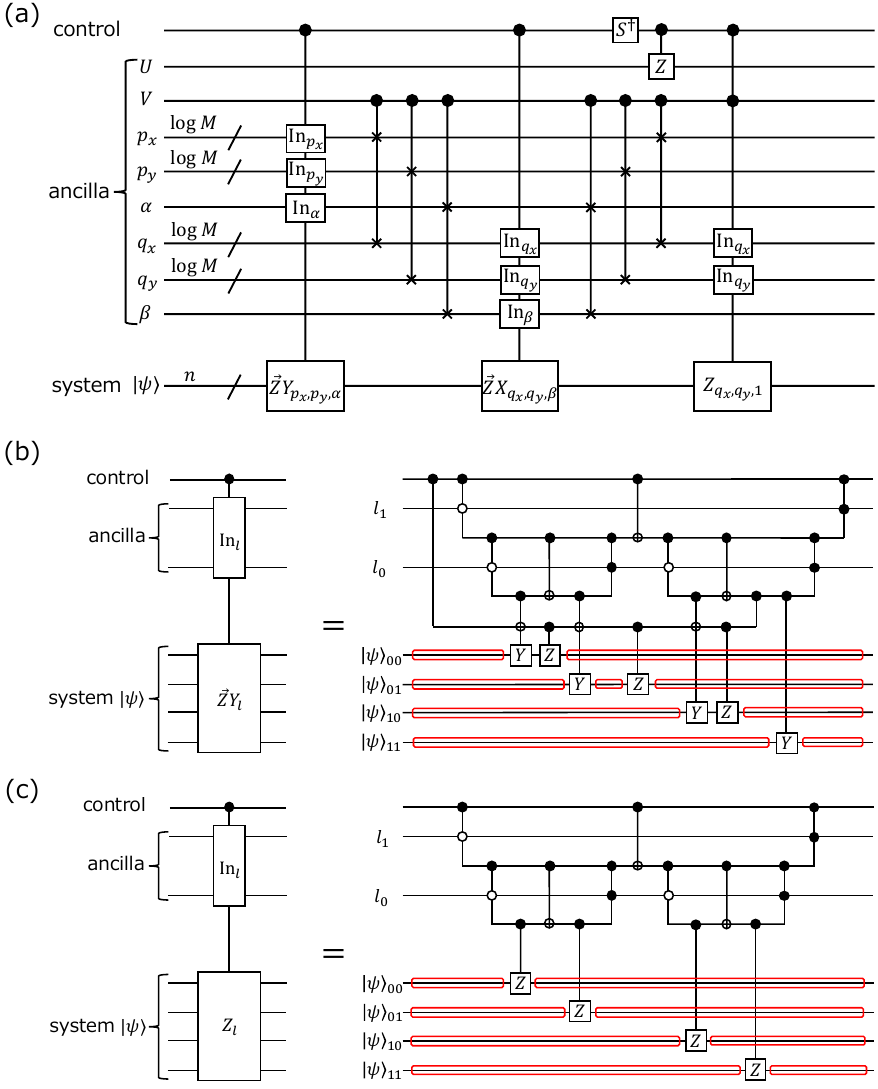}
        \caption{\texttt{SELECT} operation of the 2D Fermi-Hubbard model on a square lattice with linear system size $M$ and system qubit count $n = 2M^2$ that use the adder to iterate the indices of ancilla~\cite{babbush2018encoding, gidney2018halving}. Panel (a) represents the entire circuit, while panels (b) and (c) represent the subcircuits shown in panel (a). The system qubits include many idling qubits, illustrated in red. Idling errors occurring here can be mitigated using virtual SCV. Note that ${\rm In}_l$ indicates that the operation is controlled by $l$-th computational basis of ancilla system, $p_x, p_y, q_x$ and $q_y$ denote the two-dimensional site indices, $\alpha$ and $\beta$ indicate the spin of fermions, and $U$ and $V$ are introduced to discriminate between $Z$ and $ZZ$ contribution to the Hamiltonian~\cite{babbush2018encoding}.
        }
        \label{fig_select_circuit}
    \end{center}
\end{figure*}

\section{\texttt{SELECT} operation of the 2D Fermi-Hubbard model on a square lattice}
\label{sec_select}
In this section, we explain the \texttt{SELECT} operation for the 2D Fermi-Hubbard model on a square lattice with linear system size $M$ and system qubit count $n = 2M^2$, which we have used to demonstrate the performance of virtual SCV in Sec.~\ref{sec_VSCV_idle}.
The Hamiltonian of the Fermi-Hubbard model is given by
\begin{equation}
    H = -t \sum_{\langle p,q \rangle, \sigma} a_{p,\sigma}^{\dagger} a_{q,\sigma} + \frac{u}{2} \sum_{p, \alpha \neq \beta} n_{p,\alpha} n_{p,\beta},
\end{equation}
where $a_{p,\sigma}^{(\dagger)}$ indicates the fermionic annihilation (creation) operator on the site $p$ with spin $\sigma$, $t$ is the hopping amplitude, and $u$ is the magnitude of the onsite interaction.
Here, the summation $\sum_{\langle p,q \rangle}$ runs over adjacent pairs on a square lattice under periodic boundary conditions. 

Under the Jordan-Wigner transformation, this Hamiltonian is mapped to a qubit Hamiltonian as
\begin{equation}
    \begin{aligned}
        H = -\frac{t}{2} &\sum_{\langle p,q \rangle, \sigma} \left( X_{p,\sigma} \overrightarrow{Z} X_{q,\sigma} + Y_{p,\sigma} \overrightarrow{Z} Y_{q,\sigma} \right) \\
        &+ \frac{u}{8} \sum_{p, \alpha \neq \beta} Z_{p,\alpha} Z_{p,\beta} - \frac{u}{4} \sum_{p, \sigma} Z_{p,\sigma},
    \end{aligned}
\end{equation}
where we have neglected the constant shift of energy. Also, the notation $\vec{Z}$ indicates the so-called ``$Z$-string." Namely, $X_{p, \sigma} \vec{Z} X_{q, \sigma}$ is a product of Pauli-$X$ operators and Pauli-$Z$ operators whose indices run between integers that are used to uniquely label the pairs $(p, \sigma)$ and $(q, \sigma).$

As proposed in Ref.~\cite{babbush2018encoding}, the \texttt{SELECT} operation for this Hamiltonian can be implemented by the circuit shown in Fig.~\ref{fig_select_circuit}.
As illustrated in red, the system qubits include many idling qubits.
By applying virtual SCV to these red idling qubits, idling errors can be mitigated, as discussed in Sec.~\ref{sec_VSCV_idle}.

Notably, the same ancilla qubit of the virtual SCV gadget can be used to perform virtual SCV on all these idling qubits.
This is because the ancilla qubit in Fig.~\ref{fig_VSCV} remains in the state $\ket{+}$ and always serves as the control qubit, determining whether to apply $P_i$ or $P_j$ before the noisy channel $\mathcal{U}_{\mathcal{N}}$, and $P_k$ or $P_l$ after the noisy channel.
More precisely, the ancilla qubit is first prepared in $\ketbra{+} = \frac{1}{2}(\ketbra{0}{0} + \ketbra{0}{1} + \ketbra{1}{0} + \ketbra{1}{1})$.
For the terms with $\ket{0}$ (or $\ket{1}$), $P_i$ and $P_k$ (or $P_j$ and $P_l$) are applied to $\mathcal{U}_{\mathcal{N}}(\cdot)$ from the left, while for terms with $\bra{0}$ (or $\bra{1}$), the same Pauli operators are applied from the right.  
Averaging over an observable $X$ leaves only the off-diagonal terms $\ketbra{0}{1}$ and $\ketbra{1}{0}$, leading to Eq.~\eqref{eq_virtual_supermap}.
Since the ancilla’s sole role is to determine which Pauli operations are implemented, the same ancilla in the $\ket{+}$ state can be sequentially reused for all system qubits.  
This remains true even when idle times are distributed across different qubits, as the virtual SCV operations can be scheduled so that the ancilla interacts with each idle qubit at the appropriate time.

For example, consider the circuit in Fig.~\ref{fig_idling_error_multi}.
If the input state to the circuit is $\rho_1 \otimes \rho_2$, the state before the measurement can be written as
\begin{equation}
    \begin{aligned}
        \frac{1}{2}(P_{k_1}\mathcal{N}_1(P_{i_1}\rho_1P_{i_1})P_{k_1} \otimes P_{k_2}\mathcal{N}_2(P_{i_2}\rho_2P_{i_2})P_{k_2} &\otimes \ketbra{0}{0} \\
        +P_{l_1}\mathcal{N}_1(P_{j_1}\rho_1P_{i_1})P_{k_1} \otimes P_{l_2}\mathcal{N}_2(P_{j_2}\rho_2P_{i_2})P_{k_2} &\otimes \ketbra{1}{0} \\
        +P_{k_1}\mathcal{N}_1(P_{i_1}\rho_1P_{j_1})P_{l_1} \otimes P_{k_2}\mathcal{N}_2(P_{i_2}\rho_2P_{j_2})P_{l_2} &\otimes \ketbra{0}{1} \\
        +P_{l_1}\mathcal{N}_1(P_{j_1}\rho_1P_{j_1})P_{l_1} \otimes P_{l_2}\mathcal{N}_2(P_{j_2}\rho_2P_{j_2})P_{l_2} &\otimes \ketbra{1}{1}).
    \end{aligned}
\end{equation}
By averaging over the operator $X$, we obtain
\begin{equation}
    \begin{aligned}
        \frac{1}{2}(P_{l_1}\mathcal{N}_1(P_{j_1}\rho_1P_{i_1})P_{k_1} \otimes P_{l_2}\mathcal{N}_2(P_{j_2}\rho_2P_{i_2})P_{k_2} \\
        +P_{k_1}\mathcal{N}_1(P_{i_1}\rho_1P_{j_1})P_{l_1} \otimes P_{k_2}\mathcal{N}_2(P_{i_2}\rho_2P_{j_2})P_{l_2}).
    \end{aligned}
\end{equation}
Next, by uniformly sampling Pauli operators from $\mathcal{Q}_I^{\mathrm{com}} = \{I, X, Y, Z\}$ and averaging, we obtain
\begin{equation}
    \begin{aligned}
        \Theta^{\mathrm{det}}_{\mathcal{Q}^{\mathrm{com}}_{I}}(\mathcal{N}_1)(\rho_1) \otimes \Theta^{\mathrm{det}}_{\mathcal{Q}^{\mathrm{com}}_{I}}(\mathcal{N}_2)(\rho_2) \propto \rho_1 \otimes \rho_2,
    \end{aligned}
\end{equation}
which shows that idling errors on both qubits are successfully mitigated.
Thus, only a single additional ancilla qubit is required to implement virtual SCV, regardless of the number of system qubits $n$.

\begin{figure}[t]
    \begin{center}
        \includegraphics[width=0.99\linewidth]{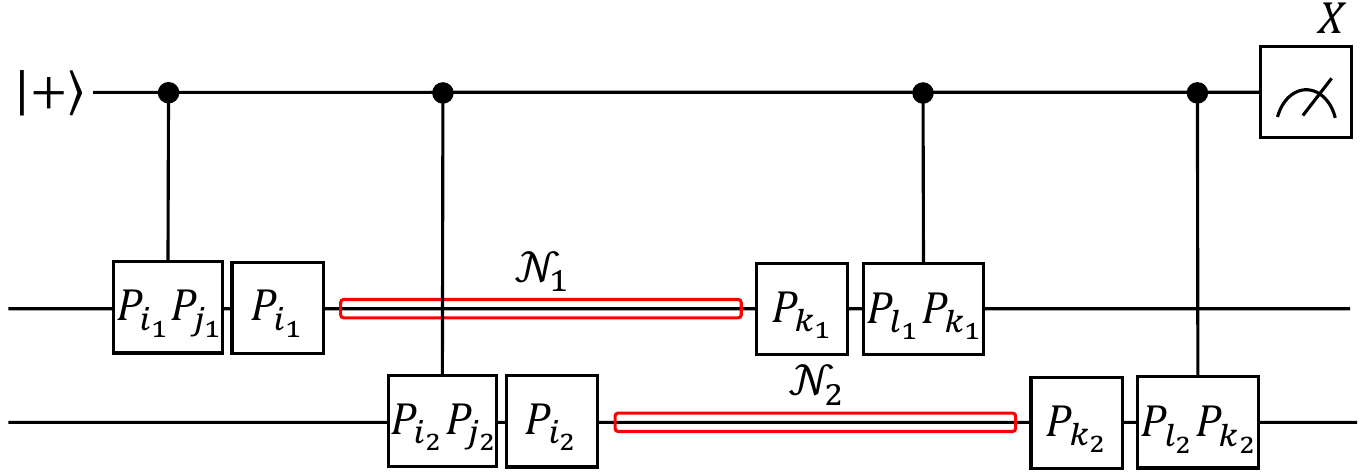}
        \caption{Schematic illustration of mitigating idling errors on multiple system qubits using a single ancilla qubit. Idling errors $\mathcal{N}_1$ and $\mathcal{N}_2$, shown in red, can be mitigated by implementing virtual SCV, where the same ancilla qubit is reused across multiple qubits.}
        \label{fig_idling_error_multi}
    \end{center}
\end{figure}

\section{Proof of the main results}
\label{sec_proof}
In this section, we provide detailed proofs of the theorems and corollaries stated in the main text, except Theorem~\ref{thm_5}.

\subsection{Proof of Theorem~\ref{thm_1}}
We begin with the proof of Theorem~\ref{thm_1}.
\begin{proof}
    As discussed in the main text, we have
    \begin{equation}
        (\Theta^{\mathrm{det}}_S(\mathcal{U}_\mathcal{N}))(\cdot) = \sum_{ij}\Pi_i\mathcal{U}_\mathcal{N}(\Pi_i \cdot \Pi_j)\Pi_j.
    \end{equation}
    Therefore, it suffices to show that $\sum_{ij}\Pi_i\mathcal{U}_\mathcal{N}(\Pi_i \cdot \Pi_j)\Pi_j \propto \mathcal{U}$.
    From the assumption $\sum_i \Pi_i N_j \Pi_i \propto I_n$, we have
    \begin{equation}
        \begin{aligned}
            \sum_{ij}\Pi_i\mathcal{U}_\mathcal{N}(\Pi_i \cdot \Pi_j)\Pi_j
            &=\sum_{ijk} \Pi_i N_k\mathcal{U}(\Pi_i \cdot \Pi_j)N_k^\dag\Pi_j \\
            &=\sum_{ijk} \Pi_i N_k\Pi_i\mathcal{U}(\cdot)\Pi_jN_k^\dag\Pi_j \\
            &\propto \mathcal{U}(\cdot),
        \end{aligned}
    \end{equation}
    completing the proof.
\end{proof}

\subsection{Proof of Corollary~\ref{cor_1}}
Next, we proceed with the proof of Corollary~\ref{cor_1}.
\begin{proof}
    From Theorem~\ref{thm_1}, $\Theta^{\mathrm{det}}_{\mathcal{Q}^{\mathrm{com}}_{U}}(\mathcal{U}_\mathcal{N}) = p_0\mathcal{U}$ implies
    \begin{equation}
        \label{eq_proof_2_1}
        \sum_jp_j\sum_{i_1\cdots i_r, i_1'\cdots i_r'} P_{j,i_1\cdots i_r}\mathcal{U}(\cdot) P_{j,i_1\cdots i_r} =  p_0\mathcal{U}(\cdot),
    \end{equation}
    where $P_{j,i_1\cdots i_r} = \Pi_{i_1}^1 \cdots \Pi_{i_r}^r P_j \Pi_{i_r}^r \cdots \Pi_{i_1}^1$ and $\Pi_{i_k}^k = \frac{1}{2}(I_n + (-1)^{i_k} Q_{k})$ represents the projector onto the eigenspace of $Q_k$.
    Since $P_{j,i_1\cdots i_r}$ satisfies
    \begin{equation}
        \sum_{i_1\cdots i_r}P_{j,i_1\cdots i_r}
        =
        \begin{cases}
            P_j  & (\forall k\in\{1,\ldots r\}, [P_j, Q_k] = 0), \\
            0    & (\mathrm{otherwise}),
        \end{cases}
    \end{equation}
    Eq.~\eqref{eq_proof_2_1} holds if and only if for all $P_j \in \mathcal{P}_{\mathcal{N}}^{\mathrm{det}}$, there exists a generator $Q_k\in\mathcal{Q}^{\mathrm{com}}_{U}$ that anti-commutes with $P_j$.
    Therefore, we obtain Eq.~\eqref{eq_cor_1_1}.

    To derive Eq.~\eqref{eq_cor_1_2}, it suffices to show
    \begin{equation}
        \label{eq_proof_2_2}
        P \notin \mathcal{Q}_{U} \Leftrightarrow \exists Q\in\mathcal{Q}_{U}^{\mathrm{com}}, [P, Q]\neq 0.
    \end{equation}
    To establish this equivalence, let us decompose $\mathcal{Q}_{U}$ as
    \begin{equation}
        \mathcal{Q}_{U} = W^\dag \qty(\qty{I,X,Y,Z}^{\otimes n_1} \otimes \qty{I,Z}^{\otimes n_2} \otimes \qty{I}^{\otimes n_3}) W,
    \end{equation}
    where $W$ is an $n$-qubit Clifford operator and $n_1+n_2+n_3 = n$.
    The existence of such a decomposition was established in Ref.~\cite{mitsuhashi2023clifford}.
    Since $\mathcal{Q}_{U}^{\mathrm{com}}$ can also be represented as 
    \begin{equation}
        \mathcal{Q}_{U}^{\mathrm{com}} = \{Q\in\mathcal{P}_n ~|~ \forall P \in \mathcal{Q}_{U}, [P, Q] = 0\},
    \end{equation}
    we have
    \begin{equation}
        \mathcal{Q}_{U}^{\mathrm{com}} = W^\dag \qty(\qty{I}^{\otimes n_1} \otimes \qty{I,Z}^{\otimes n_2} \otimes \qty{I,X,Y,Z}^{\otimes n_3}) W.
    \end{equation}
    Therefore, $P\in \mathcal{Q}_{U}$ if and only if all $Q\in\mathcal{Q}_{U}^{\mathrm{com}}$ commute with $P$, establishing Eq.~\eqref{eq_proof_2_2}.
\end{proof}

\subsection{Proof of Theorem~\ref{thm_3'}}
Then, we provide the proof of Theorem~\ref{thm_3'}.
\begin{proof}
    Let $U_{\mathrm{E}}$ be the initial half of the SCV gadget $\Theta_S^{\mathrm{det}}$ depicted in Fig.~\ref{fig_SCV}.
    Then, we have
    \begin{equation}
        \label{eq_proof_3'_1}
        U_{\mathrm{E}} \ket{\psi}\ket{0^m} = \frac{1}{2^{m/2}}\sum_{jk}\exp[\frac{2\pi i}{2^m}jk]\Pi_j\ket{\psi}\ket{k}
    \end{equation}
    for $n$-qubit quantum state $\ket{\psi}$.
    Using this $U_{\mathrm{E}}$ and some $n+m$ qubit unitary $U_{\mathrm{D}}$, we define a quantum superchannel $\Theta^{\mathrm{cor}}_S$ as in Eq.~\eqref{eq_purification_correction}.
    Since $U_{\mathrm{E}}$ commutes with $U$, we have
    \begin{equation}
        \label{eq_proof_3'_2}
        \Theta^{\mathrm{cor}}_S (\mathcal{U}_{\mathcal{N}}) = (\Theta^{\mathrm{cor}}_S (\mathcal{N}))\circ \mathcal{U}.
    \end{equation}
    
    Let us further define a code space $\mathcal{C}$ as a subspace of $(n+m)$-qubit system spanned by $U_{\mathrm{E}} \ket{\psi}\ket{0^m}$, where $\ket{\psi}$ is an arbitrary $n$-qubit quantum state.
    Then, from the quantum error-correction condition~\cite{nielsen2010quantum, gottesman2016surviving}, there exists $U_{\mathrm{D}}$ such that $\Theta^{\mathrm{cor}}_S(\mathcal{N}) = \mathcal{I}$ if and only if for all $n$-qubit quantum state $\ket{\psi}$ and $\ket{\phi}$ and $j,k\in\{0,\ldots,K\}$,
    \begin{equation}
        \label{eq_proof_3'_3}
         \bra{\psi} \bra{0^m} U_{\mathrm{E}}^\dag (N_j^\dag N_k\otimes I_{\mathrm{a}})U_{\mathrm{E}} \ket{\phi}\ket{0^m} = C_{ij} \braket{\psi}{\phi}.
    \end{equation}
    Here, $C_{ij}$ is a complex number independent of $\ket{\psi}$ and $\ket{\phi}$ and $\mathcal{I}$ is a identity channel on the $n$-qubit system.
    
    From Eq.~\eqref{eq_proof_3'_1}, the left-hand side of Eq.~\eqref{eq_proof_3'_3} can be represented as
    \begin{equation}
        \label{eq_proof_3'_4}
        \sum_i \bra{\psi} \Pi_iN_j^\dag N_k\Pi_i \ket{\phi}.
    \end{equation}
    Therefore,  Eq.~\eqref{eq_proof_3'_3} is equivalent to 
    \begin{equation}
        \sum_i \Pi_iN_j^\dag N_k\Pi_i \propto I_n,
    \end{equation}
    which completes the proof.
\end{proof}

\subsection{Proof of Corollary~\ref{cor_2}}
Then, we proceed with the proof of Corollary~\ref{cor_2}.

\begin{proof}
    From Theorem~\ref{thm_3'}, for Pauli noise $\mathcal{N}(\cdot) = \sum_{i=0}^K p_i P_i\cdot P_i$, we have $\Theta^{\mathrm{cor}}_{\mathcal{Q}^{\mathrm{com}}_{U}}(\mathcal{U}_{\mathcal{N}}) = \mathcal{U}$ if and only if for all $j,k\in\{0,\ldots.K\}$,
    \begin{equation}
        \label{eq_proof_4'_1}
        \sum_{i_1\cdots i_r} \Pi_{i_1}^1 \cdots \Pi_{i_r}^r P_jP_k \Pi_{i_r}^r \cdots \Pi_{i_1}^1 \propto I_n,
    \end{equation}
    where $\Pi_{i_k}^k = \frac{1}{2}(I_n + (-1)^{i_k} Q_{k})$.
    For $P_jP_k\neq I_n$, Eq.~\eqref{eq_proof_4'_1} holds if and only if there exists $Q\in\mathcal{Q}_U^{\mathrm{com}}$ that anti-commutes with $P_jP_k$.
    From Eq.~\eqref{eq_proof_2_2} in the proof of Corollary~\ref{cor_1}, this is equivalent to 
    \begin{equation}
        \label{eq_proof_4'_2}
        \pm P_jP_k, \pm iP_jP_k \notin \mathcal{Q}_U.
    \end{equation}
    Note that we have added a phase because $P_jP_k$ may not be a Pauli operator in the set $\mathcal{P}_n = \{I,X,Y,Z\}^{\otimes n}$ without the phase.
    Therefore, we arrive at the conclusion of Corollary~\ref{cor_2}.
\end{proof}

\subsection{Proof of Theorem~\ref{thm_3}}
Then, we provide the proof of Theorem~\ref{thm_3}.

\begin{proof}
    We first prove the necessary condition.
    Let us define a code space $\mathcal{C}$ as a subspace of $(n+m)$-qubit system spanned by $U_{\mathrm{E}} \ket{\psi}\ket{0^m}$, where $\ket{\psi}$ is an arbitrary $n$-qubit quantum state.
    We define the stabilizer group of the code space $\mathcal{C}$ as $\mathcal{S} = U_E (\qty{I}^{\otimes n } \otimes \qty{I,Z}^{\otimes m})U_E^\dag$, a set of logical Pauli operators as $\bar{\mathcal{P}}_n = U_E (\mathcal{P}_n \otimes \qty{I}^{\otimes m})U_E^\dag$, and a set $N(\mathcal{S}) = U_E (\mathcal{P}_n \otimes \qty{I,Z}^{\otimes m})U_E^\dag$.
    We also define sets with additional phase of $\pm 1$ as $\mathcal{S}^{\pm} = \qty{\pm1}\cdot\mathcal{S}$ and $N(\mathcal{S})^{\pm} = \qty{\pm1}\cdot N(\mathcal{S})^{\pm}$.
    Note that $N(\mathcal{S})^{\pm} = \bar{\mathcal{P}}_n\cdot\mathcal{S}^{\pm} = \qty{\pm1}\cdot\qty{P\in\mathcal{P}_{n+m}\;|\;\forall Q\in\mathcal{S}, [P,Q]=0}$.

    When we set $\theta = 0$, we have
    \begin{equation}
        \label{eq_proof_3_1}
        \Theta^{\mathrm{det}}_{\mathrm{Clif}}(\mathcal{N}_{P_i}) = \delta_{i0}\mathcal{I}_n.
    \end{equation}
    This means that the code space $\mathcal{C}$ can detect the set of errors $\{P_i\otimes I^{\otimes m}\}_{i=1}^{K}$.
    Therefore, from the quantum error-detection condition~\cite{nielsen2010quantum, gottesman2016surviving}, we obtain
    \begin{equation}
        \label{eq_proof_3_2}
        \forall P_i \in \mathcal{P}_{\mathcal{N}}^{\mathrm{det}},\; P_i\otimes I^{\otimes m} \notin  N(\mathcal{S})^{\pm} \backslash \mathcal{S}^{\pm}.
    \end{equation}
    
    When we consider the case where the error did not occur, we have
    \begin{equation}
        \label{eq_proof_3_3}
        \Theta^{\mathrm{det}}_{\mathrm{Clif}}(\mathcal{U}) = \mathcal{U}.
    \end{equation}
    For $\theta = 0$, this means
    \begin{equation}
        \label{eq_proof_3_4}
        U_{\mathrm{D}}U_{\mathrm{E}} \ket{\psi}\ket{0^m} = \ket{\psi}\ket{0^m}
    \end{equation}
    for all input state $\ket{\psi}$.
    In other words,
    \begin{equation}
        \label{eq_proof_3_4_2}
        U_{\mathrm{D}}U_{\mathrm{E}} (I^{\otimes n} \otimes \ketbra{0^m})  = I^{\otimes n} \otimes \ketbra{0^m},
    \end{equation}
    so the Clifford unitary $U_{\mathrm{D}}$ satisfies
    \begin{equation}
        \label{eq_proof_3_5}
        U_{\mathrm{E}}(I^{\otimes n} \otimes \ketbra{0^m})  = U_{\mathrm{D}}^\dag(I^{\otimes n} \otimes \ketbra{0^m}).
    \end{equation}
    By combining Eq.~\eqref{eq_proof_3_3} with Eq.~\eqref{eq_proof_3_5}, we obtain
    \begin{equation}
        \label{eq_proof_3_6}
        \begin{aligned}
            &(I^{\otimes n}\otimes\ketbra{0^m})U_{\mathrm{E}}^\dag (e^{i\theta H}\otimes I^{\otimes m})U_{\mathrm{E}} (I^{\otimes n}\otimes\ketbra{0^m})\\
            =&(e^{i\theta H}\otimes\ketbra{0^m}).
        \end{aligned}
    \end{equation}
    When we differentiate this equation by $\theta$ and substituting $\theta = 0$, we have
    \begin{equation}
        \label{eq_proof_3_7}
        \begin{aligned}
            &(I^{\otimes n}\otimes\ketbra{0^m})U_{\mathrm{E}}^\dag (H\otimes I^{\otimes m})U_{\mathrm{E}} (I^{\otimes n}\otimes\ketbra{0^m})\\
            =&(H\otimes\ketbra{0^m}).
        \end{aligned}
    \end{equation}
    By sandwiching this equality with $U_{\mathrm{E}}$ and $U_{\mathrm{E}}^\dag$, we obtain
    \begin{equation}
        \label{eq_proof_3_8}
        \Pi (H \otimes I^{\otimes m}) \Pi = \Pi\bar{H}\Pi,
    \end{equation}
    where $\Pi = U_{\mathrm{E}}(I^{\otimes n}\otimes\ketbra{0^m})U_{\mathrm{E}}^\dag$ is the projector to the code space $\mathcal{C}$ and $\bar{H} = U_E (H \otimes I^{\otimes m})U_E^\dag$ is the logical $H$ operator on $\mathcal{C}$.

    Let us decompose the Hamiltonian $H$ using Pauli operator $Q_i \in \mathcal{P}_n$ as $H = \sum_{i=1}^A h_i Q_i$, where $h_i \neq 0$. 
    Then, from Eq.~\eqref{eq_proof_3_8}, we have
    \begin{equation}
        \label{eq_proof_3_9}
        \sum_{i=1}^A h_i \Pi(Q_i\otimes I^{\otimes m})\Pi = \sum_{i=1}^A h_i \Pi \bar{Q}_i\Pi,
    \end{equation}
    where $\bar{Q}_i = U_E (Q_i\otimes I^{\otimes m})U_E^\dagger\in \bar{\mathcal{P}}_{n}$ is the logical $Q_i$ operator on the code space $\mathcal{C}$.
    Since the number of Pauli terms in the right-hand side and left-hand side of Eq.~\eqref{eq_proof_3_9} are both $A$, $Q_i\otimes I^{\otimes m}$ should be a logical Pauli operator on the code space $\mathcal{C}$. 
    Therefore, there exists $S_1,\ldots,S_A\in\mathcal{S}^{\pm}$ such that
    \begin{equation}
        \label{eq_proof_3_10}
        \qty{Q_i \otimes I^{\otimes m}}_{i=1}^A = \qty{\bar{Q}_i S_i}_{i=1}^A.
    \end{equation}
    
    When we consider the set of Pauli operators generated from this set, we have
    \begin{equation}
        \label{eq_proof_3_11}
        \mathcal{Q}_H \otimes \{I^{\otimes m}\} \subset \bar{\mathcal{P}}_n \cdot \mathcal{S}^{\pm} = N(\mathcal{S})^{\pm}.
    \end{equation}
    Here, the definition of $\mathcal{Q}_H$ is given in Eq.~\eqref{eq_Q_U}.
    Let us assume that there exists $Q\in\mathcal{Q}_H\backslash\{I^{\otimes n}\}$ such that $Q\otimes I^{\otimes m} = S \in \mathcal{S}^{\pm}$.
    In this case, we can show that $\abs{\mathcal{Q}_H} \geq \abs{\mathcal{Q}_H} + 1$, which is a contradiction.
    To show this, let us define the logical version of the set $\mathcal{Q}_H$ as $\bar{\mathcal{Q}}_H = U_E (\mathcal{Q}_H \otimes {I}^{\otimes m})U_E^\dagger \subset \bar{\mathcal{P}}_n$.
    Then, from Eq.~\eqref{eq_proof_3_10}, for all $\bar{Q}' \in \bar{\mathcal{Q}}_H$, there exists different $Q'' \in \mathcal{Q}_H$ such that $Q'' \otimes I^{\otimes n} = \bar{Q}' S'$ for some stabilizer $S'\in\mathcal{S}^{\pm}$.
    However, under the assumption of the existence of $Q\otimes I^{\otimes m} = S \neq I^{\otimes(n+m)}$, there exists at least two different $Q'' \in \mathcal{Q}_H$ for $\bar{Q}' = I^{\otimes n+m}$.
    This means $\abs{\mathcal{Q}_H} \geq \abs{\bar{\mathcal{Q}}_H} + 1 = \abs{\mathcal{Q}_H} + 1$, so we have $(\mathcal{Q}_H\otimes \{I^{\otimes m}\}\backslash\{I^{\otimes (n+m})\}  )\cap \mathcal{S}^{\pm} = \emptyset$. 
    Thus, we obtain 
    \begin{equation}
        \label{eq_proof_3_13}
        (\mathcal{Q}_H \backslash \qty{I^{\otimes n}})\otimes \qty{I^{\otimes m}} \subset N(\mathcal{S})^{\pm} \backslash \mathcal{S}^{\pm}.
    \end{equation}

    By combining Eq.~\eqref{eq_proof_3_2} and \eqref{eq_proof_3_13}, we obtain
    \begin{equation}
        \label{eq_proof_3_14}
        \forall P_i\in \mathcal{P}_{\mathcal{N}}^{\mathrm{det}},\; P_i \notin \mathcal{Q}_{H}\backslash\qty{I^{\otimes n}},
    \end{equation}
    which is equivalent to 
    \begin{equation}
        \mathcal{P}^{\mathrm{{det}}}_{\mathcal{N}} \cap \mathcal{Q}_{H} = \emptyset.
    \end{equation}
    
    For the sufficiency, we use SCV gadget $\Theta^{\mathrm{det}}_{\mathcal{Q}^{\mathrm{com}}_{H}}$.
    From Corollary~\ref{cor_1}, $\Theta^{\mathrm{det}}_{\mathcal{Q}^{\mathrm{com}}_{H}}$ can be used to detect errors satisfying Eq.~\eqref{eq_thm_3_2}.
\end{proof}

\subsection{Proof of Theorem~\ref{thm_4}}
 Finally, we prove Theorem~\ref{thm_4}.

 \begin{proof}
    We first prove the necessary condition.
    As in the proof of Theorem~\ref{thm_3}, let us define a code space $\mathcal{C}$ as a subspace of $(n+m)$-qubit system spanned by $U_{\mathrm{E}} \ket{\psi}\ket{0^m}$, where $\ket{\psi}$ is an arbitrary $n$-qubit quantum state.
    We define the stabilizer group of the code space $\mathcal{C}$ as $\mathcal{S} = U_E (\qty{I}^{\otimes n } \otimes \qty{I,Z}^{\otimes m})U_E^\dag$, a set of logical Pauli operators as $\bar{\mathcal{P}}_n = U_E (\mathcal{P}_n \otimes \qty{I}^{\otimes m})U_E^\dag$, and a set $N(\mathcal{S}) = U_E (\mathcal{P}_n \otimes \qty{I,Z}^{\otimes m})U_E^\dag$.
    We also define sets with additional phase of $\pm 1$ as $\mathcal{S}^{\pm} = \qty{\pm1}\cdot\mathcal{S}$ and $N(\mathcal{S})^{\pm} = \qty{\pm1}\cdot N(\mathcal{S})^{\pm}$.
    Note that $N(\mathcal{S})^{\pm} = \bar{\mathcal{P}}_n\cdot\mathcal{S}^{\pm} = \qty{\pm1}\cdot\qty{P\in\mathcal{P}_{n+m}\;|\;\forall Q\in\mathcal{S}, PQ=QP}$.

    When we set $\theta = 0$, we have
    \begin{equation}
        \label{eq_proof_4_1}
        \Theta^{\mathrm{cor}}_{\mathrm{Clif}}(\mathcal{N}) = \mathcal{I},
    \end{equation}
    where $\mathcal{I}$ is an identity operator of an $n$-qubit system.
    This means that the code space $\mathcal{C}$ can correct the noise $\mathcal{N}\otimes \mathcal{I}_a(\cdot) = \sum_{i=0}^K p_i (P_i \otimes I^{\otimes m})\cdot (P_i \otimes I^{\otimes m})$.
    Therefore, from the quantum error-correction condition~\cite{nielsen2010quantum, gottesman2016surviving}, we obtain
    \begin{equation}
        \label{eq_proof_4_2}
        \forall P\in \mathcal{P}^{\mathrm{cor}}_{\mathrm{N}},\; P\otimes I^{\otimes m} \notin  N(\mathcal{S})^{\pm} \backslash \mathcal{S}^{\pm}.
    \end{equation}

    Since no error occurs with probability $p_0 \neq 0$, we have
    \begin{equation}
        \label{eq_proof_4_3}
        \Theta^{\mathrm{cor}}_{\mathrm{Clif}}(\mathcal{U}) = \mathcal{U}.
    \end{equation}
    For $\theta = 0$, this means
    \begin{equation}
        \label{eq_proof_4_4}
        U_{\mathrm{D}}U_{\mathrm{E}} \ket{\psi}\ket{0^m} = \ket{\psi}\ket{\alpha}
    \end{equation}
    for all input state $\ket{\psi}$, where $\ket{\alpha}$ is a $m$-qubit quantum state independent of $\ket{\psi}$.
    Since the operation performed on the ancilla does not affect the system qubits, we can modulate $U_{\mathrm{D}}$ so that $\ket{\alpha} = \ket{0^{m}}$.
    In this case, we have
    \begin{equation}
        \label{eq_proof_4_5}
        (I^{\otimes n} \otimes \ketbra{0^m})U_{\mathrm{D}}U_{\mathrm{E}}  = I^{\otimes n} \otimes \ketbra{0^m},
    \end{equation}
    so the Clifford unitary $U_{\mathrm{D}}$ satisfies
    \begin{equation}
        \label{eq_proof_4_6}
        U_{\mathrm{E}}(I^{\otimes n} \otimes \ketbra{0^m})  = U_{\mathrm{D}}^\dag(I^{\otimes n} \otimes \ketbra{0^m}).
    \end{equation}
    From Eq.~\eqref{eq_proof_4_3} and Eq.~\eqref{eq_proof_4_6}, when we consider  post-selecting measurement results instead of discarding them, we obtain
    \begin{equation}
        \label{eq_proof_4_7}
        \begin{aligned}
            &(I^{\otimes n}\otimes\ketbra{0^m})U_{\mathrm{E}}^\dag (e^{i\theta H}\otimes I^{\otimes m})U_{\mathrm{E}} (I^{\otimes n}\otimes\ketbra{0^m})\\
            \propto&(e^{i\theta H}\otimes\ketbra{0^m}).
        \end{aligned}
    \end{equation}
    Then, in the same way as in the proof of Theorem~\ref{thm_3}, we obtain
    \begin{equation}
        \label{eq_proof_4_8}
        (\langle \mathcal{Q}_H \rangle \backslash \qty{I^{\otimes n}})\otimes \qty{I^{\otimes m}} \subset N(\mathcal{S})^{\pm} \backslash \mathcal{S}^{\pm}.
    \end{equation}
    By combining Eq.~\eqref{eq_proof_4_2} and \eqref{eq_proof_4_8}, we obtain
    \begin{equation}
        \label{eq_proof_4_9}
        \forall P\in \mathcal{P}^{\mathrm{cor}}_{\mathrm{N}},\; P \notin \mathcal{Q}_{H}\backslash\qty{I^{\otimes n}},
    \end{equation}
    which is equivalent to 
    \begin{equation}
        \label{eq_proof_4_10}
        \mathcal{P}^{\mathrm{{cor}}}_{\mathcal{N}} \cap \mathcal{Q}_{H} = \emptyset.
    \end{equation}

    For the sufficiency, we use SCV gadget $\Theta^{\mathrm{cor}}_{\mathcal{Q}^{\mathrm{com}}_{H}}$.
    From Theorem~\ref{thm_4}, this can be used to detect errors satisfying Eq.~\eqref{eq_cor_2_1}.
 \end{proof}

\section{Resource-Theoretic Analysis on the Limitations of Channel Purification}
\label{sec_resource}
In this section, we provide the proof of Theorem~\ref{thm_5} by employing the fundamental tools of the dynamical resource theory of magic~\cite{seddon2019quantifying,regula2021fundamental,Fang2022nogo}.
The resource theory of magic~\cite{Veitch2014resource} is the framework that allows us to quantitatively analyze the properties of the states that cannot be prepared by a stabilizer operation.
Stabilizer operations are operations consisting of (i) preparing states $\ket{0}$, (ii) applying Clifford unitaries, (iii) measuring with the computational basis, and (iv) applying the above operations conditioned on the measurement outcomes and classical randomness.
The set of states that can be prepared by stabilizer operations are called stabilizer states, and it coincides with the probabilistic mixture of pure states prepared by Clifford unitaries.
Here, we denote the set of stabilizer states as $\mathbb{F}$.

In resource-theoretic analysis, one considers the set of operations accessible in the given operational setting. In the case of the magic resource theory, the most operational choice is to take the set of stabilizer operations as such freely available operations. However, the set of stabilizer operations is theoretically hard to deal with because of its complex mathematical structure. 
It turns out that it is useful to consider a larger set of operations that contains stabilizer operations as a subset. The standard choice is to consider the set of completely stabilizer-preserving channels defined by
\begin{equation}
    \mathbb{O} = \qty{\mathcal{E}\in\mathrm{CPTP} ~|~ \forall m \in \mathbb{N}, \forall \sigma\in\mathbb{F}, \mathcal{E}\otimes\mathcal{I}_m(\sigma) \in\mathbb{F}},
\end{equation}
where $\mathrm{CPTP}$ denotes the set of completely positive trace-preserving maps, and $\mathcal{I}_m$ represents the identity channel on an $m$-qubit system.
The advantage of considering this set of operations is that it can be concisely characterized by looking at the magic property of their Choi states.
Namely, it is known that a channel $\mathcal{E}$ is completely stabilizer preserving if and only if its Choi state $J_{\mathcal{E}} = \mathcal{E}\otimes\mathcal{I}(\ketbra{\Psi})$ is a stabilizer state~\cite{seddon2019quantifying}:
\begin{equation}
    \label{eq_CSP_channel}
    \mathbb{O} = \qty{\mathcal{E}\in\mathrm{CPTP} ~|~ J_{\mathcal{E}} \in\mathbb{F}}.
\end{equation}
Here, $\ket{\Psi}= 2^{-n/2}\sum_i \ket{i}\ket{i}$ is the maximally entangled state with $\{\ket{i}\}$ being a computational basis, and $\mathcal{I}$ is the identity channel of the same dimension.

Toward the application to our setting of SCV, we now extend this framework to the dynamical resource theory of magic, where we consider the transformation of quantum channels by available processes. 
In general, the transformation of quantum channels is realized by quantum superchannels~\cite{Chiribella2008transforming,Chiribella2008quantum}, which can be implemented by applying pre- and post-processing channels before and after the given channel. 
The dynamical resource theory allows us to analyze the consequences of limiting the available superchannels.

In the setting discussed in Sec.~\ref{sec_Clifford_purification}, the allowed set of superchannels consists of pre- and post-processing Clifford unitaries followed by the partial trace, represented as
\begin{equation}
    \label{eq_purification_correction_3}
    \Theta_{\mathrm{Clif}}^{\mathrm{cor}}(\mathcal{E}):~\rho \mapsto \mathrm{tr}_{\mathrm{a}}[\mathcal{U}_{\mathrm{D}}\circ(\mathcal{E}\otimes \mathcal{I}_{\mathrm{a}})\circ\mathcal{U}_{\mathrm{E}}(\rho\otimes\ketbra{0^m})],
\end{equation}
where $\mathcal{U}_{\mathrm{E}}$ and $\mathcal{U}_{\mathrm{D}}$ are Clifford unitary channels and $\mathcal{I}_{\mathrm{a}}$ is the identity channel on arbitrary dimension.
One can easily check that these superchannels surely preserve the set of completely stabilizer-preserving channels---if we sandwich an arbitrary channel in $\mathbb{O}$ by Clifford unitaries, then the resulting channel after the partial trace is also in $\mathbb{O}$.  
This means that if we consider the set of stabilizer-preserving superchannels defined by 
\begin{equation}
    \label{eq_stab_superchanel}
    \mathbb{S} = \qty{\Theta\in\mathrm{superchannel} ~|~ \forall \mathcal{E} \in \mathbb{O}, \Theta (\mathcal{E})\in\mathbb{O}},
\end{equation}
the superchannels $\Theta_{\mathrm{Clif}}^{\mathrm{cor}}$ of the form \eqref{eq_purification_correction_3} considered in Sec.~\ref{sec_Clifford_purification} are included in $\mathbb{S}$.
This implies that, while the set $\mathbb{S}$ is strictly larger than the set of superchannels $\Theta_{\mathrm{Clif}}^{\mathrm{cor}}$ constructed from Clifford unitaries, limitations on channel purification using axiomatic superchannels $\Theta\in\mathbb{S}$ also impose fundamental limitations on operationally motivated superchannels $\Theta_{\mathrm{Clif}}^{\mathrm{cor}}$.
Therefore, we derive the fundamental limits of channel purification using stabilizer-preserving superchannels $\Theta\in\mathbb{S}$.

To achieve this, we employ the resource robustness $R_{\mathbb{F}}(\rho)$~\cite{vidal1999robustness, brandao2015reversible} and resource weight $W_{\mathbb{F}}(\rho)$~\cite{lewenstein1998separability, uola2020all} of a quantum state $\rho$, defined as
\begin{align}
    \label{eq_resource_robustness}
    R_{\mathbb{F}}(\rho) &= \min\qty{\lambda\geq1 ~|~ \exists\sigma\in\mathbb{F},~ \rho\leq \lambda\sigma}, \\
    \label{eq_resource_weight}
    W_{\mathbb{F}}(\rho) &= \max\qty{\lambda\leq1 ~|~ \exists\sigma\in\mathbb{F},~ \rho\geq \lambda\sigma}.    
\end{align}
Similarly, we define the resource robustness $R_{\mathbb{O}}(\mathcal{E})$~\cite{Diaz2018usingreusing,takagi2019general} and resource weight $W_{\mathbb{O}}(\mathcal{E})$~\cite{uola2020all} of a quantum channel $\mathcal{E}$ as
\begin{align}
    \label{eq_channel_robustness}
    R_{\mathbb{O}}(\mathcal{E}) &= \min\qty{\lambda\geq1 ~|~ \exists\mathcal{M}\in\mathbb{O},~ J_{\mathcal{E}}\leq \lambda J_{\mathcal{M}}}, \\
    \label{eq_channel_weight}
    W_{\mathbb{O}}(\mathcal{E}) &= \max\qty{\lambda\leq1 ~|~ \exists\mathcal{M}\in\mathbb{O},~ J_{\mathcal{E}}\geq \lambda J_{\mathcal{M}}}.
\end{align}
These measures quantify how far a quantum state $\rho$ and a quantum channel $\mathcal{E}$ deviate from the sets of stabilizer states $\mathbb{F}$ and completely stabilizer-preserving channels $\mathbb{O}$, respectively.
Additionally, we introduce the fidelity-based measure
\begin{align}
    F_{\mathbb{O}}(\mathcal{E}) = \max_{\mathcal{M}\in\mathbb{O}} F(\mathcal{E},\mathcal{M}),
\end{align}
where $F(\mathcal{E},\mathcal{M})$ is the worst-case fidelity defined in Eq.~\eqref{eq_worst_case_fid}.

Using these resource measures, Ref.~\cite{regula2021fundamental} derived an upper bound on the worst-case fidelity achievable through channel purification with stabilizer-preserving superchannels $\Theta\in\mathbb{S}$:
\begin{lem}
    \label{lem_1}
    (Theorem~1 of Ref.~\cite{regula2021fundamental}) Let $\mathcal{U}$ be a unitary channel, and let $\mathcal{U}_{\mathcal{N}} = \mathcal{N}\circ\mathcal{U}$ be a noisy channel affected by a noise channel $\mathcal{N}$. Then, the worst-case fidelity between the ideal channel $\mathcal{U}$ and the purified noisy channel $\Theta(\mathcal{U}_{\mathcal{N}})$ under a stabilizer-preserving superchannel $\Theta\in\mathbb{S}$ is upper bounded as
    \begin{align}
        F(\Theta(\mathcal{U}_{\mathcal{N}}), \mathcal{U}) &\leq F_{\mathbb{O}}(\mathcal{U})R_{\mathbb{O}}(\mathcal{U}_{\mathcal{N}}), \\
        F(\Theta(\mathcal{U}_{\mathcal{N}}), \mathcal{U}) &\leq 1 - (1-F_{\mathbb{O}}(\mathcal{U}))W_{\mathbb{O}}(\mathcal{U}_{\mathcal{N}}).\label{eq:weight bound}
    \end{align}
\end{lem}
We remark that \eqref{eq:weight bound} was also shown in Ref.~\cite{Fang2022nogo}.
Lemma~\ref{lem_1} ensures that, if we can evaluate the resource measures $R_{\mathbb{O}}(\mathcal{U}_{\mathcal{N}})$, $W_{\mathbb{O}}(\mathcal{U}_{\mathcal{N}})$, and $F_{\mathbb{O}}(\mathcal{U})$, we can derive an upper bound on the worst-case fidelity.
However, analytical computation of these measures is usually not feasible.
The examples known to come with analytical expressions have been limited to the case of the channels implementable by gate teleportation with a magic state, where the magic measures for channels can be reduced to those of the corresponding magic states~\cite{regula2021fundamental}.
The channel $\mathcal{U}_{\mathcal{N}}$ of our interest is generally not in such a class, preventing us from directly applying the previous approach.

Here, we present a new approach to computing these resource measures and obtain analytical expressions for an operationally important class of noisy unitary channels.
To this end, we consider the Pauli-$Z$ rotation gate $\mathcal{U}(\cdot) = e^{-i\frac{\theta}{2} Z}\cdot e^{i\frac{\theta}{2} Z}$ under single-qubit Pauli noise $\mathcal{N}(\cdot) = p_0 \cdot +p_x X\cdot X + p_y Y\cdot Y + p_z Z\cdot Z$.
In this case, as shown in the following lemma, the resource robustness and resource weight for quantum the counterparts of a quantum state:
\begin{lem}
    \label{lem_2}
    Let $\mathcal{U}(\cdot) = e^{-i\frac{\theta}{2} Z}\cdot e^{i\frac{\theta}{2} Z}$ be a Pauli-$Z$ rotation gate with $0\leq\theta\leq\pi/4$, and let $\mathcal{U}_{\mathcal{N}} = \mathcal{N}\circ\mathcal{U}$ be a noisy channel affected by single-qubit Pauli noise $\mathcal{N}(\cdot) = p_0 \cdot +p_x X\cdot X + p_y Y\cdot Y + p_z Z\cdot Z$. Then, the resource robustness and the resource weight of the noisy channel $\mathcal{U}_{\mathcal{N}}$ is given by
    \begin{align}
        \label{eq_lem2_1}
        \begin{aligned}
            R_{\mathbb{O}}(\mathcal{U}_{\mathcal{N}}) &= (p_0+p_z)R\qty(\theta,\frac{p_z}{p_0+p_z}) \\
        &+ (p_x+p_y)R\qty(\theta, \frac{p_y}{p_x+p_y}),
        \end{aligned} \\
        \begin{aligned}
            W_{\mathbb{O}}(\mathcal{U}_{\mathcal{N}}) &= (p_0+p_z)W\qty(\theta,\frac{p_z}{p_0+p_z}) \\
        &+ (p_x+p_y)W\qty(\theta, \frac{p_y}{p_x+p_y}).
        \end{aligned}
    \end{align}
    Here, $R(\theta, p)$ and $W(\theta, p)$ represents the resource robustness and the resource weight of the quantum state $\rho(\theta, p)=\frac{1}{2}(I + (1-2p)(\cos\theta X + \sin \theta Y))$, given by
    \begin{align}
        \label{eq_lem2_2_1}
        R(\theta, p) &=
        \begin{cases}
            1 & (y \leq 1-x)\\
            \frac{\sqrt{2}+x+y}{\sqrt{2}+1} & (y \geq 1-x)
        \end{cases},\\
        \label{eq_lem2_2_2}
        W(\theta, p) &=
        \begin{cases}
            1 & (y \leq 1-x)\\
            \frac{\sqrt{2}-(x+y)}{\sqrt{2}-1} & (1-x \leq y \leq (\sqrt{2}+1)(1-x))\\
            \frac{1-(x^2+y^2)}{2(1-x)} & (y> (\sqrt{2}+1)(1-x))
        \end{cases},
    \end{align}
    where $x = \abs{1-2p}\cos\theta$ and $y = \abs{1-2p}\sin\theta$.
\end{lem}

\begin{proof}
    Since the proof for the resource weight is essentially the same as that for the resource robustness, we concentrate on the proof for the resource robustness.
    We first show Eq.~\eqref{eq_lem2_1}, and then derive Eq.~\eqref{eq_lem2_2_1} and Eq.~\eqref{eq_lem2_2_2}.
    
    From Eq.~\eqref{eq_CSP_channel} and Eq.~\eqref{eq_channel_robustness}, resource robustness $R_{\mathbb{O}}(\mathcal{U}_{\mathcal{N}})$ can be represented as
    \begin{equation}
        \label{eq_proof_b2_1}
        R_{\mathbb{O}}(\mathcal{U}_{\mathcal{N}}) = \min\qty{\lambda\geq1 ~|~ \exists\sigma\in\mathbb{F},~ \mathrm{tr}_{\mathrm{S}}[\sigma]=I/2,~ J_{\mathcal{U}_{\mathcal{N}}}\leq \lambda \sigma}.
    \end{equation}
    Here, $\mathrm{tr}_{\mathrm{S}}$ represents the partial trace over the subspace where $\mathcal{U}_{\mathcal{N}}$ acts, and $\mathrm{tr}_{\mathrm{S}}[\sigma]=I/2$ represents the trace-preserving condition of the channel $\mathcal{M}\in\mathbb{O}$.
    Since the maximally entangled state can be represented as
    \begin{equation}
        \label{eq_proof_b2_2}
        \ketbra{\Psi} = \Pi_0 \rho_L(0, 0)\Pi_0,
    \end{equation}
    where $\Pi_i = \frac{1}{2}(II + (-1)^iZZ)$ is the projector onto the eigenspace of $ZZ$ with eigenvalue $(-1)^i$ and $\rho_L(\theta, p) = \frac{1}{2}(II + (1-2p)(\cos\theta XX+\sin\theta YX))$, the Choi state $J_{\mathcal{U}_{\mathcal{N}}}$ can be expressed as
    \begin{equation}
        \begin{aligned}
            J_{\mathcal{U}_{\mathcal{N}}} 
            &= (p_0+p_z)\Pi_0 \rho_L\qty(\theta, \frac{p_z}{p_0+p_z})\Pi_0 \\
            &+ (p_x+p_y)\Pi_1 \rho_L\qty(-\theta, \frac{p_y}{p_x+p_y})\Pi_1.
        \end{aligned}
    \end{equation}
    We note that $\Pi_i \rho_L(\theta, p)\Pi_i$ can be interpreted as a logical state on the $n=2$ repetition code, where the quantum state $\rho(\theta, p) = \frac{1}{2}(I + (1-2p)(\cos\theta X+\sin\theta Y))$ is encoded.
    Below, we denote $q_0 = p_0 + p_z$, $q_1 = p_x + p_y$, $\rho_0 = \rho(\theta, \frac{p_z}{p_0+p_z})$, $\rho_1 = \rho(-\theta, \frac{p_y}{p_x+p_y})$, $\rho_{0,L} = \rho_L(\theta, \frac{p_z}{p_0+p_z})$ and $\rho_{1,L} = \rho_L(-\theta, \frac{p_y}{p_x+p_y})$ for simplicity.
    Under this notation, we have
    \begin{equation}
        \label{eq_proof_b2_3}
        J_{\mathcal{U}_{\mathcal{N}}} = q_0 \Pi_0\rho_{0,L}\Pi_0 + q_1 \Pi_1\rho_{1,L}\Pi_1.
    \end{equation}

    For single-qubit quantum state $\rho_{i}$ and the corresponding logical state $\Pi_i\rho_{i,L}\Pi_i$, let us show
    \begin{equation}
        \label{eq_proof_b2_4}
        R_{\mathbb{F}}(\Pi_i\rho_{i,L}\Pi_i) 
        = R_{\mathbb{F}}(\rho_i).
    \end{equation}
    The essential point is that for any single-qubit quantum state $\sigma_{i}=\frac{1}{2}(I + x_i X+y_i Y +  z_i Z)$ and the corresponding logical state $\Pi_i\sigma_{i,L}\Pi_i$ with $\sigma_{i,L} = \frac{1}{2}(I + x_i XX+y_i YX +  z_i ZI)$,
    two inequalities
    \begin{align}
        \rho_i &\leq \lambda \sigma_i, \\
        \Pi_i\rho_{i,L}\Pi_i &\leq \lambda \Pi_i \sigma_{i,L}\Pi_i
    \end{align}
    are equivalent.
    This equivalence can be shown by the fact that there exist isometries $V_i$ for encoding that satisfy $V_i\rho_iV_i^\dagger = \Pi_i\rho_{i,L}\Pi_i$ and $V_i\sigma_iV_i^\dagger = \Pi_i\sigma_{i,L}\Pi_i$.
    Therefore, it is sufficient to show that the optimality in Eq.~\eqref{eq_resource_robustness} can be achieved with a logical state $\Pi_i\sigma_{i,L}\Pi_i$ belonging to the $(-1)^i$ eigenspace of $ZZ$.
    Let $\sigma_{i,\mathrm{opt}}\in\mathbb{F}$ be a 2-qubit stabilizer state satisfying $\Pi_i\rho_{i,L}\Pi_i \leq R_{\mathbb{F}}(\Pi_i\rho_{i,L}\Pi_i)\sigma_{i,\mathrm{opt}}$.
    Then, by sandwiching this state with $\Pi_i$, we obtain
    \begin{align}
        \label{eq_proof_b2_5}
        \Pi_i\rho_{i,L}\Pi_i 
        &\leq R_{\mathbb{F}}(\Pi_i\rho_{i,L}\Pi_i) \mathrm{tr}[\Pi_i\sigma_{i,\mathrm{opt}}\Pi_i] \frac{\Pi_i\sigma_{i,\mathrm{opt}}\Pi_i}{\mathrm{tr}[\Pi_i\sigma_{i,\mathrm{opt}}\Pi_i]}.
    \end{align}
    Since $\frac{\Pi_i\sigma_{i,\mathrm{opt}}\Pi_i}{\mathrm{tr}[\Pi_i\sigma_{i,\mathrm{opt}}\Pi_i]} \in \mathbb{F}$, we have $\mathrm{tr}[\Pi_i\sigma_{i,\mathrm{opt}}\Pi_i] = 1$ from the definition of $R_{\mathbb{F}}(\Pi_i\rho_{i,L}\Pi_i)$.
    This implies $\sigma_{i,\mathrm{opt}} = \Pi_i \sigma_{i,\mathrm{opt}} \Pi_i$, meaning that the optimal state $\sigma_{i,\mathrm{opt}}$ is always a logical state belonging to the $(-1)^i$ eigenspace of $ZZ$.
    Therefore, we have Eq.~\eqref{eq_proof_b2_4}.
    We note that the optimal state $\sigma_i\in\mathbb{F}$ achieving $\rho_i \leq R_{\mathbb{F}}(\rho) \sigma_i$ lies on the $xy$-plane within the Bloch sphere.
    Therefore, the optimal state $\Pi_i\sigma_{i,L}\Pi_i\in\mathbb{F}$ satisfying $\Pi_i\rho_{i,L}\Pi_i \leq R_{\mathbb{F}}(\Pi_i\rho_{i,L}\Pi_i)\Pi_i\sigma_{i,L}\Pi_i$ can be represented as
    \begin{equation}
        \label{eq_proof_b2_8}
        \Pi_i\sigma_{i,L}\Pi_i = \Pi_i\frac{1}{2}\qty(II + x_iXX+y_iYX)\Pi_i
    \end{equation}
    for some $x_i,y_i\in\mathbb{R}$.

    Next, let us show
    \begin{equation}
        \label{eq_proof_b2_9}
        R_{\mathbb{F}}(J_{\mathcal{U}_{\mathcal{N}}}) 
        = q_0 R_{\mathbb{F}}(\rho_{0}) + q_1 R_{\mathbb{F}}(\rho_{1}).
    \end{equation}
    Let $\sigma_{\mathrm{opt}}\in\mathbb{F}$ be a 2-qubit stabilizer state satisfying $J_{\mathcal{U}_{\mathcal{N}}} \leq R_{\mathbb{F}}(J_{\mathcal{U}_{\mathcal{N}}})\sigma_{\mathrm{opt}}$.
    By sandwiching this inequality with $\Pi_i$, we obtain
    \begin{equation}
        \label{eq_proof_b2_10}
        q_i\Pi_i\rho_{i,L}\Pi_i \leq R_{\mathbb{F}}(J_{\mathcal{U}_{\mathcal{N}}}) \mathrm{tr}[\Pi_i\sigma_{\mathrm{opt}}\Pi_i] \frac{\Pi_i\sigma_{\mathrm{opt}}\Pi_i}{\mathrm{tr}[\Pi_i\sigma_{\mathrm{opt}}\Pi_i]}.
    \end{equation}
    Since $\frac{\Pi_i\sigma_{\mathrm{opt}}\Pi_i}{\mathrm{tr}[\Pi_i\sigma_{\mathrm{opt}}\Pi_i]} \in\mathbb{F}$, this inequality implies
    \begin{equation}
        \label{eq_proof_b2_11}
        R_{\mathbb{F}}(J_{\mathcal{U}_{\mathcal{N}}}) \mathrm{tr}[\Pi_i\sigma_{\mathrm{opt}}\Pi_i] \geq q_iR_{\mathbb{F}} (\Pi_i\rho_{i,L}\Pi_i) = q_iR_{\mathbb{F}} (\rho_{i}),
    \end{equation}
    where we use Eq.~\eqref{eq_proof_b2_4}. 
    Since $\mathrm{tr}[\Pi_0\sigma_{\mathrm{opt}}\Pi_0] + \mathrm{tr}[\Pi_1\sigma_{\mathrm{opt}}\Pi_1] = 1$, we obtain
    \begin{equation}
        \label{eq_proof_b2_12}
        R_{\mathbb{F}}(J_{\mathcal{U}_{\mathcal{N}}})
        \geq q_0 R_{\mathbb{F}}(\rho_{0}) + q_1 R_{\mathbb{F}}(\rho_{1}).
    \end{equation}
    Furthermore, from Eq.~\eqref{eq_proof_b2_3} and the convexity of the resource robustness, we obtain 
    \begin{equation}
        \label{eq_proof_b2_13}
        R_{\mathbb{F}}(J_{\mathcal{U}_{\mathcal{N}}})
        \leq q_0 R_{\mathbb{F}}(\rho_{0}) + q_1 R_{\mathbb{F}}(\rho_{1}).
    \end{equation}
    From Eq.~\eqref{eq_proof_b2_12} and Eq.~\eqref{eq_proof_b2_13}, we obtain Eq.~\eqref{eq_proof_b2_9}.
    
    The optimal state $\sigma_{\mathrm{opt}}\in\mathbb{F}$ satisfying $J_{\mathcal{U}_{\mathcal{N}}} \leq R_{\mathbb{F}}(J_{\mathcal{U}_{\mathcal{N}}})\sigma_{\mathrm{opt}}$ can be represented as
    \begin{equation}
        \sigma_{\mathrm{opt}} = \frac{q_0R_{\mathbb{F}}(\rho_{0}) \Pi_0\sigma_{0,L}\Pi_0 + q_1R_{\mathbb{F}}(\rho_{1}) \Pi_1\sigma_{1,L}\Pi_1}{q_0 R_{\mathbb{F}}(\rho_{0}) + q_1 R_{\mathbb{F}}(\rho_{1})},
    \end{equation}
    where $\sigma_{i,L}$ is defined in Eq.~\eqref{eq_proof_b2_8}.
    Since this state satisfies the trace-preserving condition $\mathrm{tr}_{\mathrm{S}}[\sigma_{\mathrm{opt}}]=I/2$, we obtain
    \begin{equation}
         R_{\mathbb{O}}(\mathcal{U}_{\mathcal{N}}) = R_{\mathbb{F}}(J_{\mathcal{U}_{\mathcal{N}}}) = q_0 R_{\mathbb{F}}(\rho_{0}) + q_1 R_{\mathbb{F}}(\rho_{1})
    \end{equation}
    from Eq.~\eqref{eq_proof_b2_1} and Eq.~\eqref{eq_proof_b2_9}.

    Finally, let us derive Eq.~\eqref{eq_lem2_2_1} and Eq.~\eqref{eq_lem2_2_2}.
    The key point is that the definition of the resource robustness depicted in Eq~\eqref{eq_resource_robustness} can also be represented as
    \begin{equation}
        R_{\mathbb{F}}(\rho) = \min\qty{\lambda\geq1 ~|~ \exists \tau\in\mathbb{D}, ~ \frac{\rho+(\lambda-1)\tau}{\lambda}\in\mathbb{F}},
    \end{equation}
    where $\mathbb{D}$ is the set of density matrices (or quantum states).
    Since we consider $0\leq \theta \leq \pi/4$, we analyze the resource robustness for a quantum state $\rho = \frac{1}{2}(I + xX + yY)$ satisfying $0 \leq y \leq x$.
    When $\rho \in \mathbb{F}$, i.e., $y \leq 1-x$, we have $\frac{\rho+(1-1)\tau}{1}\in\mathbb{F}$.
    Therefore, $R_{\mathbb{F}}(\rho) = 1$.
    On the other hand, when $\rho \notin \mathbb{F}$, i.e., $y > 1-x$, the optimal state $\tau\in\mathbb{D}$ satisfying $\frac{\rho+(\lambda-1)\tau}{\lambda}\in\mathbb{F}$ can be derived from a geometrical analysis on the Bloch sphere, resulting in $\tau = \frac{1}{2}\qty(I - \frac{1}{\sqrt{2}}X - \frac{1}{\sqrt{2}}Y)$.
    Therefore, we obtain 
    \begin{equation}
        R_{\mathbb{F}}(\rho) = \frac{\sqrt{2}+x+y}{\sqrt{2}+1}.
    \end{equation}
    We note that this result coincides with the value obtained from a general formula shown in Proposition~5 of Ref.~\cite{Rubboli2024mixedstate}.

    For the resource weight, its definition depicted in Eq.~\eqref{eq_resource_weight} can also be represented as 
    \begin{equation}
        W_{\mathbb{F}}(\rho) = \max\qty{\lambda\leq1 ~|~ \exists \sigma\in\mathbb{F}, \exists\tau\in\mathbb{D}, ~ \lambda\sigma + (1-\lambda)\tau = \rho}.
    \end{equation}
    For $\rho \in \mathbb{F}$, i.e., $y \leq 1-x$, we obtain $W_{\mathbb{F}}(\rho) = 1$ by setting $\sigma=\rho$.
    For $\rho \notin \mathbb{F}$, i.e., $y > 1-x$, we employ a geometrical analysis on the Bloch sphere.
    We can show that when $y \leq (\sqrt{2}+1)(1-x)$, the optimal $\tau$ can be represented as $\tau = \frac{1}{2}\qty(I + \frac{1}{\sqrt{2}}X + \frac{1}{\sqrt{2}}Y)$, and when $y \leq (\sqrt{2}+1)(1-x)$, the optimal $\sigma$ can be represented as $\sigma = \frac{1}{2}(I+X)$.
    Therefore, when $\rho \notin \mathbb{F}$, i.e., $y > 1-x$, we obtain
    \begin{equation}
        W(\theta, p) =
        \begin{cases}
            \frac{\sqrt{2}-(x+y)}{\sqrt{2}-1} & (y \leq (\sqrt{2}+1)(1-x))\\
            \frac{1-(x^2+y^2)}{2(1-x)} & (y > (\sqrt{2}+1)(1-x))
        \end{cases}.
    \end{equation}
\end{proof}

For the fidelity-based measure $F_{\mathbb{O}}(\mathcal{U})$, we have
\begin{align}
     F_{\mathbb{O}}(\mathcal{U}) 
     &= \max_{\mathcal{M}\in\mathbb{O}}\min_{\rho} F(\mathcal{I}\otimes\mathcal{U}(\rho), \mathcal{I}\otimes\mathcal{M}(\rho)) \label{eq_fidelity_1} \\
     &\leq \max_{\mathcal{M}\in\mathbb{O}}F(J_{\mathcal{U}}, J_{\mathcal{M}}) \label{eq_fidelity_2}\\
     &= \max_{\mathcal{M}\in\mathbb{O}} \mathrm{tr}[(\Pi_0 \rho_L\qty(\theta)\Pi_0) (\Pi_0J_{\mathcal{M}}\Pi_0)] \label{eq_fidelity_3}\\
     &\leq \max_{\sigma\in\mathbb{F}} \mathrm{tr}\qty[\Pi_0 \rho_L\qty(\theta)\Pi_0 \frac{\Pi_0\sigma\Pi_0}{\mathrm{tr}[\Pi_0\sigma\Pi_0]}] \label{eq_fidelity_4}\\
     &\leq \max_{\sigma\in\mathbb{F}} \mathrm{tr}[\rho(\theta) \sigma] \label{eq_fidelity_5}\\
     &= \frac{1+\cos\theta}{2}. \label{eq_fidelity_6}
\end{align}
Here, $\Pi_0 = \frac{1}{2}(II + ZZ)$ is the projector onto the eigenspace of $ZZ$ with eigenvalue $+1$, $\rho_L(\theta) = \frac{1}{2}(II + \cos\theta XX+\sin\theta YX)$, and $\rho(\theta) = \frac{1}{2}(I + \cos\theta X+\sin\theta Y)$.
Eq.~\eqref{eq_fidelity_1} is the definition of the fidelity-based measure $F_{\mathbb{O}}(\mathcal{U})$.
In Eq.~\eqref{eq_fidelity_2}, we set $\rho = \ketbra{\Psi}$.
Eq.~\eqref{eq_fidelity_3} can be obtained from a simple calculation, and we use Eq.~\eqref{eq_CSP_channel} to derive Eq.~\eqref{eq_fidelity_4}.
Since $\Pi_0 \rho_L\qty(\theta)\Pi_0$ and $\frac{\Pi_0\sigma\Pi_0}{\mathrm{tr}[\Pi_0\sigma\Pi_0]}$ are a logical state on the 2-qubit repetition code where a single-qubit state is encoded, we obtain Eq.~\eqref{eq_fidelity_5}.
Finally, Eq.~\eqref{eq_fidelity_6} is obtained because the maximum is achieved for $\sigma= \frac{1}{2}(I+X)$, where we have used $0\leq \theta\leq\pi/4$.

By combining Lemma~\ref{lem_1}, Lemma~\ref{lem_2}, and Eq.~\eqref{eq_fidelity_5}, we obtain the following theorem.
\begin{thm}
    \label{thm_5'}
    (Restatement of Theorem~\ref{thm_5})  
    Let $\mathcal{U}(\cdot) = e^{-i\frac{\theta}{2} Z}\cdot e^{i\frac{\theta}{2} Z}$ be a Pauli-$Z$ rotation gate with $0\leq\theta\leq\pi/4$, and let $\mathcal{U}_{\mathcal{N}} = \mathcal{N}\circ\mathcal{U}$ be a noisy channel affected by single-qubit Pauli noise $\mathcal{N}(\cdot) = p_0 \cdot +p_x X\cdot X + p_y Y\cdot Y + p_z Z\cdot Z$. Then, the worst-case fidelity between the ideal channel $\mathcal{U}$ and the purified noisy channel $\Theta(\mathcal{U}_{\mathcal{N}})$ under a stabilizer-preserving superchannel $\Theta\in\mathbb{S}$ is upper bounded as
    \begin{equation}
        \begin{aligned}
            F(\Theta(\mathcal{U}_{\mathcal{N}}), \mathcal{U})
            \leq \frac{1+\cos\theta}{2}\biggl((&p_0+p_z)R\qty(\theta,\frac{p_z}{p_0+p_z})\\
            + (&p_x+p_y)R\qty(\theta, \frac{p_y}{p_x+p_y})\biggr),
        \end{aligned}
    \end{equation}
    \begin{equation}
        \begin{aligned}
            F(\Theta(\mathcal{U}_{\mathcal{N}}), \mathcal{U})
            \leq 1-\frac{1-\cos\theta}{2}\biggl((&p_0+p_z)W\qty(\theta,\frac{p_z}{p_0+p_z})\\
            + (&p_x+p_y)W\qty(\theta, \frac{p_y}{p_x+p_y})\biggr).
        \end{aligned}
    \end{equation}
    Here, $R(\theta,p)$ and $W(\theta,p)$ are defined in Eq.~\eqref{eq_lem2_2_1} and Eq.~\eqref{eq_lem2_2_2}.
\end{thm}

Notably, for $\theta = \pi/4$ and assuming $\frac{p_z}{p_0+p_z}, \frac{p_y}{p_x+p_y}< \frac{2-\sqrt{2}}{4}$, the two bound simplifies to the same bound of
\begin{equation}
    F(\Theta^{\mathrm{cor}}(\mathcal{U}_{\mathcal{N}}), \mathcal{U}) \leq 1-p_y-p_z,
\end{equation}
which can be saturated by SCV under Pauli symmetry.
While the original bounds are based on two distinct resource measures, which are resource robustness and resource weight, they boil down to the same inequality.
This may suggest an as-yet-unexplored relationship between the two resource measures.
Moreover, while the original bound in Ref.~\cite{regula2021fundamental} (Lemma~\ref{lem_1} in our manuscript) applies to axiomatically defined superchannels $\Theta\in\mathbb{S}$ in Eq.~\eqref{eq_stab_superchanel}—which is much larger than the set of operationally motivated superchannels $\Theta_{\mathrm{Clif}}^{\mathrm{cor}}$ and thus was expected to only give loose bounds for practical settings—the bound can, in fact, be achieved using our operational protocols.
This result not only demonstrates the effectiveness of SCV but also provides theoretical insights into dynamical resource theories by establishing the tightness and practical significance of this general bound.

\bibliography{bib.bib}
\end{document}